\documentclass[12pt,reqno]{amsart}
\usepackage[latin9]{inputenc}
\usepackage{geometry}
\usepackage{units}
\usepackage{enumitem}
\usepackage{amscd}
\usepackage{amsbsy}
\usepackage{amstext}
\usepackage{amsthm}
\usepackage{amssymb}
\usepackage{graphicx}
\usepackage[unicode=true,
 bookmarks=false,
 breaklinks=false,pdfborder={0 0 1},backref=section,colorlinks=false]
 {hyperref}

\makeatletter
\numberwithin{equation}{section}
\numberwithin{figure}{section}
\theoremstyle{plain}
\newtheorem{thm}{\protect\theoremname}[section]
\theoremstyle{definition}
\newtheorem{example}[thm]{\protect\examplename}
\theoremstyle{definition}
\newtheorem{defn}[thm]{\protect\definitionname}
\theoremstyle{remark}
\newtheorem{rem}[thm]{\protect\remarkname}
\theoremstyle{plain}
\newtheorem{cor}[thm]{\protect\corollaryname}
\theoremstyle{plain}
\newtheorem{prop}[thm]{\protect\propositionname}
\theoremstyle{plain}
\newtheorem{lem}[thm]{\protect\lemmaname}
\theoremstyle{remark}
\newtheorem*{rem*}{\protect\remarkname}

\usepackage{amscd}
\usepackage{tikz}

\usetikzlibrary{arrows, automata, positioning, quotes}
\usetikzlibrary{external}
\usetikzlibrary{calc}

\DeclareMathOperator{\spn}{span}

\DeclareMathOperator{\spec}{spec}
\DeclareMathOperator{\diag}{diag}
\DeclareMathOperator{\Image}{Image}

\newcommand{\Z}{\mathbb{Z}}
\newcommand{\R}{\mathbb{R}}
\newcommand{\N}{\mathbb{N}}

\renewcommand{\kappa}{\varkappa}

\newcommand{\V}{\mathcal{V}}
\newcommand{\e}{\mathbf{e}}

\renewcommand{\vec}{\mathrm{vec}}
\newcommand{\Op}{\mathrm{Op}}
\newcommand{\reg}{\mathrm{reg}}

\theoremstyle{plain}

\numberwithin{equation}{section}

\author{Ram Band}
\address{Department of Mathematics, Technion -- Israel Institute of Technology, Haifa 32000, Israel and Institute of Mathematics, University of Potsdam, Potsdam, Germany}
\author{Gregory Berkolaiko}
\address{Department of Mathematics, Texas A\&M University, College
  Station, TX 77843-3368, USA}
\author{Christopher H. Joyner}
\address{School of Mathematical Sciences, Queen Mary University of London, London, E1 4NS, UK}
\author{Wen Liu}
\address{Department of mathematics, Lamar University, Beaumont, TX 77710, USA}

\makeatother

\providecommand{\corollaryname}{Corollary}
\providecommand{\definitionname}{Definition}
\providecommand{\examplename}{Example}
\providecommand{\lemmaname}{Lemma}
\providecommand{\propositionname}{Proposition}
\providecommand{\remarkname}{Remark}
\providecommand{\theoremname}{Theorem}

\begin{document}
\global\long\def\C{\mathbb{C}}%

\global\long\def\N{\mathbb{N}}%

\global\long\def\Id{\mathbb{\boldsymbol{I}}}%

\global\long\def\triv{\mathrm{V_{triv}}}%

\global\long\def\Hom{\mathrm{Hom}_{G}}%

\global\long\def\term#1{\emph{#1}}%

\global\long\def\o#1{\emph{\ensuremath{\overline{#1}}}}%

\global\long\def\Op{\textrm{T}}%

\global\long\def\GL{\textrm{GL}}%

\global\long\def\Tr{\textrm{Tr}}%

\global\long\def\Dom{\textrm{Dom}}%

\global\long\def\Rank{\textrm{Rank}}%

\global\long\def\diag{\textrm{diag}}%

\global\long\def\K{\mathcal{K}_{G}}%

\global\long\def\Ki{\mathcal{K}_{i}}%

\global\long\def\P{\mathbf{P}_{\rho,\pi}}%

\global\long\def\Pb{\mathbf{P}}%

\global\long\def\PPi{\mathbf{P}_{i}}%

\global\long\def\Pt{\mathcal{P}}%

\global\long\def\ei{\mathbf{e}_{i}}%

\global\long\def\ej{\mathbf{e}_{j}}%

\global\long\def\ones{\mathbf{1}}%

\global\long\def\Vr{\textrm{V}_{\rho}}%

\global\long\def\Vp{\textrm{V}_{\pi}}%

\global\long\def\Vs{\textrm{V}_{\sigma}}%

\global\long\def\Cp{\C^{p}}%

\global\long\def\sgn{\textrm{sgn}}%

\global\long\def\reg{\textrm{reg}}%

\global\long\def\rev{\mathfrak{i}}%

\global\long\def\tr{\textrm{Trace}}%

\global\long\def\vec{\textrm{vec}}%

\global\long\def\vect{\textrm{vec}_{T}}%

\global\long\def\unvec{\textrm{unvec}}%

\global\long\def\spn{\textrm{span}}%

\global\long\def\Image{\textrm{Image}}%

\global\long\def\e{\mathbf{e}}%

\global\long\def\Pt{\mathcal{P}}%

\global\long\def\vect{\mathrm{vec}}%

\global\long\def\cc#1{\emph{\ensuremath{\overline{#1}}}}%

\global\long\def\set#1#2{\left\{  #1~:~#2\right\}  }%

\global\long\def\Espace#1#2{\mathbf{E}_{#2}^{#1}}%

\global\long\def\Erep#1#2{\epsilon_{#2}^{#1}}%

\global\long\def\Ind#1#2#3{\textrm{\ensuremath{\mathrm{Ind}}}_{#1}^{#2}\left(#3\right)}%

\global\long\def\Res#1#2#3{\mathrm{Res}_{#1}^{#2}\left(#3\right)}%

\global\long\def\tA{\widetilde{A}}%

\global\long\def\tB{\widetilde{B}}%

\global\long\def\rmi{\mathrm{i}}%

\global\long\def\K{\mathcal{K}_{G}}%

\global\long\def\D{\mathcal{D}}%

\global\long\def\E{\mathcal{E}}%

\global\long\def\EE{\overleftrightarrow{\E}}%

\global\long\def\V{\mathcal{V}}%

\global\long\def\ue{\mathrm{e}}%

\global\long\def\ui{\mathrm{i}}%

\global\long\def\uj{\mathrm{j}}%

\global\long\def\uk{\mathrm{k}}%

\global\long\def\ii{\mathfrak{i}}%

\global\long\def\jj{\mathfrak{j}}%

\global\long\def\kk{\mathfrak{k}}%

\global\long\def\ud{\mathrm{d}}%

\global\long\def\spec{\mathrm{Spec}}%

\title{Quotients of graph operators by symmetry representations}
\begin{abstract}
A finite dimensional operator that commutes with some symmetry group
admits quotient operators, which are determined by the choice of associated
representation. Taking the quotient isolates the part of the spectrum
supporting the chosen representation and reduces the complexity of
the spectral problem. Yet, such a quotient operator is not uniquely
defined. Here we present a computationally simple way of choosing
a special basis for the space of intertwiners, allowing us to construct
a quotient that reflects the structure of the original operator. This
quotient construction generalizes previous definitions for discrete
graphs, which either dealt with restricted group actions or only with
the trivial representation.

We also extend the method to quantum graphs, which simplifies previous
constructions within this context, answers an open question regarding
self-adjointness and offers alternative viewpoints in terms of a scattering
approach. Applications to isospectrality are discussed, together with
numerous examples and comparisons with previous results.
\end{abstract}

\maketitle

\section{Introduction}

The analysis of linear operators possessing a certain symmetry dates
backs to the seminal works of Schur and Weyl \cite{Schur-1927,Weyl-1946}.
More specifically, given a linear representation $\pi$ of a group
$G$, they studied the properties of the algebra of matrices $M$
satisfying $\pi(g)M=M\pi(g)$ for all $g\in G$. Their answer ---
that by a suitable change of basis $M$ can be decomposed into blocks
associated with the irreducible representations of $\pi$ --- is
a cornerstone of modern representation theory.


For a given particular system, a quantum Hamiltonian for example,
representation theory is often used to classify the states and thus
deduce the properties of the full system. In order to isolate the
states of a particular representation one must `desymmetrize' in some
appropriate manner, or, equivalently, find the corresponding block
of the reduced linear operator. This philosophy underlies, for instance,
the Floquet--Bloch--Gelfand transform of periodic operators \cite[Sec.~XIII.16]{ReedSimon_v4},
\cite{Kuc_bams16}, where the symmetry group $\Z^{d}$ is abelian
and thus has particularly simple irreducible representations. With
regards to discrete symmetries, this philosophy has been utilized,
for example, to construct trace formulae \cite{Robbins-1989,Lauritzen-1991}
and analyze the spectral statistics \cite{Keating-1997,Joyner-2012}
of quantum chaotic systems with symmetries. A more systematic construction
procedure applicable to unbounded self-adjoint operators on metric
graphs and manifolds was demonstrated by Band, Parzanchevski and Ben-Shach
\cite{Ban09,Par10}, who showed how one may perform this desymmetrization
for arbitrary representations of a discrete group of symmetries. The
resulting desymmetrized system was termed a `quotient graph' or `quotient
manifold'. In particular, this notion generalized Sunada's method
\cite{Sun85} for constructing pairs of isospectral manifolds, famously
used by Gordon, Webb and Wolpert \cite{Gor92b,Gor92a} to answer negatively
Marc Kac's question `Can one hear the shape of a drum?' \cite{Kac66}.
Joyner, Müller and Sieber \cite{Joy14} utilized the quotient graph
technique to propose a spin-less quantum system whose spectrum exhibits
the Gaussian symplectic ensemble statistics from random matrix theory\footnote{See also Example \ref{Ex: The quaternion     group}.};
this effect was later demonstrated experimentally \cite{Rehemanjiang-2016}.
The quotient construction has also been used to prove the existence
of conical points in the dispersion relation of graphene by Berkolaiko
and Comech \cite{BerCom_jst18}.

The main result of this work is a compact and transparent formula
(Definition~\ref{def: matrix_form_of_quotient}) that extracts the
`quotient of an operator by a representation of its symmetry group',
or \emph{quotient operator} for short, in the context of finite dimensional
operators. While the topic is a classical one, the formula appears
to be new; it generalizes previous works by Brooks \cite{Brooks_aif99,Brooks_cm99},
Chung and Sternberg \cite{Chu92}, and Halbeisen and Hungerbühler
\cite{Hal99} that either dealt with restricted group actions or only
with the trivial representation.

We chose to state our main result, equation~\eqref{eq:Quotient_formula_explicit},
as a definition and then prove, in Sections~\ref{sec: fundamental property of   quotient}-\ref{SubSec: Proof of fundamental property},
that it has the properties one expects from a desymmetrized part of
an operator that corresponds to a particular representation. While
formulated for general matrices, our results are geared towards spectral
analysis of graphs in that the connectivity structure of the original
operator is reflected in the quotient operator. Naturally, this is
the effect of a particular choice of basis we make.

To illustrate the types of results one obtains by using the quotient
construction, we present some examples.
\begin{example}
\label{ex:tetra_intro} 
\begin{figure}
\centering \includegraphics{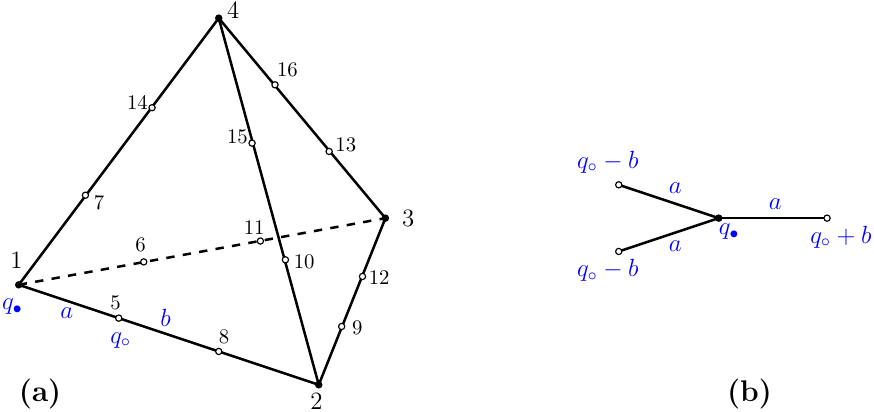} \caption{(a) A graph with symmetry group $S_{4}$ acting as permutations of
the vertices $1,2,3,4$ and the corresponding permutation of other
vertices. The graph is associated with the operator $\protect\Op$
in (\ref{eq:tetraH}).\protect \\
(b) A graph which is associated to the quotient operator $\protect\Op_{\sigma}$
given by (\ref{eq:Hsigma_intro}). This quotient operator captures
all the $\protect\Op$ eigenvalues of multiplicity three.}
\label{fig:discrete_tetra}
\end{figure}

Consider the graph on 16 vertices shown in Figure \ref{fig:discrete_tetra}(a)
and a corresponding family of operators $\Op:\C^{16}\to\C^{16}$ that
are invariant under the action of the symmetry group $G=S_{4}$ induced
by permuting the vertices $1$, $2$, $3$ and $4$. Explicitly, grouping
the vertices into two sets, $O_{\bullet}=\{1,\ldots,4$\} and $O_{\circ}=\{5,\ldots,16\}$,
we let 
\begin{equation}
\Op_{ij}=\begin{cases}
a & \mbox{if }i\sim j,\ \ i\in O_{\bullet},j\in O_{\circ}\mbox{ or }i\in O_{\circ},j\in O_{\bullet}\\
b & \mbox{if }i\sim j,\ \ i\in O_{\circ},j\in O_{\circ}\\
q_{\bullet} & \mbox{if }i=j\in O_{\bullet}\\
q_{\circ} & \mbox{if }i=j\in O_{\circ},
\end{cases}\label{eq:tetraH}
\end{equation}
where $\sim$ indicates adjacency of the vertices.

For \emph{every} choice\footnote{Some parameter choices result in a non-Hermitian $\Op$; this does
not affect the validity of our statements.} of the parameters $a,b,q_{\bullet},q_{\circ}\in\C$, there are at
least four triply degenerate eigenvalues in the spectrum of $\Op$
(typically, all other eigenvalues have lower multiplicity). To give
a numerical example, the choice $a=b=-1$, $q_{\bullet}=3$, $q_{\circ}=2$
in \eqref{eq:tetraH} yields the spectrum 
\begin{equation}
\spec(\Op)=\{0,\ \mathbf{0.44},\ \mathbf{0.44},\ \mathbf{0.44},\ 1,\ 1,\ \mathbf{2},\ \mathbf{2},\ \mathbf{2},\ \mathbf{3},\ \mathbf{3},\ \mathbf{3},\ 4,\ \mathbf{4.56},\ \mathbf{4.56},\ \mathbf{4.56}\}.\label{eq:tetra_spec_intro}
\end{equation}
It is well known that the eigenvalue multiplicities arise due to the
symmetry of the operator. More precisely, they can be traced to the
3-dimensional irreducible representations of the group $S_{4}$ by
studying the action of $\Op$ on the space of intertwiners. The current
work offers a simple prescription for extracting these eigenvalues
\emph{in the form of the spectrum of a smaller graph}, shown in this
particular case in Figure \ref{fig:discrete_tetra}(b) and given explicitly
by 
\begin{equation}
\Op_{\sigma}=\begin{pmatrix}q_{\bullet} & a & a & a\\
a & q_{\circ}-b & 0 & 0\\
a & 0 & q_{\circ}-b & 0\\
a & 0 & 0 & q_{\circ}+b
\end{pmatrix}.\label{eq:Hsigma_intro}
\end{equation}

To compare to \eqref{eq:tetra_spec_intro}, substituting $a=b=-1$,
$q_{\bullet}=3$, $q_{\circ}=2$, we get 
\begin{equation}
\spec(\Op_{\sigma})=\{0.44,\ 2,\ 3,\ 4.56\}.\label{eq:tetra_spec_quot_intro}
\end{equation}
The details of the derivation of \eqref{eq:Hsigma_intro} are contained
in Example~\ref{Ex: Reducing points} and \ref{ex:tetrahedrom_all3}.
\end{example}

While the main body of the paper deals with finite-dimensional operators,
in Section~\ref{sec:QuantumGraphs} we show how to use it to compute
quotients of unbounded self-adjoint operators on metric (``quantum'')
graphs, where one has to be careful to take an appropriate quotient
of the operator's domain. Thus we are simplifying and making more
explicit the constuction of \cite{Ban09,Par10}, incidentally proving
that it preserves self-adjointness. Delaying all details to Section~\ref{sec:QuantumGraphs},
we present here an example of a non-trivial quotient.


\begin{example}
\label{ex:tetra_qg_intro}

\begin{figure}
\centering \includegraphics{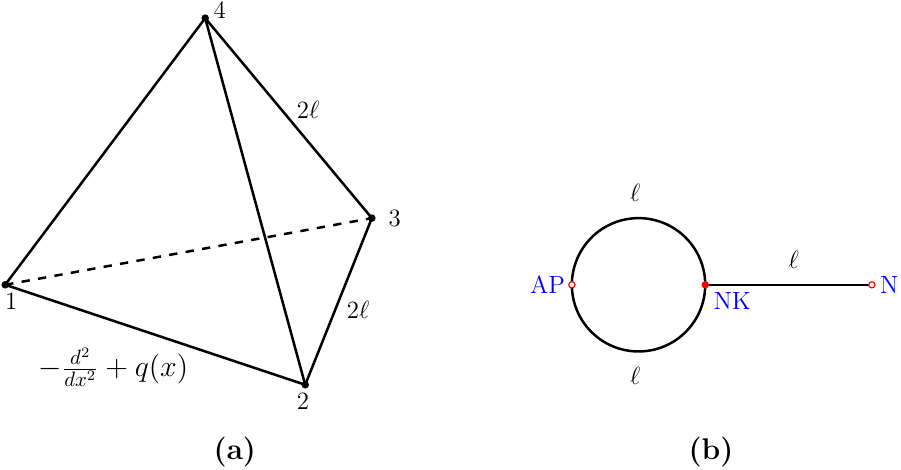} \caption{A quantum graph with symmetry group $S_{4}$ and its quotient that
captures the eigenvalues of multiplicity three.}
\label{fig:discrete_tetra_qg}
\end{figure}

Consider a quantum graph analogue of Example~\ref{ex:tetra_intro},
namely the graph in Fig.~\ref{fig:discrete_tetra_qg}(a). We defer
the rigorous description of the corresponding differential operator
$\Gamma$ to Section~\ref{sec:QuantumGraphs}; here we just mention
that all edges have the same length $2\ell$, all vertices have Neumann-Kirchhoff
(NK or ``standard'') conditions and the potential is assumed to
be symmetric with respect to reflection around the midpoint, $q(2\ell-x)=q(x)$.

The symmetry group here is $S_{4}$, so the presence of eigenvalues
of multiplicity three is not surprising. What is remarkable is that
those eigenvalues are exactly the eigenvalues of the graph shown in
Fig.~\ref{fig:discrete_tetra_qg}(b), where the central vertex has
NK conditions, right vertex has Neumann and the left vertex has anti-periodic
conditions\footnote{Alternatively, we can remove the left vertex and add a magnetic potential
with flux $\pi$ through the loop of the graph.}. As above, the main result of this work is \emph{representing the
action} of the original graph operator on an appropriate invariant
subspace \emph{as a concrete operator on a smaller graph}.
\end{example}

The paper is structured as follows. Our main results are presented
in Section~\ref{sec:quotient_operators}, where we define the quotient
by a compact formula in Definition~\ref{def: matrix_form_of_quotient}
and then, in Theorem~\ref{thm: Quotient fundamental property} and
Propositions~\ref{prop: Algebraic properties} and \ref{Prop:Spectral_property},
establish its various algebraic and spectral properties. Section~\ref{sec:Examples}
provides several examples that illustrate the group-theoretic constructions
used in the proofs of Section~\ref{sec:quotient_operators}. In Section~\ref{sec:QuantumGraphs}
we adapt the quotient operator construction (back) to quantum graphs
and show how this broadens the previous results in \cite{Ban09,Par10}.
Section \ref{Sec: Applications of quot ops} outlines some further
applications of the current results and explains how they extend previous
works on isospectrality \cite{Brooks_aif99,Brooks_cm99,Hal99} and
eigenvalue computation \cite{CveDooSac_spectra_of_graphs_book,Chu92,Chung_spectralgraph}.

At the end, Appendix~\ref{sec: appendix-representation_theory} contains
the necessary background material on representation theory, in particular
on Frobenius isomorphism. Due to our motivation coming from self-adjoint
operators, our emphasis is on the preservation of the Hilbert structure,
which is not usually covered in algebra-themed textbooks. Appendix~\ref{app:qg_quotient}
contains the proofs of the results from Section~\ref{sec:QuantumGraphs}.


\section{Quotient Operators - definitions and main properties}

\label{sec:quotient_operators}

This paper (and in particular the current section) is substantially
based on elements from representation theory. For completness, we
describe the relevant definitions and statements in Appendix \ref{sec: appendix-representation_theory}.


\subsection{Defining the quotient operator\label{sec:def_quotient_operators}}

This subsection defines the main object of this paper --- the quotient
operator. Examples illustrating the construction of the quotient operator
and its properties are collected in Section~\ref{sec:Examples}.
A reader less familiar with the topic is encouraged to consult these
examples while reading this section.
\begin{defn}[symmetric operator]
\label{def:sym_operator} Let $G$ be a finite group and $\Vp$ be
a vector space over $\C$ which is a $G$-module. We say that a linear
operator $\Op:\Vp\to\Vp$ is $G$-symmetric (or that $\Op$ commutes
with $G$) if 
\begin{equation}
\forall\ g\in G,\quad g\,\Op=\Op\,g.\label{eq:symmetric_operator}
\end{equation}
\end{defn}

That $\Vp$ is a $G$-module means that the group $G$ acts on $\Vp$
by linear transformations, i.e. it is a representation of $G$. In
the following, it will be convenient to use the terminology of $G$-modules
and group representations, interchangeably. The reader is referred
to Appendix \ref{sec: appendix-representation_theory} for the corresponding
background and for further details.

We will use the symmetry of the operator $\Op$ to construct a new
operator, which we call the quotient operator. We will restrict ourselves
to the case when the $G$-module $\Vp$ is finite dimensional\footnote{An extension to infinite dimensional spaces $\Vp$ and bounded operators
$T$ should not be difficult; more difficulties arise if $T$ is unbounded.
These difficulties are handled in Section~\ref{sec:QuantumGraphs},
where the quotient construction is described for differential operators
on metric graphs.} and corresponds to some permutation representation $\pi:G\rightarrow GL(\Vp)$.
Namely, we assume that all $g\in G$ act as permutation matrices,
for a suitably chosen basis of $\Vp$. We will refer to this basis
as the standard basis, $\{\ei\}_{i\in\Pt}$, where the corresponding
index set is $\Pt:=\{1,\ldots,p\}$, with $p:=\dim\Vp$. Hence, $\Vp\cong\C[\Pt]$.
In what follows, we use both $\Vp$ and $\C[\Pt]$ interchangeably,
depending on whether we wish to emphasize that it is a $G$-module,
or rather that these are functions on the point set $\Pt$.

The permutation action of $G$ on the standard basis $\{\ei\}$ induces
an action of $G$ on the index set $\Pt$, namely $g\thinspace\ei=\mathbf{e}_{g\thinspace i}$.
We use this action to decompose $\Pt$ into orbits, 
\begin{equation}
O_{i}:=\{j\in\Pt~:~j=g\thinspace i~\textrm{for some }g\in G\},\label{eq:orbit_definition}
\end{equation}
and choose a representative of each orbit. Without loss of generality,
we denote those representatives by $\D=\{1,\ldots,\left|\D\right|\}$
and call $\D$ a \emph{fundamental domain}.

Let $\Vr$ be another finite-dimensional $G$-module and denote $r:=\dim\Vr$.
We assume that $\Vr$ is equipped with an inner product (hence, it
is a Hilbert space) and that the $G$ action preserves this inner
product (equivalently, that the representation $\rho$ is unitary).

We start connecting $\Vr$ with $\Vp$, towards defining the quotient
operator (of $\Op$ with respect to $\Vr$). For each index $i\in\D$
consider the stabilizer group 
\begin{equation}
G_{i}:=\{g\in G:g\thinspace i=i\},\label{eq:def_stabilizer}
\end{equation}
and define the $G_{i}$-invariant subspace of $\Vr$ by 
\begin{equation}
(\Vr)^{G_{i}}:=\set{v\in\Vr}{\forall g\in G_{i},~gv=v},\label{eq:def_Gi_invariant_subspace}
\end{equation}
For each $i\in\D$, denote $d_{i}:=\dim(\Vr)^{G_{i}}$ and pick $\{\varphi_{i}^{(n)}\}_{n=1}^{d_{i}}$
to be an orthonormal basis for $(\Vr)^{G_{i}}$, with respect to the
inner product inherited from $\Vr$. For use in the next definition,
we collect all vectors $\{\varphi_{i}^{(n)}\}_{n=1}^{d_{i}}$ as the
columns of a matrix, which we denote by $\Phi_{i}$. Namely, $\Phi_{i}$
is an $r\times d_{i}$ matrix providing an isometry between $\C^{d_{i}}$
and $(\Vr)^{G_{i}}$. We further denote $d:=\sum_{i\in\D}d_{i}$.
\begin{defn}[Quotient operator]
\label{def: matrix_form_of_quotient} Using the notation introduced
above, we define $\Op_{\rho}$ to be a $d\times d$ matrix, which
is comprised of $\left|\D\right|\times\left|\D\right|$ blocks. The
$(i,j)$-th block is is the following $d_{i}\times d_{j}$ matrix
\begin{align}
\left[\Op_{\rho}\right]_{i,j} & :=\frac{1}{\sqrt{|G_{i}||G_{j}|}}\Phi_{i}^{*}\,\left(\sum_{g\in G}\Op_{i,gj}\,g\right)\,\Phi_{j}.\label{eq:Quotient_formula_explicit}
\end{align}
$\Op_{\rho}$ is called the quotient operator of $\Op$ with respect
to $\Vr$.
\end{defn}

\begin{rem}
The notation of quotient operator includes only the auxiliary $G$-module,
$\Vr$, with respect to which the quotient is constructed. The $G$-module
$\Vp$ is understood from the context, as it is the domain of the
original operator, $T$. The term ``quotient operator'' arises from
the specific case when $\Vr$ is the trivial representation of $G$
(see Corollary \ref{cor:simplifications_of_quotient_formula}). In
this case the domain of $\Op_{\rho}$ is the quotient $\nicefrac{\Vp}{G}$
(in the topological sense). Hence, we adopt the name quotient also
for the more general case (as was already done in the previous works
\cite{Ban09,Par10}).
\end{rem}

Before presenting the properties of the quotient operator, we mention
a few general cases in which formula \eqref{eq:Quotient_formula_explicit}
reduces to a simpler form. These are useful in various applications,
some of which are surveyed in Section \ref{Sec: Applications of quot ops}.
In particular, when $\Op$ is a discrete graph operator, the special
case of the quotient with respect to the trivial representation, (\ref{eq:quotient_formula_trivial_rep}),
is a well-known concept. It is known by various names such as a graph
divisor, an equitable partition, a coloration, a quotient graph or
an orbigraph (more details on these special cases are given in Section
\ref{Subsec: Applications - Computations}).
\begin{cor}[Simplifications and variations of the formula]
\label{cor:simplifications_of_quotient_formula} ~
\begin{enumerate}
\item \textbf{Trivial representation:} If $\Vr$ is the trivial $G$-module
then $\Op_{\rho}$ is a $\left|\D\right|\times\left|\D\right|$ matrix
whose entries are 
\begin{equation}
[\Op_{\rho}]_{i,j}=\frac{1}{\sqrt{|G_{i}||G_{j}|}}\sum_{g\in G}\Op_{i,gj}.\label{eq:quotient_formula_trivial_rep}
\end{equation}
\item \textbf{Free action:} If the action of $G$ on $\Pt$ is free then
for all $i,j\in\D$, the $(i,j)$-th block of $\Op_{\rho}$ is given
by 
\begin{equation}
[\Op_{\rho}]_{i,j}=\sum_{g\in G}\Op_{i,gj}\,g.\label{eq:quotient_formula_free}
\end{equation}
\item \textbf{Double coset expansion:} $\left[\Op_{\rho}\right]_{i,j}$
can also be computed as 
\begin{equation}
\left[\Op_{\rho}\right]_{i,j}=\sqrt{\left|G_{i}\right|\left|G_{j}\right|}\Phi_{i}^{*}\left(\sum_{G_{i}gG_{j}\in G_{i}\backslash G/G_{j}}\frac{1}{\left|G_{i}g\cap gG_{j}\right|}\Op_{i,gj}\,g\right)\Phi_{j},\label{eq:Quotient_formula_double_cosets}
\end{equation}
where the sum is over representatives of the double cosets $G_{i}\backslash G/G_{j}$.
\end{enumerate}
\end{cor}

\begin{proof}[Proof of Corollary \ref{cor:simplifications_of_quotient_formula}]

To get (\ref{eq:quotient_formula_trivial_rep}), note that $\Phi_{i}=(1)$
for all $i\in\D$. Substituting this into (\ref{eq:Quotient_formula_explicit}),
we obatin that the $(i,j)$ block is just a single matrix entry given
by (\ref{eq:quotient_formula_trivial_rep}).

To get (\ref{eq:quotient_formula_free}) note that the free action,
$G_{i}=\{1\}$, implies $(\Vr)^{G_{i}}\cong\Vr$. Taking dimensions
gives\footnote{One may also have $d_{i}=r$ when the action is not free, although
free action is required to achieve the form (\ref{eq:quotient_formula_free}).} $d_{i}=r=\deg(\rho)$. Therefore, we can choose $\Phi_{i}$ to be
any $r\times r$ unitary matrix, including $\Phi_{i}=\Id_{r}$. This
yields (\ref{eq:quotient_formula_free}).

To get (\ref{eq:Quotient_formula_double_cosets}) we only need to
employ the decomposition of $G$ to doubles cosets $G_{i}\backslash G/G_{j}$
and note that the expression $\left(\Phi_{i}^{*}\,\left[\Op\right]_{i,gj}\rho(g)\,\Phi_{j}\right)$
is invariant with respect to the choice of particular representative
of the double coset and that $\left|G_{i}gG_{j}\right|=\frac{\left|G_{i}\right|\left|G_{j}\right|}{\left|G_{i}g\cap gG_{j}\right|}$.
Hence, in (\ref{eq:Quotient_formula_explicit}) we may only sum over
representatives of the double cosets $G_{i}\backslash G/G_{j}$ to
get (\ref{eq:Quotient_formula_double_cosets}).
\end{proof}

\subsection{Fundamental property of the quotient operator\label{sec: fundamental property of quotient}}

We present here a fundamental property which the quotient operator
defined by equation~\eqref{eq:Quotient_formula_explicit} satisfies
and from which follow other algebraic and spectral properties of the
quotient. To do so, we need to introduce a few more algebraic notions.
Let $G$ be a finite group and $\Vr$ and $\Vp$ be two finite-dimensional
$G$-modules. As before we assume that $\Vr$ and $\Vp$ are, equipped
with inner products, which are preserved by the $G$ action. We consider
all linear maps from $\Vr$ to $\Vp$ which commute with the group
action, and denote this set by 
\begin{equation}
\Hom(\Vr,\Vp):=\set{\phi\in\mathrm{Hom}(\Vr,\Vp)}{g\phi=\phi g}.\label{eq:hom_space_definition}
\end{equation}
This is a vector space whose elements are commonly called \emph{intertwiners}
(see more in Appendix \ref{sec: appendix-representation_theory}).
The natural inner product on $\Hom$ is the Frobenius product 
\begin{equation}
\left<\phi_{1},\phi_{2}\right>_{\Hom(\Vr,\Vp)}:=\Tr\left(\phi_{2}^{*}\phi_{1}\right)=\sum_{v\in B(\Vr)}\left<\phi_{1}(v),\phi_{2}(v)\right>_{\Vp},\label{eq:Frobenius_inner_product}
\end{equation}
where $\left<\cdot,\cdot\right>_{\Vp}$ denotes the inner product
in $\Vp$ and $B(\Vr)$ is a basis for $\Vr$. Let $\Op:\Vp\rightarrow\Vp$
be a linear operator which is $G$-symmetric. Such an operator acts
on each intertwiner $\phi\in\Hom(\Vr,\Vp)$ by composition, 
\begin{equation}
\forall u\in\Vr,\quad\quad\left(\Op\phi\right)(u):=\Op\thinspace\phi(u).\label{eq:T_acts_on_Hom_part1}
\end{equation}
Clearly, $\Op\phi\in\mathrm{Hom}(\Vr,\Vp)$. But we even have $\Op\phi\in\Hom(\Vr,\Vp)$
since 
\begin{equation}
g(\Op\phi)=\Op g\phi=(\Op\phi)g,\label{eq:T_acts_on_Hom_part2}
\end{equation}
because $\Op$ is $G$-symmetric\footnote{We may also consider $\Op$ itself as an intertwiner, $\Op\in\Hom(\Vp,\Vp)$.
In view of this, the action of $\Op$ on $\Hom(\Vr,\Vp)$ is a composition
of intertwiners.} and $\phi\in\Hom(\Vr,\Vp)$ commutes with $g$.

This linear action of $\Op$ on $\Hom(\Vr,\Vp)$ leads to our main
theorem.
\begin{thm}[Fundamental property of the quotient operator]
\label{thm: Quotient fundamental property} Let $G$ be a finite
group and $\Vr,\Vp$ two finite-dimensional $G$-modules, equipped
with inner products, which are preserved by the group action. Let
$\Op:\Vp\rightarrow\Vp$ be a linear operator which is $G$-symmetric.

Then the quotient operator of $\Op$ with respect to $\Vr$ is unitarily
equivalent to $\left.\Op\right|_{\Hom(\Vr,\Vp)},$ 
\begin{equation}
\Op_{\rho}\cong\left.\Op\right|_{\Hom(\Vr,\Vp)}.\label{eq: Quotient fundamental property}
\end{equation}
Namely, there exists an ortonormal basis of $\Hom(\Vr,\Vp)$ with
respect to which the matrix form of $\left.\Op\right|_{\Hom(\Vr,\Vp)}$
is $\Op_{\rho}$.
\end{thm}

We postpone the proof of the theorem to the end of Section \ref{sec:quotient_operators}
and first present the various properties which follow from the theorem.

\subsection{Algebraic and spectral properties of the quotient}

We start by summarizing some of the algebraic properties of the quotient
operator. These are general properties arising from basic representation
theory, and have useful applications. In particular, Proposition ~\ref{prop: Algebraic properties}\,(\ref{enu: prop_algebraic_subgroup_quotient})
is a vital component in the construction of isospectral objects \cite{Ban09,Par10}
(see Section \ref{Subsec: Applications - Isospectrality}) and Proposition~\ref{prop: Algebraic properties}\,(\ref{enu: prop_algebraic_Operator decomposition})
is the well-known decomposition of a symmetric operator, following
from Schur's Lemma (see Section \ref{Subsec: Applications - Computations}
and \cite[sec. 5]{CveDooSac_spectra_of_graphs_book}).

Recall that we assume $G$ is a finite group and that the $G$-modules
$\Vp$ and $\Vr$ are finite-dimensional and with inner products preserved
by the group action. We denote $\mathrm{V}_{\rho_{1}}\cong\mathrm{V}_{\rho_{2}}$
when the $G$-modules $\mathrm{V}_{\rho_{1}}$ and $\mathrm{V}_{\rho_{2}}$
are isomorphic as Hilbert spaces (i.e., unitarily isomorphic) and
this isomorphism commutes with the $G$-action. We use a similar notation
$\Op_{1}\cong\Op_{2}$ to denote that the operators $\Op_{1}$ and
$\Op_{2}$ are unitarily equivalent.
\begin{prop}
\emph{{[}Algebraic properties{]}} \label{prop: Algebraic properties}
~

Let $G$ be a finite group. Let $\Op$ be a $G$-symmetric linear
operator.
\begin{enumerate}
\item \label{enu: prop_algebraic_decomp_2parts} Let $\mathrm{V}_{R}\cong\bigoplus_{\rho}\Vr$
be a $G$-module which is isomorphic to the (orthogonal) direct sum
of some (possibly repeating and not necessarily irreducible) $G$-modules
$\Vr$, then 
\begin{equation}
\Op_{R}\cong\bigoplus_{\rho}\Op_{\rho}.\label{eq:quotient_linear_decomposition}
\end{equation}
\item \label{enu: prop_algebraic_Operator decomposition} Let $\reg_{G}$
denote the regular representation of $G$. Then 
\begin{equation}
\Op\cong\Op_{\reg_{G}}\cong\bigoplus_{\rho}\left[\Id_{\deg\rho}\otimes\Op_{\rho}\right],\label{eq:quotient_of_regular_rep}
\end{equation}
where the direct sum above is over all irreducible representations
$\rho$ of $G$, and $\Id_{\deg\rho}$ is the identity operator of
rank $\deg\rho$.\\
\item \label{enu: prop_algebraic_subgroup_quotient} Let $H$ be a subgroup
of $G$. If $\Vs$ is an $H$-module and $\Vr=\Ind HG{\Vs}$ is the
corresponding induced $G$-module from $H$ to $G$, then 
\begin{equation}
\Op_{\sigma}\cong\Op_{\rho}.\label{eq:quotient_of_induced_rep}
\end{equation}
\end{enumerate}
\end{prop}

\begin{rem}
Appendix \ref{sec: appendix-representation_theory} contains a short
introduction on induced modules and representations.
\end{rem}

\begin{proof}
To prove~\eqref{eq:quotient_linear_decomposition}, we note that
$\Hom(\mathrm{V}_{R},\Vp)$ is isomorphic to $\bigoplus_{\rho}\Hom(\Vr,\Vp)$
as vector spaces (see, e.g.\ \cite[ch.~10]{DummitFoote_abstract_algebra}).
In addition, it is not hard to verify that this isomorphism preserves
the inner product. Hence, (\ref{eq:quotient_linear_decomposition})
follows straightforwardly from the orthogonal decomposition $\Hom(\mathrm{V}_{R},\Vp)\cong\bigoplus_{\rho}\Hom(\Vr,\Vp)$
together with Definition \ref{thm: Quotient fundamental property}.

To prove~\eqref{eq:quotient_of_regular_rep}, first note that the
second equivalence follows from \eqref{eq:quotient_linear_decomposition}
together with the decomposition of the regular representation in terms
of irreducible representations (see e.g., \cite[ch.~2.4]{Serre_linear_representations}).
To see that $\Op\cong\Op_{\reg_{G}}$, we use that $\reg_{G}=\Ind HG{\Vs}$,
where $H=\{\mathrm{id}\}$ is the trivial subgroup of $G$ and $\Vs=\C$
is the trivial representation of $H$ (see Appendix~\ref{sec: appendix-representation_theory}).
We may now apply (\ref{eq:quotient_of_induced_rep}) and get that
its right hand side is $\Op_{\reg_{G}}$. It is left only to explain
why the left hand side in (\ref{eq:quotient_of_induced_rep}) is just
$\Op_{\sigma}=\Op$. To see this, it is enough to note that $\mathrm{Hom}_{H}(\Vs,\Vp)\cong\mathrm{Hom}(\Vs,\Vp)\cong\Vp$,
keeping in mind that $H=\{\mathrm{id}\}$ and $\Vs=\C$.

Finally, we prove (\ref{eq:quotient_of_induced_rep}). Frobenius reciprocity
theorem (namely, Theorem \ref{thm:Frobenius_reciprocity}\,(\ref{enu:     thm-Frobenius-reciprocity-2})
together with Proposition \ref{prop:     Frobenius_preserves_inner_product})
gives the Hilbert space isomorphism 
\begin{equation}
\mathrm{Hom}_{H}(\Vs,\Vp)\cong\mathrm{Hom}_{G}(\Ind HG{\Vs},\Vp)=\mathrm{Hom}_{H}(\Vr,\Vp).\label{eq:FRT_Hilbert}
\end{equation}
By Proposition~\ref{prop:Frobenius_commutes}, this isomorphism commutes
with $T$ and therefore establishes the unitary equivalence 
\begin{equation}
T\big|_{\mathrm{Hom}_{H}(\Vs,\Vp)}\simeq T\big|_{\mathrm{Hom}_{G}(\Vr,\Vp)}.\label{eq:quotients_ue}
\end{equation}
The conclusion follows from \eqref{eq: Quotient fundamental property}.
\end{proof}
\vspace{10pt}

In many instances, the main purpose for constructing the quotient
operator $\Op_{\rho}$, is to isolate the spectral properties of $\Op$
associated to a particular representation $\rho$. To state this formally,
let us introduce the generalized $\lambda$-eigenspace of order $k$,
\begin{equation}
\Espace{\Op}k(\lambda):=\ker\left((\Op-\lambda)^{k}\right).\label{eq:eigenspace-2}
\end{equation}
In particular $\Espace{\Op}1(\lambda)$ is the usual eigenspace corresponding
to eigenvalue $\lambda$.
\begin{prop}[Spectral properties]
\label{Prop:Spectral_property} ~
\begin{enumerate}
\item \label{enu:Prop_Spectral_property_1} Let $\lambda\in\C$ and $k\in\N$.
Then $\Espace{\Op}k(\lambda)$ is a $G$-module and 
\begin{equation}
\Espace{\Op_{\rho}}k(\lambda)\cong\Hom(\Vr,\Espace{\Op}k(\lambda)).\label{eq:spectral_property_of_quotient}
\end{equation}
\item \label{enu:Prop_Spectral_property_2} The eigenspace of $T$ is decomposed
over the \emph{irreducible} representations $\rho$ of $G$ as follows,
\begin{equation}
\Espace{\Op}k(\lambda)\cong\bigoplus_{\rho}\Vr\otimes\Espace{\Op_{\rho}}k(\lambda).\label{eq:eigenspace_irrep}
\end{equation}
In particular, denoting by $m^{\Op}(\lambda)$ either geometric or
the algebraic multiplicity of $\lambda\in\C$ as an eigenvalue of
$\Op$, we get 
\begin{equation}
m^{\Op}(\lambda)=\sum_{\rho}\deg\rho\cdot m^{\Op_{\rho}}(\lambda),\label{eq:eigenvalues_multiplicities_by_irreps}
\end{equation}
where the sum above is over all irreducible representations $\rho$
of $G$.
\item \label{enu:Prop_Spectral_property_3} If $\Op$ is self-adjoint then
$\Op_{\rho}$ is self-adjoint.
\end{enumerate}
\end{prop}

\begin{proof}
For the first part of the proposition observe first that $\Espace{\Op}k(\lambda)\subset\Vp$
already comes equipped with an action of $G$ and, furthermore, it
is invariant under this action: for all $f\in\Espace{\Op}k(\lambda)$
and $g\in G$ we have 
\[
\left(\Op-\lambda\Id\right)^{k}\left(gf\right)=g\left(\Op-\lambda\Id\right)^{k}f=0.
\]
Therefore, $\Espace{\Op}k(\lambda)$ is a $G$-module.

Further note that $\Hom(\Vr,\Espace{\Op}k(\lambda))$ is a subspace
of $\Hom(\Vr,\Vp)$. Therefore, by \eqref{eq: Quotient   fundamental property},
\begin{multline}
\Espace{\Op_{\rho}}k(\lambda)\cong\Espace{\Op|_{\Hom(\Vr,\Vp)}}k(\lambda)=\left\{ f\in\Hom(\Vr,\Vp)\colon(\Op-\lambda)^{k}f=0\right\} \\
=\left\{ f\in\Hom(\Vr,\Vp)\colon f(v)\in\Espace{\Op}k(\lambda)\ \forall v\in\Vr\right\} =\Hom(\Vr,\Espace{\Op}k(\lambda)).\label{eq:proof_eigenspaces}
\end{multline}

The second part of the proposition is an immediate corollary of the
unitary equivalence in (\ref{eq:quotient_of_regular_rep}).

For the third part of the proposition, it is enough to show that $\Op_{\rho}$
is symmetric (since all our spaces are finite-dimensional). Let $\phi_{1},\phi_{2}\in\Hom(\Vr,\Vp)$
and $B(\Vr)$ an orthonormal basis of $\Vr$. 
\begin{align*}
\left<\Op_{\rho}\phi_{1}\,,\,\phi_{2}\right>_{\Hom(\Vr,\Vp)} & =\sum_{v\in B(\Vr)}\left\langle T\phi_{1}(v)\,,\,\phi_{2}(v)\right\rangle _{\Vp}\\
 & =\sum_{v\in B(\Vr)}\left\langle \phi_{1}(v)\thinspace,\thinspace T\phi_{2}(v)\right\rangle _{\Vp}=\left<\phi_{1}\,,\,\Op_{\rho}\phi_{2}\right>_{\Hom(\Vr,\Vp)},
\end{align*}
where in the middle, we used that $\Op$ is self-adjoint. Hence, $\Op_{\rho}$
is symmetric and self-adjoint. Proposition \ref{Prop:Spectral_property}
provides an important application of the quotient construction\footnote{For more applications see Section \ref{Sec: Applications of quot ops}.}.
It is well known that each eigenspace corresponds to a representation
of the symmetry group. Namely, that $\Espace{\Op}k(\lambda)$ is a
$G$-module. A fundamental question is which representations (or $G$-modules)
actually appear in $\Op$'s spectrum and in which frequency. Equivalently,
we ask for which $G$-modules $\Vr$, we have a nontrivial $\Hom(\Vr,\Espace{\Op}k(\lambda))$,
and how frequently (in $\Op$'s spectrum) this happens. We are now
able to answer this via formula (\ref{eq:Quotient_formula_explicit})
and equation (\ref{eq:spectral_property_of_quotient}). To find out
whether a $G$-module $\Vr$ appears in $\Op$'s spectrum, one first
computes $\Op_{\rho}$ using (\ref{eq:Quotient_formula_explicit}).
If the dimension of $\Op_{\rho}$ is non-zero then $\Vr$ appears
in $\Op$'s spectrum. The frequency of its appearance is given as
the dimension ratio of $\Op_{\rho}$ and $\Op$. A similar analysis
may be done for operators $\Op$ which are not finite dimensional,
such as the metric graph operators discussed in Section \ref{sec:QuantumGraphs}.
The existence (or nonexistence) and frequency of various representations
in the spectrum is nicely demonstrated in Example \ref{Ex: Reducing points}
for discrete graphs and in Examples \ref{ex:Tetra_QG}, \ref{ex:tetrahedrom_all3}
for metric graphs.
\end{proof}

\subsection{Proof of Theorem \ref{thm: Quotient fundamental property}\label{SubSec: Proof of fundamental property}}

Theorem \ref{thm: Quotient fundamental property} is proven with the
aid of the following lemma, which describes the ortonormal basis with
respect to which we may write $\left.\Op\right|_{\Hom(\Vr,\Vp)}$
to get $\Op_{\rho}$, as in the statement of Theorem \ref{thm: Quotient fundamental property}.
\begin{lem}
\label{lem: orthonormal basis for Hom_G}For each $i\in\D$ and $1\le n\leq d_{i}$,
let $\theta_{i}^{(n)}:\Vr\rightarrow\Vp$ be given by 
\begin{equation}
\forall v\in\Vr,\quad\quad\theta_{i}^{(n)}(v):=\frac{1}{\sqrt{\left|G\right|\left|G_{i}\right|}}\sum_{g\in G}\left\langle g\varphi_{i}^{(n)}\thinspace,\thinspace v\right\rangle _{\Vr}\thinspace\mathbf{e}_{g\thinspace i},\label{eq:definition_of_abstract_thetas}
\end{equation}
where $\{\varphi_{i}^{(n)}\}_{n=1}^{d_{i}}$ is an orthonormal basis
for $(\Vr)^{G_{i}}$, as described in Section \ref{sec:def_quotient_operators}.

The set $\bigcup_{i\in\D}\left\{ \theta_{i}^{(n)}\right\} _{n=1}^{d_{i}}$
is an orthonormal basis for $\Hom(\Vr,\Vp)$.
\end{lem}

\begin{proof}[Proof of Lemma \ref{lem: orthonormal basis for Hom_G}]

The decomposition of $\Pt$ into orbits, $\Pt=\cup_{i\in\D}O_{i}$,
induces a $G$-module orthogonal decomposition, $\Vp\cong\C[\Pt]\cong\oplus_{i\in\D}\C[O_{i}]$,
which in turn leads to the following orthogonal decomposition of the
intertwiner space, 
\begin{equation}
\Hom(\Vr,\Vp)\cong\oplus_{i\in\D}\Hom(\Vr,\C[O_{i}]).\label{eq:Hom_space_decomposed_to_orbits}
\end{equation}
The orthonormal basis we choose for $\Hom(\Vr,\Vp)$ is aligned with
the orthogonal decomposition (\ref{eq:Hom_space_decomposed_to_orbits}):
we take as a basis for $\Hom(\Vr,\Vp)$ a union of orthonormal bases
for each $\Hom(\Vr,\C[O_{i}])$. We will show that 
\begin{equation}
\Hom(\Vr,\C[O_{i}])\cong\mathrm{Hom}_{G_{i}}(\triv,\Vr),\label{eq:Hom_space_equivalence_by_Frobenius}
\end{equation}
where the above is a unitary isomorphism of vector spaces, which we
denote by $U$. Recall that we have chosen an orthonormal basis, $\{\varphi_{i}^{(n)}\}_{n=1}^{d_{i}}$,
for $(\Vr)^{G_{i}}\cong\mathrm{Hom}_{G_{i}}(\triv,\Vr)$ (see description
before Theorem \ref{def: matrix_form_of_quotient}). Therefore, $\{U\varphi_{i}^{(n)}\}_{n=1}^{d_{i}}$
is an orthonormal basis for $\Hom(\Vr,\C[O_{i}])$. Hence, by (\ref{eq:Hom_space_decomposed_to_orbits}),
taking the union $\bigcup_{i\in\D}\{U\varphi_{i}^{(n)}\}_{n=1}^{d_{i}}$
yields an orthonormal basis for $\Hom(\Vr,\Vp)$. The Lemma will then
follow once we show that the isomorphism $U$ in (\ref{eq:Hom_space_equivalence_by_Frobenius})
is unitary and $\theta_{i}^{(n)}=U\varphi_{i}^{(n)}$, where $\theta_{i}^{(n)}$
is defined in (\ref{eq:definition_of_abstract_thetas}).

The unitary isomorphism (\ref{eq:Hom_space_equivalence_by_Frobenius})
follows from a version of Frobenius reciprocity\footnote{Note that we take here a different version of Frobenius reciprocity
than the one used in the proof of Proposition \ref{prop: Algebraic properties}.} (Theorem \ref{thm:Frobenius_reciprocity},(\ref{enu: thm-Frobenius-reciprocity-1})).
We argue this below and provide an explicit expression for the unitary
map $U$. First, note that 
\begin{align*}
\Ind{G_{i}}G{\triv} & =\set{f:G\rightarrow\triv}{\forall h\in G_{i},~~f(hg)=hf(g)}\\
 & =\set{f:G\rightarrow\triv}{\forall h\in G_{i},~~f(hg)=f(g)}\\
 & \cong\C[\nicefrac{G}{G_{i}}]\cong\C[O_{i}],
\end{align*}
where the first equality is by definition of induced representation
(Definition \ref{def: induced_representation}), the second is because
$f$ maps to the trivial representation and in the third line $\nicefrac{G}{G_{i}}$
denotes the set of right cosets and we employ the orbit-stabilizer
theorem, $\left|\nicefrac{G}{G_{i}}\right|=\left|O_{i}\right|$. Explicitly,
the isomorphism $\Ind{G_{i}}G{\triv}\cong\C[O_{i}]$ is given by 
\begin{equation}
f\mapsto\frac{1}{\sqrt{\left|G\right|\left|G_{i}\right|}}\sum_{g\in G}f(g)\thinspace g^{-1}\ei.\label{eq:isomorphism_to_C_Oi}
\end{equation}
One may verify that this is ineed an isomorphism and that it is unitary
(according to the inner product on $\C[\Pt]$ and (\ref{eq:inner_product_induced_rep})).
Next, we describe a map 
\[
\widetilde{U}:\mathrm{Hom}_{G_{i}}(\triv,\Vr)\rightarrow\Hom(\Vr,\Ind{G_{i}}G{\triv})
\]
as a composition of the adjoint and the Frobenius reciprocity isomorphism,
(\ref{eq: Frobenius_map_relation-1}). Specifically, $\widetilde{U}$
maps any $\varphi\in\mathrm{Hom}_{G_{i}}(\triv,\Vr)$ such that 
\begin{align*}
\forall v\in\Vr,\forall g\in G,\quad & \widetilde{U}(\varphi)(v)(g)=\left\langle \varphi\thinspace,\thinspace gv\right\rangle _{\Vr}.
\end{align*}
The map $\widetilde{U}$ is unitary, as both the adjoint and the Frobenius
reciprocity isomorphism are unitary (for the latter see Proposition
\ref{prop: Frobenius_preserves_inner_product}). We may now combine
$\widetilde{U}$ with (\ref{eq:isomorphism_to_C_Oi}) to get 
\begin{align}
U & :\mathrm{Hom}_{G_{i}}(\triv,\Vr)\rightarrow\Hom(\Vr,\C[O_{i}])\nonumber \\
\forall v\in\Vr,\quad U(\varphi)(v) & =\frac{1}{\sqrt{\left|G\right|\left|G_{i}\right|}}\sum_{g\in G}\left\langle \varphi\thinspace,\thinspace gv\right\rangle _{\Vr}\thinspace g^{-1}\ei\label{eq:Unitary_from_phi_to_theta}\\
 & =\frac{1}{\sqrt{\left|G\right|\left|G_{i}\right|}}\sum_{g\in G}\left\langle g\varphi_{i}^{(n)}\thinspace,\thinspace v\right\rangle _{\Vr}\thinspace\mathbf{e}_{g\thinspace i},
\end{align}
where in the last line, we used the unitarity of $\rho(g)$, changed
summation $g\leftrightarrow g^{-1}$ and used that $g\thinspace\ei=\mathbf{e}_{g\thinspace i}$.
Furthermore, as we have the natural embedding $\Hom(\Vr,\C[O_{i}])\hookrightarrow\Hom(\Vr,\Vp)$,
we can actually consider $U$ as a map to $\Hom(\Vr,\Vp)$. The map
$U$ is unitary since we showed above that both (\ref{eq:isomorphism_to_C_Oi})
and $\widetilde{U}$ are unitary. Hence, defining $\theta_{i}^{(n)}=U\varphi_{i}^{(n)}$
for all $i\in\D$ and $1\leq n\leq d_{i}$, we get that $\theta_{i}^{(n)}$
is given by (\ref{eq:definition_of_abstract_thetas}) and that the
set $\bigcup_{i\in\D}\left\{ \theta_{i}^{(n)}\right\} _{n=1}^{d_{i}}$
is indeed an orthonormal basis for $\Hom(\Vr,\Vp)$.
\end{proof}
\medskip{}

\begin{proof}[Proof of Theorem \ref{thm: Quotient fundamental property}]

We explicitly write the matrix representing $\left.\Op\right|_{\Hom(\Vr,\Vp)}$
in the orthonormal basis $\bigcup_{i\in\D}\{\theta_{i}^{(n)}\}_{n=1}^{d_{i}}$
given in Lemma \ref{lem: orthonormal basis for Hom_G}. By doing so,
we will recover the expression (\ref{eq:Quotient_formula_explicit})
which is used to define $\Op_{\rho}$ and this will prove the Theorem.\allowdisplaybreaks
\begin{align*}
\left\langle \Op\thinspace\theta_{j}^{(m)},\theta_{i}^{(n)}\right\rangle _{\Hom(\Vr,\Vp)}= & \sum_{v\in B(\Vr)}\left\langle \Op\theta_{j}^{(m)}(v)~,~\theta_{i}^{(n)}(v)\right\rangle _{\Vp}\\
=\frac{1}{\sqrt{|G_{i}||G_{j}|}} & \frac{1}{\left|G\right|}\sum_{v\in B(\Vr)}\sum_{g_{j}\in G}\sum_{g_{i}\in G}\left<g_{j}\varphi_{j}^{(m)},v\right>_{\Vr}\overline{\left<g_{i}\varphi_{i}^{(n)},v\right>_{\Vr}}\thinspace\left\langle \Op g_{j}\ej\,,\,g_{i}\ei\right\rangle _{\Vp}\\
=\frac{1}{\sqrt{|G_{i}||G_{j}|}} & \frac{1}{\left|G\right|}\sum_{g_{j}\in G}\sum_{g_{i}\in G}\left<g_{j}\varphi_{j}^{(m)},g_{i}\varphi_{i}^{(n)}\right>_{\Vr}\thinspace\left<g_{i}^{-1}\Op g_{j}\ej\,,\,\ei\right>_{\Vp}\\
=\frac{1}{\sqrt{|G_{i}||G_{j}|}} & \frac{1}{\left|G\right|}\sum_{g_{j}\in G}\sum_{g_{i}\in G}\left<g_{i}^{-1}g_{j}\varphi_{j}^{(m)},\,\varphi_{i}^{(n)}\right>_{\Vr}\thinspace\left<\Op g_{i}^{-1}g_{j}\ej\,,\,\ei\right>_{\Vp}\\
=\frac{1}{\sqrt{|G_{i}||G_{j}|}} & \sum_{g\in G}\left<g\varphi_{j}^{(m)},\varphi_{i}^{(n)}\right>_{\Vr}\left<\Op g\ej,\ei\right>_{\Vp}\\
=\frac{1}{\sqrt{|G_{i}||G_{j}|}} & \sum_{g\in G}\left[\Phi_{i}^{*}\,g\Phi_{j}\right]_{m,n}\left[\Op\right]_{i,gj}\\
=\frac{1}{\sqrt{|G_{i}||G_{j}|}} & \left[\Phi_{i}^{*}\,\left(\sum_{g\in G}\left[\Op\right]_{i,gj}g\right)\Phi_{j}\right]_{m,n}
\end{align*}
 where to obtain the third line, we used that $B(\Vr)$ is an orthonormal
basis. We also used the unitarity of $g_{i}$, the $G$-symmetry of
$\Op$, and changed the summation order over the group elements ($g=g_{j}^{-1}g_{i}$).
\end{proof}
%

\section{Examples of quotients \label{sec:Examples}}

The examples in this section demonstrate different procedures for
computing quotient operators and illustrate their properties as established
in Section~\ref{sec:quotient_operators}.
\begin{example}[A basic example with one-dimensional representations]
\label{Ex: Basic example}

\begin{figure}
\centering (a) \includegraphics[width=0.38\textwidth]{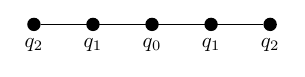}
\hspace{20pt} (b)\includegraphics[width=0.22\textwidth]{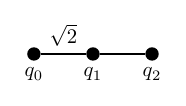}
\hspace{20pt} (c) \includegraphics[width=0.15\textwidth]{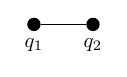}
\caption{The graph from Example~\ref{Ex: Basic example} and its two quotients.}
\label{fig:basic_example}
\end{figure}

Consider the discrete graph of Fig.~\ref{fig:basic_example}. The
discrete Laplacian we consider on this graph is, in matrix form, 
\begin{equation}
L=\begin{pmatrix}q_{2} & -1 & 0 & 0 & 0\\
-1 & q_{1} & -1 & 0 & 0\\
0 & -1 & q_{0} & -1 & 0\\
0 & 0 & -1 & q_{1} & -1\\
0 & 0 & 0 & -1 & q_{2}
\end{pmatrix},
\end{equation}
Note that we have incorporated a possibility to vary the potential
$q_{v}$ (turning the Laplacian into \term{Schrödinger operator},
although we will still refer to it as a ``Laplacian'') and we have
done it so as to preserve the reflection symmetry of the graph.

More precisely, the operator $L$ is symmetric under the action of
the cyclic group of two elements $G=C_{2}=\{e,~r\},$ where $e$ is
the identity element and $r$ encodes the horizontal reflection: it
acts on $V_{\pi}:=\C^{5}$ as the matrix 
\[
r=\begin{pmatrix}0 & 0 & 0 & 0 & 1\\
0 & 0 & 0 & 1 & 0\\
0 & 0 & 1 & 0 & 0\\
0 & 1 & 0 & 0 & 0\\
1 & 0 & 0 & 0 & 0
\end{pmatrix}.
\]
The symmetry of the Laplacian then means that it commutes with the
action of the group, i.e. $gL=Lg$ for all $g\in G$. We will compute
the quotients as the matrix representations of the operator $L\big|_{\Hom(V_{\rho},V_{\pi})}$
for all irreducible representation of the symmetry group $G$.

The group $C_{2}$ has two irreducible representations. These are
the trivial representation $V_{+}=\C$ with $r$ acting as idenitity,
and the sign representation $V_{-}=\C$ with $r$ acting as multiplication
by $-1$. From the definition of $\Hom$, equation~\eqref{eq:hom_space_definition},
we get 
\begin{align}
\Hom(V_{+},V_{\pi}) & =\left\{ \begin{pmatrix}f_{1}\\
f_{2}\\
f_{3}\\
f_{2}\\
f_{1}
\end{pmatrix}\colon f_{1},f_{2},f_{3}\in\C\right\} ,\label{eq:HomG_plus}\\
 & =\left\{ \begin{pmatrix}\phi_{1}\\
\phi_{2}\\
0\\
-\phi_{2}\\
-\phi_{1}
\end{pmatrix}\colon\phi_{1},\phi_{2}\in\C\right\} .\label{eq:HomG_minus}
\end{align}
We notice that (due to the representations being one-dimensional)
there is a natural identification of the $\Hom$ spaces with the subspaces
of $V_{\pi}$ that are symmetric or anti-symmetric with respect to
the reflection $r$.

Choosing the orthonormal basis 
\begin{equation}
\theta_{1}=e_{3},\quad\theta_{2}=\frac{e_{2}+e_{4}}{\sqrt{2}},\quad\theta_{3}=\frac{e_{1}+e_{5}}{\sqrt{2}},\label{eq:Hom+basis}
\end{equation}
for $\Hom(V_{+},V_{\pi})$, we obtain the quotient operator 
\begin{equation}
L_{\mathrm{(+)}}=\begin{pmatrix}q_{0} & -\sqrt{2} & 0\\
-\sqrt{2} & q_{1} & -1\\
0 & -1 & q_{2}
\end{pmatrix}.\label{eq:Lplus}
\end{equation}
Similarly, the quotient with respect to $V_{-}$ is 
\begin{equation}
L_{\mathrm{(-)}}=\begin{pmatrix}q_{1} & -1\\
-1 & q_{2}
\end{pmatrix}.\label{eq:Lminus}
\end{equation}
\end{example}


\begin{example}[Free action]
\label{ex:hexagon} Let us take the graph from Figure \ref{fig:benzene_setup}(a)
described by a weighted adjacency matrix $H$ with diagonal entries
$V$ and weights $a$ and $b$ assigned to the solid and dashed edges
respectively. More explicitly, 
\begin{equation}
H=\begin{pmatrix}q & a & 0 & 0 & 0 & b\\
a & q & b & 0 & 0 & 0\\
0 & b & q & a & 0 & 0\\
0 & 0 & a & q & b & 0\\
0 & 0 & 0 & b & q & a\\
b & 0 & 0 & 0 & a & q
\end{pmatrix}\label{eq:benzeneH_definition}
\end{equation}
The operator $H$ is invariant under the group $G=D_{3}=\{e,r_{1},r_{2},r_{3},\sigma,\sigma^{2}\}$,
where $r_{\alpha}$ denotes reflection over the corresponding axis
shown in the figure and $\sigma$ is rotation by $\omega=2\pi/3$.
The acton of $G$ on $\Vp:=\C^{6}$ is by permutation matrices; for
example $r_{1}$ acts as the permutation $(1\,2)(3\,6)(4\,5)$ and
$\sigma$ acts as $(1\,3\,5)(2\,4\,6)$. Recall that the action of
a permutation $\sigma$ on a vector $f\in\C^{6}$ is 
\begin{equation}
(\sigma f)_{j}=f_{\sigma^{-1}j}.\label{eq:perm_action}
\end{equation}

\begin{figure}
\centering \includegraphics{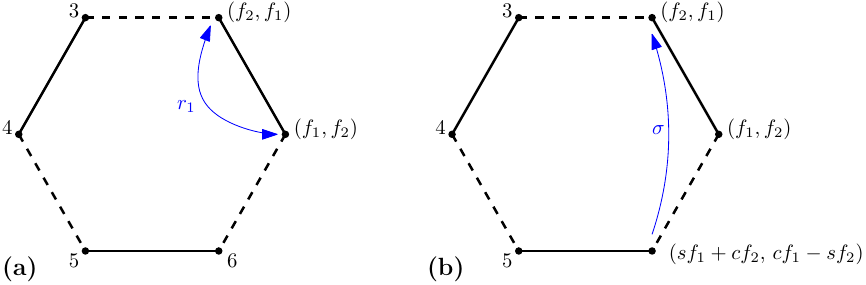} \caption{The graph from Example~\ref{ex:hexagon}, its symmetries and the
quotient with respect to representation \eqref{eq:D3stdrep}. We use
abbreviations $c=\cos\omega$ and $s=\sin\omega$.}
\label{fig:benzene_setup}
\end{figure}

The group $G=D_{3}$ has a unique two-dimensional irreducible representation
$V_{\rho}:=\C^{2}$ given by 
\begin{equation}
\sigma\mapsto\begin{pmatrix}\cos\omega & -\sin\omega\\
\sin\omega & \cos\omega
\end{pmatrix},\qquad r_{1}\mapsto\begin{pmatrix}0 & 1\\
1 & 0
\end{pmatrix},\qquad r_{3}\mapsto\begin{pmatrix}\sin\omega & \cos\omega\\
\cos\omega & -\sin\omega
\end{pmatrix},\label{eq:D3stdrep}
\end{equation}
where we included $r_{3}=r_{1}\sigma$ for future reference.

We will both describe the intertwiner space $\Hom(V_{\rho},V_{\pi})$
and compute the quotient operator $H_{\rho}$ using Definition~\ref{def: matrix_form_of_quotient}.
We will construct the intertwiner pictorially, see Fig.~\ref{fig:benzene_hom}.
An intertwiner $\phi\in\Hom(V_{\rho},V_{\pi})$ will be represented
as a $6\times2$ matrix and we will fill it up row by row. The first
row are the free parameters, 
\begin{equation}
\begin{pmatrix}\phi_{1,1} & \phi_{1,2}\end{pmatrix}=\begin{pmatrix}f_{1} & f_{2}\end{pmatrix}.\label{eq:benzene1}
\end{equation}
All other rows will be determined from the ``representative'' row
\eqref{eq:benzene1} since all other vertices lie in the orbit of
vertex 1 (the group action is transitive). We remark here that in
more sophisticated examples, there may be dependencies among the entries
of the representative row(s) if the corresponding vertex has a nontrivial
stabilizer subgroup. We will encounter this in Example~\ref{Ex: Reducing points}.

From the definition of the intertwiner, equation~\eqref{eq:hom_space_definition},
and the action of $r_{1}$ on $V_{\pi}=\C^{6}$, we get 
\begin{equation}
\phi_{2,\cdot}=(r_{1}\phi)_{1,\cdot}=\phi_{1,\cdot}r_{1}=\begin{pmatrix}f_{1} & f_{2}\end{pmatrix}\begin{pmatrix}0 & 1\\
1 & 0
\end{pmatrix}=\begin{pmatrix}f_{2} & f_{1}\end{pmatrix},\label{eq:benzene2}
\end{equation}
where $\phi_{j,\cdot}$ denotes the $j$-th row of $\phi$. This reasoning
is schematically depicted in Fig.~\ref{fig:benzene_hom}(a).

\begin{figure}
\centering \includegraphics{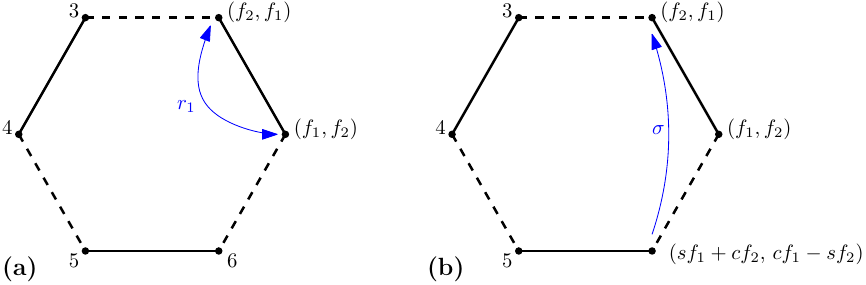} \caption{Construction of the intertwiner space in Example~\ref{ex:hexagon}.}
\label{fig:benzene_hom}
\end{figure}

Similarly, using $\sigma$, we get 
\begin{equation}
\phi_{6,\cdot}=(\sigma\phi)_{2,\cdot}=\phi_{2,\cdot}\sigma=\begin{pmatrix}f_{2} & f_{1}\end{pmatrix}\begin{pmatrix}\cos\omega & -\sin\omega\\
\sin\omega & \cos\omega
\end{pmatrix}=\begin{pmatrix}sf_{1}+cf_{2} & cf_{1}-sf_{2}\end{pmatrix},\label{eq:benzene3}
\end{equation}
where we used abbreviations $c=\cos\omega$ and $s=\sin\omega$. This
step is schematically depicted in Fig.~\ref{fig:benzene_hom}(b).

Proceeding recursively, we find 
\begin{equation}
\phi=\begin{pmatrix}f_{1} & f_{2}\\
f_{2} & f_{1}\\
cf_{1}-sf_{2} & sf_{1}+cf_{2}\\
cf_{2}-sf_{1} & sf_{2}+cf_{1}\\
sf_{2}+cf_{1} & cf_{2}-sf_{1}\\
sf_{1}+cf_{2} & cf_{1}-sf_{2}
\end{pmatrix},\label{eq:benzene_hom}
\end{equation}
where $f_{1},f_{2}\in C$ are arbitrary constants.

To summarize, equation~\eqref{eq:benzene_hom} parametrizes the space
$\Hom(V_{\rho},V_{\pi})$. Choosing a basis, we can represent the
action of $H$ on this space as a matrix (to obtain a Hermitian matrix,
the basis should be orthonormal). But to understand the action of
$H$, it is enough to focus our attention on the representative vertex
and its neighbors, already determined in \eqref{eq:benzene2} and
\eqref{eq:benzene3}: 
\begin{equation}
H\colon\begin{pmatrix}f_{1} & f_{2}\end{pmatrix}\quad\longmapsto\quad q\begin{pmatrix}f_{1} & f_{2}\end{pmatrix}+a\begin{pmatrix}f_{2} & f_{1}\end{pmatrix}+b\begin{pmatrix}sf_{1}+cf_{2} & cf_{1}-sf_{2}\end{pmatrix},\label{eq:benzene_actionH}
\end{equation}
which corresponds to the graph in Fig.~\ref{fig:benzene_setup}(c).

We now construct $H_{\rho}$ using Definition~\ref{def: matrix_form_of_quotient}.
There is only one orbit and we choose the vertex $1$ as the representative;
namely, the fundamental domain is $\D=\{1\}$. The quotient operator
$H_{\rho}$ consists of a single block $(1,1)$ of size $2\times2$:
since the group is fixed point free, $d_{1}=r=2$. Moreover, we can
use equation~\eqref{eq:quotient_formula_free} to determine its form.
The only non-zero entries in the summation are $H_{1,1}$ (which corresponds
to $g=e$), $H_{1,2}$ (with $g=r_{1}$) and $H_{1,6}$ (with $g=r_{3}$).
Therefore the resulting quotient operator is 
\begin{equation}
H_{\rho}=H_{1,1}e+H_{1,2}r_{1}+H_{1,6}r_{3}=\begin{pmatrix}q+b\sin\omega & a+b\cos\omega\\
a+b\cos\omega & q-b\sin\omega
\end{pmatrix}.\label{Eqn: D3 quot op}
\end{equation}

For example, if we take $q=0$, $a=-1$ and $b=-2$, then the spectrum
for the original matrix $H$ is 
\[
\spec(H)=\{-3,\,-\sqrt{3},\,-\sqrt{3},\,\sqrt{3},\,\sqrt{3},\,3\},
\]
while the matrix $H_{\rho}$ happens to be diagonal with entries $\pm\sqrt{3}$.
That these eigenvalues are doubly degenerate in $\spec(H)$ follow
from Proposition \ref{Prop:Spectral_property}. Explicitly, because
$\deg\rho=2$, each eigenfunction of the quotient $H_{\rho}$ corresponds
to a two-dimensional subspace of the eigenspace of $H$ with the same
eigenvalue.
\end{example}


\begin{example}[Disappearing vertices: vertices in the fundamental domain may not
be equally represented in the quotient]
\label{Ex: Reducing points}

Here we study the quotients of the graph shown in Figure \ref{fig:discrete_tetra}(a).
It is invariant under the action of the symmetry group $G=S_{4}$,
corresponding to the permutations of the vertices $1$, $2$, $3$
and $4$. The order of the group is $|S_{4}|=4!=24$ and it is generated
by the three elements $(12)$, $(23)$ and $(34)$. The fundamental
domain consists of the two vertices `$\bullet$' and `$\circ$' and
their orbits with respect to the group action are $O_{\bullet}=\{1,\ldots,4$\}
and $O_{\circ}=\{5,\ldots,16\}$. We consider the operator $H$ on
this graph given by 
\begin{equation}
H_{ij}=\begin{cases}
a & \mbox{if }i\sim j,\ \ i\in O_{\bullet},j\in O_{\circ}\mbox{ or }i\in O_{\circ},j\in O_{\bullet}\\
b & \mbox{if }i\sim j,\ \ i\in O_{\circ},j\in O_{\circ}\\
q_{\bullet} & \mbox{if }i=j\in O_{\bullet}\\
q_{\circ} & \mbox{if }i=j\in O_{\circ},
\end{cases}\label{eq:tetraH_again}
\end{equation}
where $\sim$ indicates adjacency of the vertices. In this example,
both the `$\bullet$' and `$\circ$' vertices have non-empty stabilizer
groups. Choosing $1$ and $5$ as the orbit representatives (so that
$\D=\{1,5\}$), we have $G_{\bullet}\cong S_{3}$ (permutations of
vertices $2$, $3$ and $4$) and $G_{\circ}\cong S_{2}$ (permutations
of vertices $3$ and $4$).

We now wish to find the quotient operator $H_{\rho}$ with respect
to the so-called \emph{standard} representation\footnote{This is an irreducible representation of $S_{4}$.}
of $S_{4}$ given by 
\begin{equation}
\rho_{(1\,2)}=\begin{pmatrix}0 & -1 & 0\\
-1 & 0 & 0\\
0 & 0 & 1
\end{pmatrix},\quad\rho_{(2\,3)}=\begin{pmatrix}0 & 0 & 1\\
0 & 1 & 0\\
1 & 0 & 0
\end{pmatrix},\quad\rho_{(3\,4)}=\begin{pmatrix}0 & 1 & 0\\
1 & 0 & 0\\
0 & 0 & 1
\end{pmatrix}.\label{eq:standard_rep}
\end{equation}

We start by determining the matrices $\Phi_{\circ}$ and $\Phi_{\bullet}$
consisting of orthonormal basis vectors for $(\Vr)^{G_{\circ}}$ and
$(\Vr)^{G_{\bullet}}$, defined in equation~\eqref{eq:def_Gi_invariant_subspace}.
Note that it suffices to check the invariance condition $gv=v$ on
the generators of these groups, 
\begin{align*}
(\Vr)^{G_{\circ}} & =\ker\left(\Id_{3}-\rho_{(3\,4)}\right)=\spn\left\{ \begin{pmatrix}0\\
0\\
1
\end{pmatrix},\begin{pmatrix}1\\
1\\
0
\end{pmatrix}\right\} ,\\
(\Vr)^{G_{\bullet}} & =\ker\left(\Id_{3}-\rho_{(3\,4)}\right)\bigcap\ker\left(\Id_{3}-\rho_{(2\,3)}\right)=\spn\left\{ \begin{pmatrix}1\\
1\\
1
\end{pmatrix}\right\} ,
\end{align*}
and therefore 
\begin{equation}
\Phi_{\circ}=\begin{pmatrix}0 & 1/\sqrt{2}\\
0 & 1/\sqrt{2}\\
1 & 0
\end{pmatrix},\qquad\Phi_{\bullet}=\begin{pmatrix}1/\sqrt{3}\\
1/\sqrt{3}\\
1/\sqrt{3}
\end{pmatrix}.\label{eq:tetraPhi}
\end{equation}

We proceed to finding the blocks of our quotient operator. Firstly
we notice that $H_{1,g1}$ is only non-zero if $g\in G_{\bullet}$.
Therefore we get 
\[
[H_{\rho}]_{\bullet,\bullet}=\frac{1}{\left|G_{\bullet}\right|}\sum_{g\in G_{\bullet}}\Phi_{\bullet}^{*}g\Phi_{\bullet}H_{\bullet,\bullet}=\Phi_{\bullet}^{*}\Phi_{\bullet}H_{\bullet,\bullet}=q_{\bullet},
\]
where we used that $\Phi_{\bullet}$ is invariant under $g\in G_{\bullet}$.
Similarly, we also have $H_{1,g5}\neq0$ only if $g\in G_{\bullet}$.
Therefore for the $1\times2$ block $[H_{\rho}]_{\bullet,\circ}$
we obtain 
\begin{align*}
[H_{\rho}]_{\bullet,\circ} & =\frac{1}{\sqrt{\left|G_{\bullet}\right||G_{\circ}|}}\sum_{g\in G_{\bullet}}\Phi_{\bullet}^{*}g\Phi_{\circ}H_{\bullet,\circ}\\
 & =\frac{1}{\sqrt{\left|G_{\bullet}\right||G_{\circ}|}}\sum_{g\in G_{\bullet}}(g^{-1}\Phi_{\bullet})^{*}\Phi_{\circ}H_{\bullet,\circ}\\
 & =\frac{\sqrt{|G_{\bullet}|}}{\sqrt{|G_{\circ}|}}\left(\Phi_{\bullet}^{*}\Phi_{\circ}\right)H_{\bullet,\circ}=\frac{\sqrt{6}}{\sqrt{2}}\begin{pmatrix}1/\sqrt{3} & \sqrt{2/3}\end{pmatrix}a=\begin{pmatrix}a & a\sqrt{2}\end{pmatrix}.
\end{align*}
Finally, the group elements contributing to the block $[H_{\rho}]_{\circ,\circ}$
according to (\ref{eq:Quotient_formula_explicit}) are $g=e$, $(3\,4)$,
$(1\,2)$ and $(3\,4)(1\,2)$ (for all other $g$, $H_{5,g5}=0$).
We split the summation into two parts and use the invariance properties
\[
\rho_{(3\,4)}\Phi_{\circ}=\Phi_{\circ},\qquad\rho_{(3\,4)(1\,2)}\Phi_{\circ}=\rho_{(1\,2)}\Phi_{\circ}
\]
to get 
\[
[H_{\rho}]_{\circ,\circ}=\frac{1}{|G_{\circ}|}\left(2\left(\Phi_{\circ}^{*}\Phi_{\circ}\right)q_{\circ}+2\left(\Phi_{\circ}^{*}\rho_{(1\,2)}\Phi_{\circ}\right)b\right)=\begin{pmatrix}q_{\circ}+b & 0\\
0 & q_{\circ}-b
\end{pmatrix}.
\]
Collating all these blocks together we get 
\begin{equation}
H_{\rho}=\begin{pmatrix}q_{\bullet} & a & a\sqrt{2}\\
a & q_{\circ}+b & 0\\
a\sqrt{2} & 0 & q_{\circ}-b
\end{pmatrix}.\label{eq:tetra_Hrho_answer}
\end{equation}
It is interesting to observe that the vertices $\bullet$ and $\circ$
from the fundamental domain are not equally represented in the quotient
operator (the vertex $\bullet$ partially disappeared, as $d_{\bullet}<r$).
We will see below more quotients of this kind (and even a case when
$\bullet$ disappears completely from the quotient when $d_{\bullet}=0$).

\begin{figure}
\centering \includegraphics{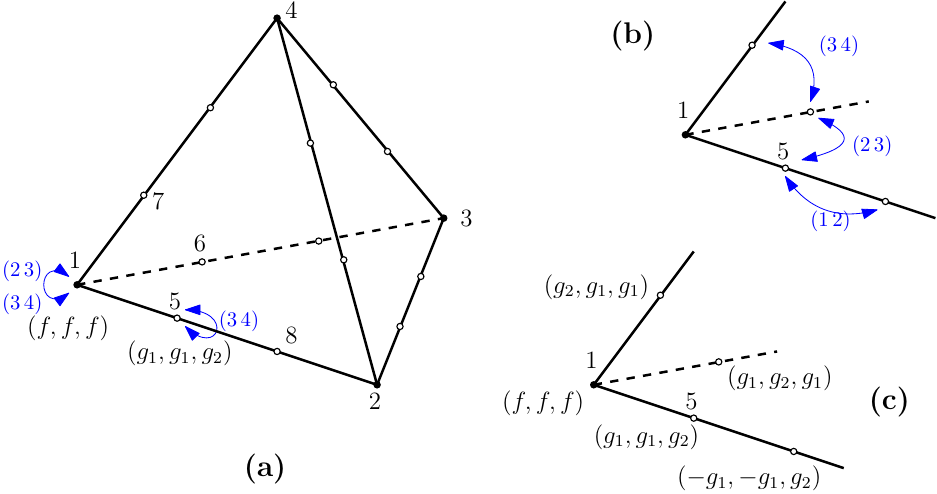} \caption{Computing the intertwiner rows in Example~\ref{Ex: Reducing points}.}
\label{fig:tetra_hom}
\end{figure}

We can also work out the intertwiners pictorially, as in the previous
examples; Figure~\ref{fig:tetra_hom} depicts the initial steps.
We note that in the first step, we observe that the number of the
free parameters describing the rows $\phi_{1,\cdot}$ and $\phi_{5,\cdot}$
are reduced because these rows must be invariant with respect to multiplication
by $\rho_{(2\,3)},\rho_{(3\,4)}$ and $\rho_{(3\,4)}$ correspondingly.
This step is analogous to finding the matrics $\Phi_{\bullet}$ and
$\Phi_{\circ}$ in equation~\eqref{eq:tetraPhi}.

To give a numerical example, we take the standard Laplacian (corresponding
to the choice $a=b=-1$, $q_{\bullet}=3$, $q_{\circ}=2$ in \eqref{eq:tetraH_again}),
and get that the spectrum of the original operator $H$ is 
\begin{equation}
\spec(H)=\{0,\ 0.44,\ 0.44,\ 0.44,\ 1,\ 1,\ 2,\ 2,\ 2,\ 3,\ 3,\ 3,\ 4,\ 4.56,\ 4.56,\ 4.56\},\label{Eqn: Complete spectrum}
\end{equation}
while the spectrum of $H_{\rho}$ picks out some triply degenerate
eigenvalues 
\[
\spec(H_{\rho})=\{0.44,\ 2,\ 4.56\}.
\]

At this point we remark that the quotient $H_{\rho}$ is smaller than
the graph announced in Example~\ref{ex:tetra_intro}. Correspondingly,
it does not capture all of the triply degenerate eigenvalues from
\eqref{Eqn: Complete spectrum}. According to Proposition \ref{Prop:Spectral_property}(\ref{enu:Prop_Spectral_property_2}),
$\spec(H)$ is the union of spectra of all quotients with respect
to the irreducible representations (and each eigenvalue is obtained
with the multiplicity of the dimension of representation). The group
$S_{4}$ has five irreducible representations. Two of them (including
$\rho$ above) are three-dimensional, one is two-dimensional and two
are one-dimensional. The latter two are the sign representation and
the trivial representation.

The quotient with respect to the trivial representation is 
\begin{equation}
H_{\mathrm{triv}}=\begin{pmatrix}q_{\bullet} & \sqrt{3}a\\
\sqrt{3}a & q_{\circ}+b
\end{pmatrix}.
\end{equation}
The same values for $q_{\circ},q_{\bullet},a$ and $b$ as before
give the spectrum $\spec(H_{\mathrm{triv}})=\{0,\ 4\}$.

The remaining eigenvalues in (\ref{Eqn: Complete spectrum}) are obtained
from the 2d irrep $H_{\text{2d}}=q_{\circ}+b=(1)$, which captures
the doubly degenerate eigenvalue $1$ in $\spec(H)$, and the other
3d irrep of $S_{4}$, whose quotient 
\begin{equation}
H_{\hat{\rho}}=q_{\circ}-b=(3)\label{eq:Hhatrho}
\end{equation}
captures the remaining triply degenerate eigenvalue in $\spec(H)$.
Note that in these instances we have $d_{\circ}=1$ and $d_{\bullet}=0$,
which leads to the quotient operators being one-dimensional. Comparing
again with Example~\ref{ex:tetra_intro}, we note that the reflection
symmetry in Fig.~\ref{fig:discrete_tetra}(b) allows to represent
$H_{\sigma}$ in \eqref{eq:Hsigma_intro} as a direct sum $H_{\rho}\oplus H_{\hat{\rho}}$
of operators \eqref{eq:tetra_Hrho_answer} and \eqref{eq:Hhatrho}.
A more direct way to obtain $H_{\sigma}$ is outlined in Example~\ref{ex:tetrahedrom_all3}
--- and it involves using a subgroup of $S_{4}$ rather than the
full group of symmetries.

We have accounted for all the eigenvalues in (\ref{Eqn: Complete   spectrum})
by examining the spectra of four quotients. However, $S_{4}$ has
another irrep, the sign representation. A quick calculation shows
that 
\[
d=\dim(\Hom(V_{\sgn},\Vp))=\frac{1}{|G|}\sum_{g\in G}\chi_{\sgn}(g)\bar{\chi}_{\pi}(g)=0,
\]
where the last equality results from Schur's orthogonality relations
and $\chi$ denote the characters of the corresponding representations.
We could have also derived this by showing that $d_{\bullet}=d_{\circ}=0$.
This means that there is no quotient operator $H_{\sgn}$ and, indeed,
there are no associated eigenvalues in $\spec(H)$.
\end{example}


\begin{example}[Directed graphs]
In all the previous examples in this section we have taken the $\pi$-symmetric
operator $\Op$ to be unitarily diagonalizable. However, Definition~\ref{def:     matrix_form_of_quotient}
of the quotient operator does not require this to be the case. Let
us therefore take the following example, which is invariant under
the symmetry group $C_{2}=\{e,\thinspace r\}$. The matrix $\Op$
in (\ref{Eqn: Directed     operator}) is invariant under the exchange
of indices $3$ and $4$. 
\begin{equation}
\Op=\begin{pmatrix}2 & 1 & 0 & 0\\
1 & 2 & 0 & 0\\
0 & 1 & 1 & 0\\
0 & 1 & 0 & 1
\end{pmatrix},\hspace{30pt}\pi_{(3\,4)}=\begin{pmatrix}1 & 0 & 0 & 0\\
0 & 1 & 0 & 0\\
0 & 0 & 0 & 1\\
0 & 0 & 1 & 0
\end{pmatrix},\hspace{30pt}\Op_{\mathrm{triv}}=\begin{pmatrix}2 & 1 & 0\\
1 & 2 & 0\\
0 & \sqrt{2} & 1
\end{pmatrix}\label{Eqn: Directed operator}
\end{equation}

The connections of $\Op$ (the non-zero off-diagonal entries) can
be interpreted in terms of the directed graph given in Figure \ref{fig: Directed graphs}
(a).

\begin{figure}[ht]
\centerline{(a) \includegraphics[width=0.35\textwidth]{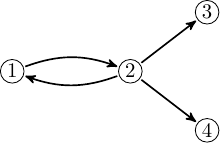}
\hspace{70pt} (b) \includegraphics[width=0.35\textwidth]{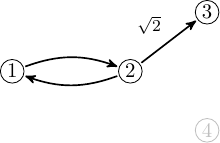}}
\caption{(a) Directed graph associated to the adjacency matrix $\protect\Op$
and (b) the directed graph associated to the quotient operator $\protect\Op_{\rho}$.}
\label{fig: Directed graphs}
\end{figure}

If we select the fundamental domain to consist of the vertices $1,2$
and $3$ and choose the trivial representation then the formula (\ref{eq:Quotient_formula_explicit})
gives, for instance, 
\[
[\Op_{\mathrm{triv}}]_{32}=\frac{1}{\sqrt{|G_{2}||G_{3}|}}(\Op_{32}+\Op_{42})=\frac{2}{\sqrt{2}}=\sqrt{2}.
\]
The complete matrix is displayed above in (\ref{Eqn: Directed operator})
and the corresponding graph is illustrated in Figure \ref{fig: Directed graphs},(b).
One sees immediately the resemblance between the original and quotient
operators. Furthermore we have the decomposition into Jordan normal
form of $\Op=PJP^{-1}$, where

\[
P=\frac{1}{4}\begin{pmatrix}2 & 0 & 2 & 0\\
2 & 0 & -2 & 0\\
1 & -2 & -1 & 0\\
1 & -2 & -1 & 1
\end{pmatrix}\hspace{50pt}J=\begin{pmatrix}3 & 0 & 0 & 0\\
0 & 1 & 1 & 0\\
0 & 0 & 1 & 0\\
0 & 0 & 0 & 1
\end{pmatrix}
\]
and similarly for $\Op_{\mathrm{triv}}=QJ_{\mathrm{triv}}Q^{-1}$,
where 
\[
Q=\frac{1}{4}\begin{pmatrix}2 & 0 & 2\\
2 & 0 & -2\\
\sqrt{2} & -2\sqrt{2} & -\sqrt{2}
\end{pmatrix}\hspace{50pt}J_{\mathrm{triv}}=\begin{pmatrix}3 & 0 & 0\\
0 & 1 & 1\\
0 & 0 & 1
\end{pmatrix}.
\]
Hence both $\Op$ and $\Op_{\mathrm{triv}}$ are non-diagonalizable.
Nevertheless we see that the spectral relation (\ref{eq:spectral_property_of_quotient})
given in Proposition \ref{Prop:Spectral_property} still holds. For
example, in this case, taking the spectral parameter $\lambda=1$
we have 
\begin{equation}
\Hom(V_{\rho},\Espace{\Op}1(1))\cong\Espace{\Op_{\rho}}1(1).\label{eq:Eigenspace_of_digraph}
\end{equation}
Note that the $\lambda=1$ eigenspace of $\Op$ is of dimension $2$
($\dim\Espace{\Op}1(1)=2$), but there is only one eigenvector (up
to scaling) in this eigenspace that transforms according to the trivial
representation. This eigenvector is explicitly given as the second
column of $P$ and it corresponds to the eigenvector of $\Op_{\mathrm{triv}}$
with the same eigenvalue (the second column of $Q$). Hence the vector
spaces in \eqref{eq:Eigenspace_of_digraph} are one-dimensional.

Similarly, $\Op$ has a $\lambda=1$ generalized eigenvector of rank
$2$, which transforms under the trivial representation. This is the
third column vector of $P$ which corresponds to the third column
vector of $Q$. Namely, the second and third column vectors of $P$
span $\Hom(V_{\rho},\Espace{\Op}2(1))$ and similarly the second and
third column vectors of $Q$ span the isomorphic space $\Espace{\Op_{\rho}}2(1)$.
\end{example}

\begin{example}[The quaternion group]
\label{Ex: The quaternion group} The current example has some interesting
physical content. Within quantum chaos, Gaussian Symplectic Ensembles
(GSE) statistics have typically been associated with the distribution
of energy levels in complex quantum systems possessing a half-integer
spin. However, recently it has been demonstrated that one may still
observe these statistics even without spin (i.e. the wavefunctions
have a single component) \cite{Joy14}, leading to the first experimental
realization of the GSE \cite{Rehemanjiang-2016}. The example in \cite{Joy14}
was achieved by obtaining a suitable quotient of a quantum graph with
a symmetry corresponding to the quaternion group $Q_{8}$. Here we
provide an analogous example using discrete graphs to highlight some
interesting properties of this quotient operator.

The quaternion group is given by the following eight group elements
\[
Q_{8}:=\{\pm1,\pm\ii,\pm\jj,\pm\kk:\ii^{2}=\jj^{2}=\kk^{2}=\ii\jj\kk=-1\}
\]
and can be generated by the two elements $\ii$ and $\jj$. We construct
a discrete graph which is symmetric under the action of $Q_{8}$.
The graph is shown in Figure \ref{fig: quaternion group} (a); it
has 24 vertices and edge weights $a,b,c,d\in\R$ such that $a\neq b$
and $c\neq d$ (one may also include vertex potentials as long as
the symmetry is retained). The graph is based on the Cayley graph
of $Q_{8}$ with respect to the generating set $\{\ii,\jj\}$ (see
e.g. \cite{Joy14}). To this Cayley graph we add a vertex at the middle
of each edge, to obtain the graph in Figure \ref{fig: quaternion group}
(a). We take a representation $V_{\pi}=\C^{24}$, which acts by permuting
the eight vertices $\{\pm1,~\pm\ii,~\pm\jj,~\pm\kk\}$ by left action
and permuting all other vertices accordingly. There is a real-symmetric
weighted adjacency matrix $A$ associated to the graph in \ref{fig: quaternion group}
(a) and this operator $A$ is $Q_{8}$-symmetric.

\begin{figure}[ht]
\centerline{(a) \includegraphics[width=0.4\textwidth]{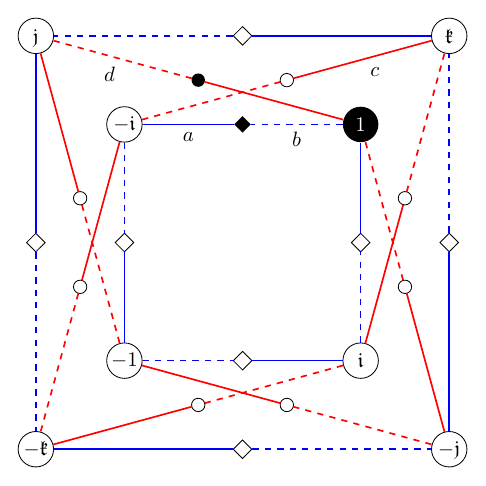}
\hspace{20pt} (b) \includegraphics[width=0.4\textwidth]{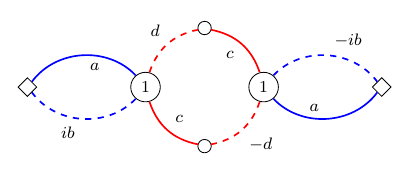}}
\caption{(a) Discrete graph with 24 vertices and edge weights $a,b,c$ and
$d$ associated to blue solid, blue dashed, red solid and red dashed
lines respectively. The vertices taken to form the fundamental domain
are marked in black. (b) The quotient graph associated to the irreducible
representation given in (\ref{Eqn: Q8 irrep}).}
\label{fig: quaternion group}
\end{figure}

The irreducible representation of $Q_{8}$, in which we are interested,
is generated by the following $2\times2$ matrices 
\begin{equation}
\ii\mapsto\begin{pmatrix}\ui & 0\\
0 & -\ui
\end{pmatrix},\hspace{30pt}\jj\mapsto\begin{pmatrix}0 & 1\\
-1 & 0
\end{pmatrix}.\label{Eqn: Q8 irrep}
\end{equation}
We choose three vertices to form the fundamental domain under the
group action (see the vertices marked in black in Figure \ref{fig: quaternion group}
(a)). Using the formula (\ref{eq:Quotient_formula_explicit}) we obtain
a quotient operator comprised of nine $2\times2$ blocks, given by
\begin{equation}
A_{\rho}=\begin{pmatrix}0 & a\Id_{2}+b\ii & c\Id_{2}+d\jj\\
a\Id_{2}+b\ii^{*} & 0 & 0\\
c\Id_{2}+d\jj^{*} & 0 & 0
\end{pmatrix}.\label{Eqn: Quaternion quotient}
\end{equation}
One may, equally validly, view each of the 6 indices as separate vertices,
in which case one has the corresponding graph shown in Figure \ref{fig: quaternion group},(b).

In contrast to the original operator $A$, the quotient operator $A_{\rho}$
is no longer real-symmetric. Instead, it is what is known as quaternion
self-dual - a generalization of self-adjointness to quaternions. Specifically,
a matrix $B$ is said to be quaternion self-dual if it has quaternion-real
entries, i.e. $B_{nm}=a_{nm}^{(1)}1+a_{nm}^{(\ii)}\ii+a_{nm}^{(\jj)}\jj+a_{nm}^{(\kk)}\kk$
where the $a_{i}$ are real coefficients, and $B_{mn}=\overline{B_{nm}}=a_{nm}^{(1)}1-a_{nm}^{(\ii)}\ii-a_{nm}^{(\jj)}\jj-a_{nm}^{(\kk)}\kk$
for all $n,m$. Here, the $2\times2$ blocks in $A_{\rho}$ are complex
representations of quaternion-real elements and the quotient operator
is quaternion self-dual since $g^{*}=-g$ for $g=\ii,\jj,\kk$. Interestingly,
a quaternion-self dual matrix will have a two-fold degeneracy in the
spectrum, which, in the physics literature, is known as Kramer's degeneracy
(see e.g. Chapter 2 of \cite{Haake-2010}).

This difference in structure between the original operator $A$ and
the quotient operator $A_{\rho}$ has important ramifications in quantum
chaos. This is because (see e.g \cite{Haake-2010} for more details)
the statistical distribution of eigenvalues in quantum systems, whose
classical counterparts are chaotic, typically align with the eigenvalues
of the Gaussian Orthogonal Ensemble (GOE) when the operator is real-symmetric
and the GSE when the operator is quaternion self-dual.
\end{example}

\section{Quantum graph quotients\label{sec:QuantumGraphs}}

Differential operators on metric graphs, also known as ``quantum
graphs'' arise in mathematical physics and spectral geometry either
as standalone models of interest or as limits of operators on thin
branching domains. We keep the description of quantum graphs below
as condensed as possible. A reader who is unfamiliar with quantum
graphs has many sources to choose from, in particular the elementary
introduction \cite{Ber_crm17}, the introductory exercises \cite{MR3880377},
the review \cite{GnuSmi_ap06} and the books \cite{Post_book12,BerKuc_graphs,Mugnolo_book}.

The Laplacian (or the Schrödinger operator) on a quantum graph is
an unbounded operator on an infinite-dimensional Hilbert space and
would not appear to fit within our current framework. However a quotient
of a quantum graph has already been defined in \cite{Ban09,Par10}
and in this section we show that we can indeed extend our methods
to this setting. The result is a more explicit and compact construction
than those given in Section 6 in \cite{Ban09} and Section 4.2 in
\cite{Par10}. Moreover, the quotient quantum graphs we obtain here
are always self-adjoint --- thus answering a question that was left
open in \cite{Ban09,Par10}.


\subsection{Quantum graphs in a nutshell}

\label{sec:qg_review}

Consider a graph $\Gamma$ with finite vertex and edge sets $\V$
and $\E$. If we associate a length $l_{e}>0$ to each edge $e\in\E$
the graph $\Gamma$ becomes a \emph{metric graph}, which allows us
to identify points $x\in[0,l_{e}]$ along each edge. To do this, we
need to assign a direction to each edge which, at the moment, we can
do arbitrarily.

We now consider functions and operators on the metric graph $\Gamma$.
To this end we take our function space to be the direct sum of spaces
of functions defined on each edge 
\begin{equation}
L^{2}(\Gamma):=\oplus_{e\in\E}L^{2}(e):=\oplus_{e\in\E}L^{2}([0,l_{e}])\quad\mbox{and}\quad H^{2}(\Gamma):=\oplus_{e\in\E}H^{2}(e):=\oplus_{e\in\E}H^{2}([0,l_{e}])\label{eq:graph_spaces}
\end{equation}
and choose some Schrödinger type operator defined on $H^{2}(\Gamma)$,
\begin{equation}
\Op:=-\triangle+q\thinspace:\thinspace\begin{pmatrix}f_{e_{1}}\\
\vdots\\
f_{e_{|\E|}}
\end{pmatrix}\mapsto\begin{pmatrix}-f_{e_{1}}''+q_{e_{1}}f_{e_{1}}\\
\vdots\\
-f_{e_{|\E|}}''+q_{e_{|\E|}}f_{e_{|\E|}}
\end{pmatrix},\label{eq:Q_graph operator}
\end{equation}
where $f_{e}$ is the component of a function $f$ on $\Gamma$ to
the edge $e\in\E$.

To make the operator self-adjoint we must further restrict the domain
of the operator by introducing vertex conditions, formulated in terms
of the values of the functions and their derivatives at the edge ends.
To formulate these conditions, we introduce the ``trace'' operators
$\gamma_{D},\gamma_{N}:H^{2}(\Gamma)\to\C^{2|\E|}$ which extract
the Dirichlet (value) and Neumann (derivative) boundary data from
a function $f$ defined on a graph, 
\begin{equation}
\gamma_{D}(f):=\begin{pmatrix}f_{e_{1}}(0)\\
f_{e_{1}}(l_{e_{1}})\\
\vdots\\
f_{e_{|\E|}}(0)\\
f_{e_{|\E|}}(l_{e_{|\E|}})
\end{pmatrix}\quad\hbox{and}\quad\gamma_{N}(f):=\begin{pmatrix}f'_{e_{1}}(0)\\
-f'_{e_{1}}(l_{e_{1}})\\
\vdots\\
f'_{e_{|\E|}}(0)\\
-f'_{e_{|\E|}}(l_{e_{|\E|}})
\end{pmatrix}\ ;\label{eq:vertex_traces}
\end{equation}
The vertex conditions are implemented by the condition 
\begin{equation}
(A\thinspace|\thinspace B)\begin{pmatrix}\gamma_{D}(f)\\
\gamma_{N}(f)
\end{pmatrix}=A\gamma_{D}(f)+B\gamma_{N}(f)=0,\label{eq:vertex_conditions_global}
\end{equation}
where $A$ and $B$ are $2|\E|\times2|\E|$ complex matrices. We shall
write $\Dom(\Op)$ for the domain of our operator, i.e.\ the set
of functions $f\in H^{2}(\Gamma)$ such that (\ref{eq:vertex_conditions_global})
is satisfied. The connectivity of the graph is retained by requiring
that only values and/or derivatives at the points corresponding to
the same vertex are allowed to enter any equation of conditions \eqref{eq:vertex_conditions_global}.
In particular this means we have the decomposition $A=\bigoplus_{v\in\V}A^{(v)}$,
and similarly for $B$. It was shown by Kostrykin and Schrader \cite{Kostrykin-1999}
that the operator is self-adjoint if and only if the following conditions
hold
\begin{enumerate}
\item \label{Itm: QG cond 1} the $2|\E|\times4|\E|$ matrix $(A|B)$ is
of full rank, i.e. $\Rank(A|B)=2|\E|$.
\item \label{Itm: QG cond 2} $AB^{*}$ is self adjoint.
\end{enumerate}
\begin{rem}
\label{Rem: Vertex condition equivalence} One obtains a unitarily
equivalent quantum graph, specified by the conditions $\tilde{A}=CA$
and $\tilde{B}=CB$, if and only if $C\in GL(\C^{2|\E|})$, \cite{Kostrykin-1999}.
\end{rem}

A compact quantum graph (i.e. a graph with a finite number of edges
and a finite length for each edge) with bounded potential $q$ has
a discrete infinite eigenvalue spectrum bounded from below, which
we denote by 
\begin{equation}
\spec(\Gamma):=\left\{ k^{2}\thinspace:\thinspace-\triangle f+qf=k^{2}f,\ (\gamma_{D}(f),\gamma_{N}(f))^{T}\in\ker(A|B)\right\} .\label{Eqn: QG spectrum}
\end{equation}

A common choice of vertex conditions are Kirchhoff--Neumann conditions:
for every vertex the values at the edge-ends corresponding to this
vertex are the same and the derivatives (with the signs given in \eqref{eq:vertex_traces})
add up to zero. In particular, at a vertex of degree 2 these conditions
are equivalent to requiring that the function is $C^{1}$ across the
vertex. Conversely, each point on an edge may be viewed as a ``dummy
vertex'' of degree 2 with Neumann-Kirchhoff conditions.


\subsection{Quantum graphs and symmetry}

We now extend the notion of a $G$-symmetric operator (Definition
\ref{def:sym_operator}) to quantum graphs. To motivate our definition
informally consider a graph $\Gamma$ with some symmetries: transformations
on the metric graph that map vertices to vertices while preserving
the graph's metric structure. We can assume, without loss of generality,
that an edge is never mapped to its own reversal.\footnote{If such an edge exists we can always introduce a dummy vertex in the
middle of it, splitting the edge into two (which are now mapped to
each other). Introducing dummy vertices at midpoints of \emph{every}
such edge restores the property of mapping vertices to vertices.} If this condition is satisfied, it is easy to see that the edge directions
(as described in Section \ref{sec:qg_review}) may be assigned so
that they are preserved under the action of the symmetry group. From
now on we will denote by $ge=e'$ the action of an element $g$ of
the symmetry group $G$ on the edges with their assigned direction.
Now the ``preservation of metric structure'' mentioned above is
simply the condition that the edge lengths are fixed by all $g\in G$,
i.e. $l_{ge}=l_{e}$. This allows us to pointwise compare functions
defined on two $G$-related edges. We are now ready to introduce the
notion of a $G$-symmetric quantum graph.
\begin{defn}
\label{def:QG_pi_symmetric} Let $\Vp:=\C^{\E}$ be the vector space
of functions $\E\to\C$ and let $\pi:G\rightarrow\GL(\Vp)$ be a group
homomorphism such that for each $g\in G$, $\pi(g)$ is a permutation
matrix. A quantum graph $(\Gamma,T)$ is $G$-symmetric (with respect
to the representation $\pi$) if for all $g\in G$
\begin{enumerate}
\item \label{enu:QG_symmetric_lengths} the edge lengths are preserved:
$l_{ge}=l_{e}$ for all $e\in\E$,
\item \label{enu:QG_symmetric_potential} the potential is preserved: $q_{ge}\equiv q_{e}$
for all $e\in\E$,
\item \label{enu:QG_symmetric_domain} the operator domain is preserved:
\begin{equation}
f=\begin{pmatrix}f_{e_{1}}\\
\vdots\\
f_{e_{|\E|}}
\end{pmatrix}\in\Dom(\Op)\qquad\Leftrightarrow\qquad\pi(g)f:=\begin{pmatrix}f_{g^{-1}e_{1}}\\
\vdots\\
f_{g^{-1}e_{|\E|}}
\end{pmatrix}\in\Dom(\Op).\label{eq:VC_preserved}
\end{equation}
\end{enumerate}
\end{defn}

Some comments are in order. Since $l_{e}=l_{ge}$, there is a natural
isomorphism $H^{2}(e)\cong H^{2}(ge)$ for every $g\in G$. This allows
us to define the action of $\pi(g)$ on the function space $H^{2}(\Gamma)$
(or any other similar space on $\Gamma$) as in the right-hand side
of \eqref{eq:VC_preserved}. Informally, this just means that the
transformation $\pi(g)$ takes the function from an edge $e$ and
places it on the edge $ge$. In ``placing'' the function, we implicitly
use the isomorphism mentioned above. In particular, the condition
on the potential $q=(q_{e_{1}},\ldots,q_{e_{|\E|}})^{T}$ in Definition~\ref{def:QG_pi_symmetric}
can be now written as $\pi(g)q=q$ for any $g\in G$.

We note that the action of $\pi(g)$ above coincides with the result
of the formal multiplication of the vector $f$ by the permutation
matrix $\pi(g)$. This multiplication is ``formal'' because different
entries of $f$ belong to different spaces and there is no a priori
way to shuffle them or create linear combinations. It is only through
symmetry of the graph and the resulting isomorphism between different
edge spaces that we get a meaningful result. We will return to this
point in the proof of Theorem~\ref{thm:commutative_diagram-QG} below.

As a corollary of the discussion above we get that if $\Gamma$ is
a $G$-symmetric graph (with respect to a representation $\pi$) then
the function spaces $L^{2}(\Gamma)$ and $H^{2}(\Gamma)$ are $G$-modules,
and so is also $\Dom(\Op)\subset H^{2}(\Gamma)$.

Moving on to the requirement on the domain of $T$, one can easily
check that 
\begin{equation}
\gamma_{D}(\pi(g)f)=\hat{\pi}(g)\gamma_{D}(f),\qquad\text{where }\hat{\pi}(g):=\pi(g)\otimes\Id_{2},\label{eq:trace_transformation}
\end{equation}
and similarly for the Neumann trace, $\gamma_{N}$. In view of description
\eqref{eq:vertex_conditions_global} of the domain $\Dom(T)$, condition~\eqref{eq:VC_preserved}
is equivalent to 
\begin{equation}
AF+BF'=0\qquad\Leftrightarrow\qquad A\hat{\pi}(g)F+B\hat{\pi}(g)F'=0\qquad\text{for all }g\in G\text{ and }F,F'\in\C^{2|\E|}.\label{eq:symmetric_VC}
\end{equation}

Finally, we remark that, in general, the matrices $A$ and $B$ are
not $G$-symmetric (see Example \ref{Ex: quotient_quantum_graph}).
But there is always an equivalent choice that is $G$-symmetric (with
respect to the permutation representation $\hat{\pi}$).
\begin{lem}
\label{lem:invariant_AB} Let $A$ and $B$ define vertex conditions
on a $G$-symmetric quantum graph. Then the matrices 
\begin{equation}
\tA=(A+\rmi B)^{-1}A\qquad\mbox{and}\qquad\tB=(A+\rmi B)^{-1}B.\label{eq:invariant_AB}
\end{equation}
define equivalent vertex conditions. Furthermore, $\tA$ and $\tB$
are $G$-symmetric, i.e., 
\begin{equation}
\tA\hat{\pi}(g)=\hat{\pi}(g)\tA\qquad\mbox{and}\qquad\tB\hat{\pi}(g)=\hat{\pi}(g)\tB\qquad\text{for all }g\in G.\label{eq:pi-invariance_def}
\end{equation}
\end{lem}

\begin{proof}
It is shown in \cite{Kostrykin-1999} (see also \cite[lemma~1.4.7]{BerKuc_graphs})
that the matrix $A+\rmi B$ is invertible. Therefore (see Remark \ref{Rem: Vertex condition equivalence})
$\widetilde{A}$ and $\widetilde{B}$ define equivalent vertex conditions.
To show (\ref{eq:pi-invariance_def}) we observe that since $\tA+\rmi\tB=\Id$
commutes with $\hat{\pi}(g)$ for any $g\in G$, it is enough to show
that $\tA-\rmi\tB$ also commutes with $\hat{\pi}(g)$ for any $g\in G$.

Take an arbitrary vector $a\in\C^{2|\E|}$ and let 
\begin{equation}
b=(\tA-\rmi\tB)a.\label{eq:b_vector_def}
\end{equation}
Writing out the definitions of $\tA$ and $\tB$ from (\ref{eq:invariant_AB})
and using \eqref{eq:symmetric_VC} we get for all $g\in G$ 
\begin{align*}
b=(\tA-\rmi\tB)a~~ & \Leftrightarrow~~b=(A+\rmi B)^{-1}(A-\rmi B)a\\
 & \Leftrightarrow A(b-a)+B(\rmi b+\rmi a)=0\\
 & \Leftrightarrow A\hat{\pi}(g)(b-a)+B\hat{\pi}(g)(\rmi b+\rmi a)=0\\
 & \Leftrightarrow(A\hat{\pi}(g)+\rmi B\hat{\pi}(g))b=(A\hat{\pi}(g)-\rmi B\hat{\pi}(g))a\\
 & \Leftrightarrow b=\hat{\pi}(g)^{*}(A+\rmi B)^{-1}(A-\rmi B)\hat{\pi}(g)a.
\end{align*}

We conclude that for every $a$ 
\begin{equation}
(\tA-\rmi\tB)a=\hat{\pi}(g)^{*}(\tA-\rmi\tB)\hat{\pi}(g)a,\label{eq:ab_play_result}
\end{equation}
which establishes the desired invariance.
\end{proof}
\begin{rem}
\label{rem:AB_symmetric} In equation~\eqref{eq:vertex_conditions_global},
the rows of $A$ and $B$ correspond to the restrictions imposed on
the domain of the operator. Unlike the columns of $A$ and $B$, the
rows sre not a priori related to any particular location in the graph.

Once symmetrized by \eqref{eq:invariant_AB}, $\tA$ and $\tB$ can
be properly viewed as operators on $\C^{\E}\otimes\C^{2}=\C^{2|\E|}$.
The latter space is a $G$-module with $g$ acting by $\pi(g)\otimes\Id_{2}$.
When $\tA$ and $\tB$ are $G$-symmetric, we can take their quotient
with respect to a representation $\Vr$.

From now on we will assume that our chosen $A$ and $B$ have already
been symmetrized using \eqref{eq:invariant_AB}. Their rows and columns
will be labeled by the index set $\big\{1,\o 1,2,\o 2,\ldots|\E|,\o{|\E|}\big\}$.
We will also use the following notation for $2\times2$ blocks of
$A$ corresponding to \emph{edges} $e_{i}$, $e_{j}$: 
\begin{equation}
A_{i,j}^{\otimes}=\begin{pmatrix}A_{i,j} & A_{i,\o j}\\
A_{\o i,j} & A_{\o i\o j}
\end{pmatrix},\label{eq:ABblock_notation}
\end{equation}
and similarly for $B$.
\end{rem}


\subsection{Quotient quantum graph\label{subsec:Quotient-QG}}

We start from a graph $\Gamma$ with a corresponding operator $\Op:\Dom(\Op)\to L^{2}(\Gamma)$
which is $G$-symmetric. For any $G$-module $\Vr$ we will describe
a new graph $\Gamma_{\rho}$ with a corresponding operator $\Op_{\rho}$
(see Definition \ref{def:quotient_QG} below). In Theorem \ref{thm:commutative_diagram-QG}
we will see that $\Op_{\rho}$ is self-adjoint and unitarily equivalent
to $\left.T\right|_{\Hom(\Vr,\Dom(\Op))}$, which is analogous to
Theorem \ref{thm: Quotient fundamental property} for discrete graph
operators.

As described in the previous subsection, both $\Dom(\Op)$ and $L^{2}(\Gamma)$
are $G$-modules. Since $G$ acts on the edge set $\E$ of $\Gamma$
(this action is given by the permutation representation $\pi$), we
may choose a fundamental domain $\D=\{e_{1},e_{2}\ldots,e_{|\D|}\}$
by selecting precisely one edge $e_{i}$ from each orbit $O_{i}=\{e_{k}\in\E:\exists g\in G,e_{k}=ge_{i}\}$.
For each $e_{i}\in\D$, consider the stabilizer group $G_{i}=\set{g\in G}{ge_{i}=e_{i}}$
and $(\Vr)^{G_{i}}$, the $G_{i}$-invariant subspace of $\Vr$. Denote
$d_{i}:=\dim(\Vr)^{G_{i}}$ and take $d_{i}$ copies of the directed
edge $e_{i}$, i.e., these are $d_{i}$ intervals, each of length
$l_{e_{i}}$. Denote these by $\{e_{i,j}\}_{j=1,\ldots,d_{i}}$. Furthermore,
assume that the matrices $A,~B$ describing the vertex conditions
of $\Gamma$ are $G$-symmetric with respect to the permutation representation
$\hat{\pi}(g)=\pi(g)\otimes\Id_{2}$ (see Lemma \ref{lem:invariant_AB}).
\begin{defn}[Quotient quantum graph]
\label{def:quotient_QG} Let $(\Gamma,T)$ be a $G$-symmetric quantum
graph, where $G$ is a finite group and let $\Vr$ be a $G$-module.
Referring to the description before this definition, we define the
quotient graph $\Gamma_{\rho}$ to be a metric graph formed from the
edges $\{e_{i,j}\colon e_{i}\in\D,\ j=1,\ldots,d_{i}\}$ of length
$l_{e_{i,j}}=l_{e_{i}}$. We define the operator $\Op_{\rho}$ by
the differential expression $-d^{2}/dx^{2}+q_{e_{i}}$ acting on the
edge $e_{i,j}$, with the vertex conditions specified by the matrices
$A_{\rho}$ and $B_{\rho}$, which are the quotients of $A$ and $B$
by $\Vr$.
\end{defn}

\begin{rem}
\label{rem:compute_Arho} To give an explicit description of $A_{\rho}$
and $B_{\rho}$, we recall formula (\ref{eq:Quotient_formula_explicit}).
For each $e_{i}\in\D$ pick $\{\varphi_{i}^{(n)}\}_{n=1}^{d_{i}}$
to be an orthonormal basis for $(\Vr)^{G_{i}}$ and form an $r\times d_{i}$
matrix $\Phi_{i}$ whose columns are these basis vectors ($r=\dim\rho$).
With respect of the action of $G$ on $\C^{2|\E|}=\Vp\otimes\C^{2}$,
we note that $G_{i}=G_{\o i}$ and therefore $d_{i}=d_{\o i}$ and
$\Phi_{i}=\Phi_{\o i}$.

Now, $A_{\rho}$ and $B_{\rho}$ are matrices comprised of $2|\D|\times2|\D|$
blocks with the $(z,w)$-th block, $z,w\in\big\{1,\o 1,\ldots,|\D|,\o{|\D|}\big\}$,
being the following $d_{z}\times d_{w}$ matrix 
\begin{align}
\left[A_{\rho}\right]_{z,w} & =\frac{1}{\sqrt{|G_{z}||G_{w}|}}\Phi_{z}^{*}\,\left(\sum_{g\in G}A_{z,gw}\,g\right)\,\Phi_{w}\label{eq: Quotient A matrix}\\
\left[B_{\rho}\right]_{z,w} & =\frac{1}{\sqrt{|G_{z}||G_{w}|}}\Phi_{z}^{*}\,\left(\sum_{g\in G}B_{z,gw}\,g\right)\,\Phi_{w}.\label{eq: Quotient B matrix}
\end{align}

Due to the order we assume on the entries of the boundary data in
\eqref{eq:vertex_traces}, the blocks computed in \eqref{eq: Quotient     A matrix}
and \eqref{eq: Quotient B matrix} are not contiguous within the matrices
$A_{\rho}$ and $A_{\rho}$. To overcome this problem, we can instead
use the notation introduced in \eqref{eq:ABblock_notation} and write
\begin{equation}
\left[A_{\rho}\right]_{i,j}^{\otimes}=\frac{1}{\sqrt{|G_{i}||G_{j}|}}\sum_{g\in G}\left(\Phi_{i}^{*}g\Phi_{j}\right)\otimes A_{i,gj}^{\otimes}\qquad\left[B_{\rho}\right]_{i,j}^{\otimes}=\frac{1}{\sqrt{|G_{i}||G_{j}|}}\sum_{g\in G}\left(\Phi_{i}^{*}g\Phi_{j}\right)\otimes B_{i,gj}^{\otimes},\label{eq:ABrho_tensor}
\end{equation}
where $i,j\in\{1,2,\ldots,|\D|\}$. In Example~\ref{Ex:     quotient_quantum_graph}
we will use this formula to compute the quotient.
\end{rem}

The operator $\Op_{\rho}$ defined above plays the same role as the
discrere graph quotient operator in Definition \ref{def:   matrix_form_of_quotient}.
To see this we note that if $\Vr$ is a $G$-module then $\Op$ acts
on elements of $\Hom(\Vr,\Dom(\Op))$ by composition and returns elements
of $\Hom(\Vr,L^{2}(\Gamma))$ (see (\ref{eq:T_acts_on_Hom_part1}),(\ref{eq:T_acts_on_Hom_part2})
for the finite dimensional case). The space $\Hom(\Vr,\Dom(\Op))$
warrants a deeper discussion as it is useful in computational examples
(see Example~\ref{ex:Tetra_QG}). It consists of the row vectors
$\phi=\left(\phi^{(1)},\ldots,\phi^{(r)}\right)$, $\phi^{(j)}\in\Dom(\Op)$,
$j=1,\ldots,r=\dim(\Vr)$, satisfying the intertwining condition 
\begin{equation}
\pi(g)\phi:=\left(\pi(g)\phi^{(1)},\ldots,\pi(g)\phi^{(r)}\right)=\phi g,\qquad\forall\ g\in G,\label{eq:intertwining_cond_qg}
\end{equation}
where each $\phi^{(j)}$ can be visualized as a column 
\[
\phi^{(j)}=\begin{pmatrix}\phi_{e_{1}}^{(j)},\\
\vdots\\
\phi_{e_{|\E|}}^{(j)}
\end{pmatrix},\qquad\phi_{e_{k}}^{(j)}\in H^{2}([0,l_{e_{k}}]).
\]
The action $\pi(g)\phi^{(j)}$ is defined in \eqref{eq:VC_preserved},
while $g$ on the right-hand side of \eqref{eq:intertwining_cond_qg}
denotes the action of $g$ on $\Vr$.

The space $\Hom(\Vr,L^{2}(\Gamma))$ is defined analogously. The Hilbert
space structure on

$\Hom(\Vr,L^{2}(\Gamma))$ is given by the Frobenius inner product
\begin{equation}
\langle\phi,\psi\rangle_{\Hom(\Vr,L^{2}(\Gamma))}=\Tr\left(\phi^{*}\psi\right):=\sum_{j=1}^{r}\sum_{e\in\E}\langle\phi_{e}^{(j)},\psi_{e}^{(j)}\rangle_{L^{2}([0,l_{e}])}.\label{eq:scalar_prod_Hom}
\end{equation}
 The action of $T$ mentioned above is 
\begin{equation}
T:\Hom(\Vr,\Dom(\Op))\to\Hom(\Vr,L^{2}(\Gamma))\qquad T:\left(\phi^{(1)},\ldots,\phi^{(r)}\right)\mapsto\left(T\phi^{(1)},\ldots,T\phi^{(r)}\right).\label{eq:T_action_Hom}
\end{equation}
For this operator, which we denote by $T\big|_{\Hom(\Vr,\Dom(\Op))}$,
we have the following.
\begin{thm}
\label{thm:commutative_diagram-QG} Let $G$ be a finite group, $(\Gamma,T)$
be a $G$-symmetric quantum graph (with the operator $T$ being self-adjoint)
and let $\Vr$ be a $G$-module. Then the operator $\Op_{\rho}$ is
self-adjoint and is unitarily equivalent to $\left.T\right|_{\Hom(\Vr,\Dom(\Op))}$.
\end{thm}

The proof of Theorem~\ref{thm:commutative_diagram-QG} follows the
spirit of the proof of Theorem~\ref{thm: Quotient fundamental property},
but has more technicalities due to the intertwiners being linear maps
into an infinite dimensional space. For this reason we defer it to
Appendix~\ref{app:qg_quotient}.

Theorem \ref{thm:commutative_diagram-QG} shows that the quotient
quantum graph we describe in Definition \ref{def:quotient_QG} is
also a quotient quantum graph according to \cite[Definition 1]{Par10}.
There are several advantages to the quotient construction in Definition
\ref{def:quotient_QG}. The first is that the construction method
of the quotient is simpler to implement and convenient for computer-aided
computation (by employing the formulas (\ref{eq:   Quotient A matrix}),(\ref{eq: Quotient B matrix})).
Another advantage is that our present construction ensures the self-adjointness
of the operator on the quotient graph and by this answers a question
which was left open in \cite{Ban09,Par10}. In addition, it is possible
to construct quotients for general Schrödinger operators and not just
the Laplacian. Finally, our construction extends to an alternative
description of quantum graphs by scattering matrices, as detailed
in the next section.

Note that the graph topology is described indirectly by the vertex
conditions, which are implemented through $A_{\rho}$ and $B_{\rho}$.
That these matrices satisfy the conditions reviewed in Section~\ref{sec:qg_review}
is established in Lemma \ref{lem:valid_sa}.


\subsection{Bond scattering matrix and quotients}

\label{subsec:Quotient-QG-scatter}

The Schrödinger operator of a quantum graph is not a finite dimensional
operator. However, there exists a so-called \emph{unitary evolution}
operator $U(k)=T(k)S(k)$ of finite dimension (equal to $2|\E|\times2|\E|$)
that can fully describe the spectrum. Here $T(k)$ is a \emph{bond
transmission matrix} and $S(k)$ a \emph{bond scattering matrix} expressed
in terms of the vertex condition matrices,

\begin{equation}
\forall k\in\mathbb{R},\quad\quad~~~S(k)=-\left(A+\rmi kB\right)^{-1}\left(A-\rmi kB\right)J,\label{eq:Vertex_Scattering_Matrix}
\end{equation}
where 
\begin{equation}
J=\Id_{|\E|}\otimes\left(\begin{smallmatrix}0 & 1\\
1 & 0
\end{smallmatrix}\right),\label{eq:J_Matrix}
\end{equation}
is a matrix that swaps between entries corresponding to start-points
and end-points of edges. Note that (\ref{eq:Vertex_Scattering_Matrix})
is well defined, as it is shown in \cite{Kostrykin-1999} (see also
\cite[lemma~1.4.7]{BerKuc_graphs}) that the matrix $A+\rmi kB$ is
invertible.

The matrix $S(k)$ is assured to be unitary if the operator is self-adjoint,
namely, if $(A|B)$ is of full rank and $AB^{*}$ is self adjoint
\cite[lemma 1.4.7]{BerKuc_graphs}. In addition, it might happen that
$S(k)$ is $k$-independent. This occurs, for example, when the vertex
conditions are either Kirchhoff-Neumann (see Example \ref{Ex:   quotient_quantum_graph})
or Dirichlet (a vertex with Dirichlet condition means that the function
vanishes at that vertex).

To understand the physical meaning of $S(k)$ recall that its rows
and columns are indexed by the set $\{1,\o 1,\ldots,|\E|,\o{|\E|}\}$.
We view $j$ as corresponding to a wave traveling along the edge $e_{j}$
in the direction we assigned to $e_{j}$ and $\o j$ as corresponding
to the wave travelling against the assigned direction. The entry $[S(k)]_{z,z'}$
then gives the quantum amplitude of wave scattering from the direction
$z'$ to direction $z$. In particular, $[S(k)]_{z,z'}\neq0$ only
if $z'$ is directed into some vertex $v$ and $z$ is directed out
of $v$. The matrix $S(k)$ therefore contains all the information
about the graph's connectivity and its vertex conditions.

In contrast, the transmission matrix $T(k)$ is a diagonal matrix
describing the wave evolution along the edges (in or against the assigned
direction). When the potential $q$ is zero, $T(k)=e^{ik\hat{L}}$,
where $\hat{L}:=\diag(l_{e_{1}},l_{e_{1}},\ldots,l_{e_{|\E|}},l_{e_{|\E|}})$
is a diagonal matrix of edge lengths, each listed twice (for the two
directions). The spectrum of the Laplacian (\ref{Eqn: QG   spectrum})
is then given by 
\begin{equation}
\spec(\Gamma)=\left\{ k^{2}\colon\det\left(\Id_{2|\E|}-U\left(k\right)\right)=0\right\} \ ,\label{eq:secular_cond}
\end{equation}
as was shown in \cite{Bel_laa85,KotSmi_prl97} (see also \cite[Theorem 2.1.8]{BerKuc_graphs}).

We may apply the theory of the current paper to compute the quotient
of the graph's scattering matrix. One then finds that taking the quotient
operator of the scattering matrix equals the scattering matrix of
the quotient graph. For simplicity we assume that the potential $q$
is zero.
\begin{prop}
\label{prop:scat_mat_fact} Let $(\Gamma,T)$ be a $G$-symmetric
quantum graph with zero potential, whose bond scattering matrix is
denoted by $S(k)$. Let $\Vr$ be a $G$-module and $(\Gamma_{\rho},T_{\rho})$
be the corresponding quotient graph as in Definition \ref{def:quotient_QG}.

Then, for all $k\in\C$, $S(k)$ and $U(k)$ are $G$-symmetric, i.e.,
\begin{equation}
\forall g\in G;\qquad S(k)\hat{\pi}(g)=\hat{\pi}(g)S(k)\qquad\mbox{and}\qquad U(k)\hat{\pi}(g)=\hat{\pi}(g)U(k).\label{eq: S(k) and U(k) are G-symmetric}
\end{equation}
In addition, the scattering matrix of $(\Gamma_{\rho},T_{\rho})$
is given by the quotient operator $S_{\rho}(k)$ and the unitary evolution
operator of $(\Gamma_{\rho},T_{\rho})$ is given by the quotient $U_{\rho}(k)$.
\end{prop}

The proof of this Proposition is given in Appendix~\ref{app:qg_quotient_scat}.

We end by pointing out that there is another meaning to the notion
`scattering matrix of a quantum graph'. This other scattering matrix
is formed by connecting semi-infinite leads to some of the graph vertices.
The dimension of this matrix is equal to the number of leads and each
of its entries equals the probability amplitude for a wave to scatter
from a certain lead to another. This amplitude is calculated by summing
over all possible paths through the graph leading from the first lead
to the second (including a direct transmission of the wave from the
first lead to the second, if it exists). This matrix is sometimes
called an exterior scattering matrix (to distinguish it from the edge-scattering
matrix discussed above). More on this matrix can be found in \cite{BanBerSmi_ahp12,KotSmi_jpa03}.

If a graph $\Gamma$ is $G$-symmetric and leads are attached to it
in a way which respects this symmetry, then the obtained exterior
scattering matrix, $\Sigma(k)$ inherits this symmetry. As a consequence,
we are able to construct a quotient of this scattering matrix with
respect to a representation $\rho$ of the symmetry group $G$. This
quotient matrix, $\Sigma_{\rho}(k)$ was shown in \cite{BanSawSmi_jpa10}
to equal the exterior scattering matrix of the quotient graph $\Gamma_{\rho}$,
with an appropriate connection of leads. Hence, the quotient theory
in the current paper may be used to recover the previously obtained
results in \cite{BanSawSmi_jpa10}.


\subsection{Quotient quantum graph examples}

\label{subsec:Quotient-QG-examle}
\begin{example}
\label{Ex: quotient_quantum_graph} Let $\Gamma$ be a star graph
with three edges, one edge of length $l_{1}$ and the other two of
length $l_{2}$ (see Figure \ref{Fig:Star_grah_and_quotient} (a)).
Equip all graph vertices with Neumann conditions. We observe that
the graph is symmetric with respect to exchanging the two edges with
the same length, i.e. $e_{2}\leftrightarrow e_{3}$. Hence, the symmetry
group we take is $C_{2}=\{\mathrm{I},\mathrm{R}:\mathrm{R}^{2}=\mathrm{I}\}$,
and the graph is $\pi$ symmetric, where the representation $\pi$
is 
\[
\pi(\mathrm{I})=\Id_{3},\hspace{50pt}\pi(\mathrm{R})=\left(\begin{smallmatrix}1 & 0 & 0\\
0 & 0 & 1\\
0 & 1 & 0
\end{smallmatrix}\right).
\]

\begin{figure}[ht]
\centerline{(a) \includegraphics[width=0.35\textwidth]{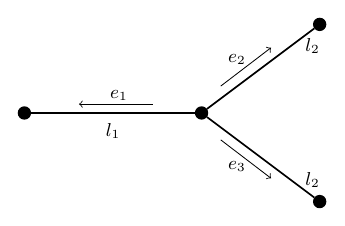}
\hspace{20pt} (b) \includegraphics[width=0.4\textwidth]{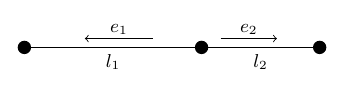}}
\caption{(a) A star graph $\Gamma$ with $C_{2}$ symmetry. (b) The corresponding
quotient graph $\Gamma_{\rho}$ with respect to the trivial representation
$\rho=\protect\triv$. Its vertex conditions are given in (\ref{Eqn: Quot graph vertex conditions}).}
\label{Fig:Star_grah_and_quotient}
\end{figure}

The vertex conditions of the graph may be described by (using the
notational convention as in (\ref{eq:vertex_conditions_global})):
\[
A=\left(\begin{smallmatrix}1 & 0 & -1 & 0 & 0 & 0\\
0 & 0 & 1 & 0 & -1 & 0\\
0 & 0 & 0 & 0 & 0 & 0\\
0 & 0 & 0 & 0 & 0 & 0\\
0 & 0 & 0 & 0 & 0 & 0\\
0 & 0 & 0 & 0 & 0 & 0
\end{smallmatrix}\right),\hspace{50pt}B=\left(\begin{smallmatrix}0 & 0 & 0 & 0 & 0 & 0\\
0 & 0 & 0 & 0 & 0 & 0\\
1 & 0 & 1 & 0 & 1 & 0\\
0 & 1 & 0 & 0 & 0 & 0\\
0 & 0 & 0 & 1 & 0 & 0\\
0 & 0 & 0 & 0 & 0 & 1
\end{smallmatrix}\right),
\]
Those matrices are not $\hat{\pi}$-symmetric (e.g., $A\hat{\pi}(R)\neq\hat{\pi}(R)A$).
With the aid of Lemma~\ref{lem:invariant_AB} we replace those by
the following $\hat{\pi}$-symmetric matrices which describe equivalent
vertex conditions: 
\[
\tA=(A+\rmi B)^{-1}A=\frac{1}{3}\left({\scriptsize
\global\long\def\arraystretch{0.8}%
\arraycolsep=.2em\begin{array}{cc|cc|cc}
2 & 0\, & -1 & \,0\, & -1 & 0\\
0 & 0\, & 0 & 0 & 0 & 0\\
\hline \rule{0pt}{8pt}-1 & 0\, & 2 & 0 & -1 & 0\\
0 & 0\, & 0 & 0 & 0 & 0\\
\hline \rule{0pt}{8pt}-1 & 0\, & -1 & 0 & 2 & 0\\
0 & 0\, & 0 & 0 & 0 & 0
\end{array}}\right),\qquad\tB=(A+\rmi B)^{-1}B=\frac{1}{3i}\left({\scriptsize
\global\long\def\arraystretch{0.8}%
\arraycolsep=.2em\begin{array}{cc|cc|cc}
1 & 0 & 1 & 0 & 1 & 0\\
0 & 3 & 0 & 0 & 0 & 0\\
\hline \rule{0pt}{8pt}1 & 0 & 1 & 0 & 1 & 0\\
0 & 0 & 0 & 3 & 0 & 0\\
\hline \rule{0pt}{8pt}1 & 0 & 1 & 0 & 1 & 0\\
0 & 0 & 0 & 0 & 0 & 3
\end{array}}\right),
\]
where the partitioning is added to align with block notation~\eqref{eq:ABblock_notation}.

We are now in a position to construct the quotient graph $\Gamma_{\rho}$,
and choose to do it for $\rho=\triv$, the trivial representation
of $C_{2}$. We choose our fundamental domain to be $\D=\{e_{1},e_{2}\}$,
with $e_{2}$ being a representative of the orbit $O_{e_{2}}=\{e_{2},e_{3}\}$.

We have $|G_{1}|=2$ and $|G_{2}|=1$, as well as $\Phi_{1}=\Phi_{2}=1$.
Using \eqref{eq:ABrho_tensor}, we get 
\begin{align}
\left[A_{\rho}\right]_{1,1}^{\otimes} & =\frac{1}{2}\left[A_{1,1}^{\otimes}+A_{1,1}^{\otimes}\right]=\frac{1}{2}\left[\left(\begin{smallmatrix}2 & 0\\
0 & 0
\end{smallmatrix}\right)+\left(\begin{smallmatrix}2 & 0\\
0 & 0
\end{smallmatrix}\right)\right]=\left(\begin{smallmatrix}2 & 0\\
0 & 0
\end{smallmatrix}\right),\label{eq:Arho11}\\
\left[A_{\rho}\right]_{1,2}^{\otimes} & =\frac{1}{\sqrt{2}}\left[A_{1,2}^{\otimes}+A_{1,3}^{\otimes}\right]=\frac{1}{\sqrt{2}}\left[\left(\begin{smallmatrix}-1 & 0\\
0 & 0
\end{smallmatrix}\right)+\left(\begin{smallmatrix}-1 & 0\\
0 & 0
\end{smallmatrix}\right)\right]=\left(\begin{smallmatrix}-\sqrt{2} & 0\\
0 & 0
\end{smallmatrix}\right)=\left[A_{\rho}\right]_{2,1}^{\otimes},\label{eq:Arho12}\\
\left[A_{\rho}\right]_{2,2}^{\otimes} & =\left[A_{2,2}^{\otimes}+A_{2,3}^{\otimes}\right]=\left[\left(\begin{smallmatrix}2 & 0\\
0 & 0
\end{smallmatrix}\right)+\left(\begin{smallmatrix}-1 & 0\\
0 & 0
\end{smallmatrix}\right]\right)=\left(\begin{smallmatrix}1 & 0\\
0 & 0
\end{smallmatrix}\right).\label{eq:Arho22}
\end{align}
Performing the same computations for $B_{\rho}$, we obtain, altogether,
\begin{equation}
A_{\rho}=\frac{1}{3}\left(\begin{array}{cc|cc}
2 & 0 & -\sqrt{2} & 0\\
0 & 0 & 0 & 0\\
\hline \rule{0pt}{14pt}-\sqrt{2} & 0 & 1 & 0\\
0 & 0 & 0 & 0
\end{array}\right),\hspace{20pt}B_{\rho}=\frac{1}{3i}\left(\begin{array}{cc|cc}
1 & 0 & \sqrt{2} & 0\\
0 & 3 & 0 & 0\\
\hline \rule{0pt}{14pt}\sqrt{2} & 0 & 2 & 0\\
0 & 0 & 0 & 3
\end{array}\right).\label{Eqn: Quot graph vertex conditions}
\end{equation}

We get a quotient graph which consists of two edges of lengths $l_{1},l_{2}$
(Figure \ref{Fig:Star_grah_and_quotient} (b)). The boundary vertices
of the quotient retain the Neumann conditions (2nd and 4th rows of
$A$ and $B$), whereas the central vertex corresponds to the conditions
(that can be deduced from the 1st and 3rd rows) 
\begin{equation}
\sqrt{2}f_{e_{1}}(0)=f_{e_{2}}(0),\hspace{20pt}f'_{e_{1}}(0)+\sqrt{2}f'_{e_{2}}(0)=0.\label{eq:graph-example-VC-at-center}
\end{equation}
\end{example}

Let us complement this example by computing the corresponding bond-scattering
matrices. First, using (\ref{eq:Vertex_Scattering_Matrix}), the scattering
matrix of the original graph is:

\[
\forall k\in\mathbb{R},\quad\quad S(k)=-J\left(A+ikB\right)^{-1}\left(A-ikB\right)=\frac{1}{3}\left({\scriptsize
\global\long\def\arraystretch{0.8}%
\arraycolsep=.2em\begin{array}{cc|cc|cc}
0 & -1 & 0 & 2 & 0 & 2\\
3 & 0 & 0 & 0 & 0 & 0\\
\hline \rule{0pt}{8pt}0 & 2 & 0 & -1 & 0 & 2\\
0 & 0 & 3 & 0 & 0 & 0\\
\hline \rule{0pt}{8pt}0 & 2 & 0 & 2 & 0 & -1\\
0 & 0 & 0 & 0 & 3 & 0
\end{array}}\right),
\]
where we note that $S(k)$ is $k$-independent and that it does not
depend on whether we take $A,B$ or $\tA,\tB$ above. The scattering
matrix of the quotient graph is obtained by following the same procedure
as in~\eqref{eq:Arho11}-\eqref{eq:Arho22} to obtain

\[
S_{\rho}(k)=\frac{1}{3}\left(\begin{array}{cc|cc}
0 & -1 & 0 & \sqrt{8}\\
3 & 0 & 0 & 0\\
\hline \rule{0pt}{14pt}0 & \sqrt{8} & 0 & 1\\
0 & 0 & 3 & 0
\end{array}\right).
\]
In particular, the vertex conditions, (\ref{eq:graph-example-VC-at-center}),
at the central vertex of the graph correspond to the unitary submatrix,
$\frac{1}{3}\left(\begin{smallmatrix}-1 & \sqrt{8}\\
\sqrt{8} & 1
\end{smallmatrix}\right)$.


\begin{example}[Metric tetrahedron graph]
\label{ex:Tetra_QG}

Consider the tetrahedron graph consisting of four vertices and six
edges of equal length $\ell$ connecting the vertices, see Fig.~\ref{fig:tetraQGhom}(a).
The operator $T$ acts as $-d^{2}/dx^{2}+q(x)$ on the functions defined
on the edges; the potential $q(x)$ is assumed to be sufficiently
regular (e.g.\
 piecewise continuous), identical on each edge and symmetric with
respect to the midpoint of the edge, namely $q(x)=q(\ell-x)$. We
impose Neumann--Kirchhoff (NK or ``standard'') conditions at the
vertices. Namely, at each vertex, the Dirichlet data agree and the
Neumann data add up to zero.

\begin{figure}
\centering \includegraphics{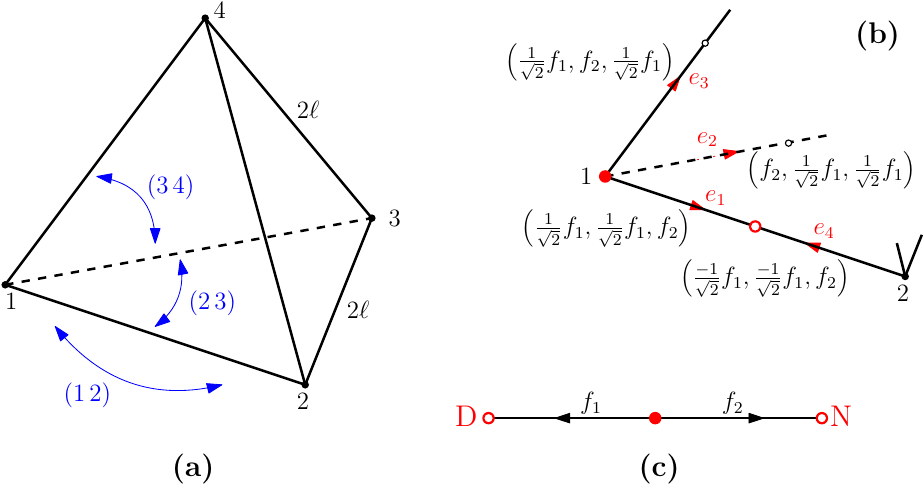} \caption{Tetrahedron quantum graph and its quotient construction via description
of the intertwiner space.}
\label{fig:tetraQGhom}
\end{figure}

The group $S_{4}$ of all permutations of the 4 vertices induces an
action on $L^{2}(\Gamma)$ that leaves the domain of the operator
invariant. For the purpose of quotient construction, we add ``dummy''
vertices at the midpoints of every edge. The conditions at the new
vertices are Neumann-Kirchhoff, which preserves the domain of the
operator modulo obvious isometries.

We compute the quotient graph with respect to the standard representation
\eqref{eq:standard_rep} of $S_{4}$, which we repeat here for convenience
\[
\rho_{(1\,2)}=\begin{pmatrix}0 & -1 & 0\\
-1 & 0 & 0\\
0 & 0 & 1
\end{pmatrix},\ \rho_{(2\,3)}=\begin{pmatrix}0 & 0 & 1\\
0 & 1 & 0\\
1 & 0 & 0
\end{pmatrix},\ \rho_{(3\,4)}=\begin{pmatrix}0 & 1 & 0\\
1 & 0 & 0\\
0 & 0 & 1
\end{pmatrix},\ \rho_{(2\,4)}=\begin{pmatrix}1 & 0 & 0\\
0 & 0 & 1\\
0 & 1 & 0
\end{pmatrix}.
\]
We will compute the quotient by both parametrizing the intertwiner
space $\Hom(\Vr,\Dom(T))$ and by following Definition~\ref{def:quotient_QG}.

Let $e_{1}$ denote the edge from the vertex $1$ to the midpoint
of the (former) edge $(1,2)$, see Fig.~\ref{fig:tetraQGhom}(b)
where the edge labels are shown in red. The orbit of this edge covers
the whole graph and it is fixed by the permutation $(3\,4)$.

Denote the row of the intertwiner $\phi$ corresponding to $e_{1}$
by $\phi_{e_{1}}=(\phi_{1,1},\phi_{1,2},\phi_{1,3})$, where each
$\phi_{1,j}$ belongs to $H^{2}(e_{1})=H^{2}([0,\ell/2])$. We first
take care of the stabilizer group which imposes the condition 
\begin{equation}
\phi_{e_{1}}=\phi_{e_{1}}\rho_{(3\,4)}.\label{eq:edge1_qg}
\end{equation}
We therefore parametrize $\phi_{e_{1}}$ as 
\begin{equation}
\phi_{e_{1}}=\left(\frac{1}{\sqrt{2}}f_{1},\frac{1}{\sqrt{2}}f_{1},f_{2}\right),\qquad(f_{1},f_{2})\in H^{2}([0,\ell/2])\oplus H^{2}([0,\ell/2]),\label{eq:phi_qg_parametrization}
\end{equation}
where the factors $1/\sqrt{2}$ are added to make it a linear isometry
(up to an overall factor). We fill out the other rows of $\phi$ as
follows.

Since the edge $e_{1}$ is the image of $e_{2}$ under the action
of $(2\,3)$, we have 
\begin{equation}
\phi_{e_{2}}=\phi_{e_{1}}\rho_{(2\,3)}=\left(f_{2},\frac{1}{\sqrt{2}}f_{1},\frac{1}{\sqrt{2}}f_{1}\right).\label{eq:edge2_qg}
\end{equation}
Similarly, we get 
\begin{equation}
\phi_{e_{3}}=\phi_{e_{2}}\rho_{(3\,4)}=\left(\frac{1}{\sqrt{2}}f_{1},f_{2},\frac{1}{\sqrt{2}}f_{1}\right),\quad\text{and}\quad\phi_{e_{4}}=\phi_{e_{1}}\rho_{(1\,2)}=\left(\frac{-1}{\sqrt{2}}f_{1},\frac{-1}{\sqrt{2}}f_{1},f_{2}\right).\label{eq:edge3_qg}
\end{equation}
This is enough to compute the matching conditions on the functions
$f_{1},f_{2}\in H^{2}([0,\ell/2])$. From NK conditions at the vertex
of degree 3 we have 
\begin{equation}
\frac{1}{\sqrt{2}}f_{1}(0)=f_{2}(0)\quad\text{and}\quad\sqrt{2}f_{1}'(0)+f_{2}'(0)=0.\label{eq:conditions_box}
\end{equation}
At the vertex of degree 2 (empty circle in Fig.~\ref{fig:tetraQGhom}(b))
we have 
\begin{align*}
 & \frac{1}{\sqrt{2}}f_{1}(\ell/2)=\frac{-1}{\sqrt{2}}f_{1}(\ell/2),\quad &  & -\frac{1}{\sqrt{2}}f_{1}'(\ell/2)-\frac{-1}{\sqrt{2}}f_{1}'(\ell/2)=0,\\
 & f_{2}(\ell/2)=f_{2}(\ell/2),\quad &  & -f_{2}'(\ell/2)-f_{2}'(\ell/2)=0.
\end{align*}
which are equivalent to 
\begin{equation}
f_{1}(\ell/2)=0,\qquad f_{2}'(\ell/2)=0.\label{eq:conditions_sides}
\end{equation}
To summarize, conditions \eqref{eq:conditions_box} and \eqref{eq:conditions_sides}
define a self-adjoint operator, acting as before by $-d^{2}/dx^{2}+q(x)$
on the edges of the graph in Fig.~\ref{fig:tetraQGhom}(c).\\

\begin{figure}
\centering \includegraphics{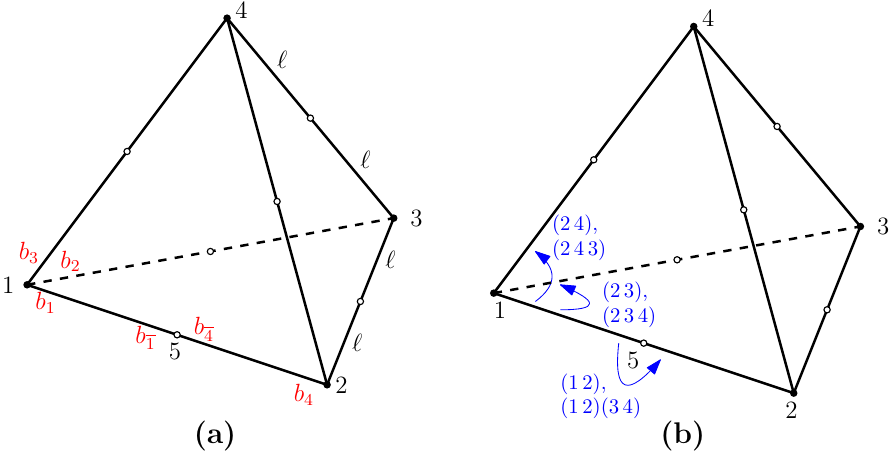} \caption{Tetrahedron quantum graph and the symmetry group action on the edge
endpoints.}
\label{fig:tetraQGaction}
\end{figure}

We now construct the quotient graph by following Definition~\ref{def:quotient_QG}.
The symmetry group $S_{4}$ has an induced action on the edge endpoints,
which we denote\footnote{$b$ for ``boundary''.} $b_{1},b_{2},\ldots,b_{\o 1},b_{\o 2},\ldots$,
see Fig.~\ref{fig:tetraQGaction}(a). There are two orbits under
the group action, and we choose $b_{1}$ and $b_{\o 1}$ as the representatives.

The vertex conditions are local in that they only link the Dirichlet
and Neumann values among the endpoints incident to the same vertex.
We can therefore work locally, separately treating the block of $A$
(or $B$) corresponding to $\{b_{1},b_{2},b_{3}\}$ (incident to the
vertex of degree 3) and the block of $A$ corresponding to $\{b_{\o 1},b_{\o 4}\}$
(incident to the midpoint vertex labeled by $5$ in Fig.~\ref{fig:tetraQGaction}).

After symmetrization \eqref{eq:invariant_AB}, the corresponding blocks
of $A$ and $B$ are 
\begin{align}
 & A=\frac{1}{3}\begin{pmatrix}2 & -1 & -1\\
-1 & 2 & -1\\
-1 & -1 & 2
\end{pmatrix},\qquad &  & B=\frac{1}{3i}\begin{pmatrix}1 & 1 & 1\\
1 & 1 & 1\\
1 & 1 & 1
\end{pmatrix}\label{eq:NKmat3}\\
 & A=\frac{1}{2}\begin{pmatrix}1 & -1\\
-1 & 1
\end{pmatrix},\qquad &  & B=\frac{1}{2i}\begin{pmatrix}1 & 1\\
1 & 1
\end{pmatrix}.\label{eq:NKmat2}
\end{align}
The stabilizer groups of $b_{1}$ and $b_{\o 1}$ are $G_{b_{1}}=G_{b_{\o 1}}=\{e,(3\,4)\}$
and therefore a valid choice of $\Phi$ is 
\begin{equation}
\Phi_{b_{1}}=\Phi_{b_{\o 1}}=\begin{pmatrix}\frac{1}{\sqrt{2}} & 0\\
\frac{1}{\sqrt{2}} & 0\\
0 & 1
\end{pmatrix}.\label{eq:Phi_tetraQG}
\end{equation}
Due to locality of the vertex conditions, the summation in equation~\eqref{eq: Quotient A matrix}
can be restricted to the subgroups fixing vertex $1$ and the midpoint
vertex $5$, correspondingly. For vertex $1$ the subgroup is $\{e,(3\,4),(2\,3),(2\,3\,4),(2\,4),(2\,4\,3)\}$.
Using $(2\,3\,4)=(2\,3)(3\,4)$ and $(2\,4\,3)=(2\,4)(3\,4)$ to exploit
the invariance properties of $\Phi_{b_{1}}$, we write 
\begin{align}
\left[A_{\rho}\right]_{b_{1},b_{1}} & =\frac{1}{2}\Phi_{b_{1}}^{*}\Big(A_{b_{1},b_{1}}\big(e+(3\,4)\big)+A_{b_{1},b_{2}}(2\,3)\big(e+(3\,4)\big)+A_{b_{1},b_{3}}(2\,4)\big(e+(3\,4)\big)\Big)\Phi_{b_{1}}\nonumber \\
 & =\Phi_{b_{1}}^{*}\big(A_{b_{1},b_{1}}e+A_{b_{1},b_{2}}(2\,3)+A_{b_{1},b_{3}}(2\,4)\big)\Phi_{b_{1}}=\frac{1}{3}\begin{pmatrix}1 & -\sqrt{2}\\
-\sqrt{2} & 2
\end{pmatrix}.\label{eq:Arho_tetraGQ1}
\end{align}
Similarly for $B_{\rho}$, 
\begin{equation}
\left[B_{\rho}\right]_{b_{1},b_{1}}=\frac{1}{3i}\begin{pmatrix}2 & \sqrt{2}\\
\sqrt{2} & 1
\end{pmatrix}\label{eq:Brho_tetraQG1}
\end{equation}
These matrices are the symmetrized versions of the condition~\eqref{eq:conditions_box}.

For the midpoint vertex $4$, the relevant subgroup is $\{e,(3\,4),(1\,2),(1,2)(3\,4)\}$
and the quotient vertex conditions are 
\begin{equation}
\left[A_{\rho}\right]_{b_{\o 1},b_{\o 1}}=\Phi_{b_{\o 1}}^{*}\big(A_{b_{\o 1},b_{\o 1}}e+A_{b_{\o 1},b_{\o 4}}(1\,2)\big)\Phi_{b_{\o 1}}=\begin{pmatrix}1 & 0\\
0 & 0
\end{pmatrix},\quad\text{and}\quad\left[B_{\rho}\right]_{b_{\o 1},b_{\o 1}}=\frac{1}{i}\begin{pmatrix}0 & 0\\
0 & 1
\end{pmatrix},\label{eq:Arho_tetraQG2}
\end{equation}
in agreement with \eqref{eq:conditions_sides}.
\end{example}


\begin{example}[Tetrahedron revisited]
\label{ex:tetrahedrom_all3}

The quotient graph we computed in Example~\ref{ex:Tetra_QG} does
not cover all triply degenerate eigenvalues of the tetrahedron graph.
This is because the group $S_{4}$, in addition to the standard representation
$\rho$ given by \eqref{eq:standard_rep} has another 3-dimensional
representation. This representation is obtained by multiplying all
matrices in equation~\eqref{eq:standard_rep} by $-1$ and we will
denote it by $\hat{\rho}$.

\begin{figure}
\centering \includegraphics{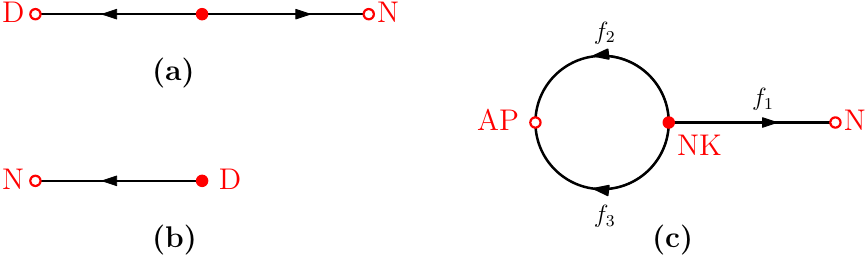} \caption{Various quotients of the tetrahedron quantum graph from Example~\ref{ex:tetrahedrom_all3}.}
\label{fig:tetraQGall3}
\end{figure}

The procedure for computing the quotient with respect to $\hat{\rho}$
is identical to the one followed in Example~\ref{ex:Tetra_QG} with
the only significant difference being in the matrix $\Phi$. The result
is shown in Fig.~\ref{fig:tetraQGall3}(b); it is a single interval
of length $\ell$ with Dirichlet and Neumann conditions.

In fact, the two quotients can be ``joined'' in a single graph,
shown in Fig.~\ref{fig:tetraQGall3}(c). The conditions at the central
vertex are the standard NK (as opposed to the more exotic conditions
at the corresponding vertex of the quotient in Fig.~\ref{fig:tetraQGall3}(a));
the conditions at the left vertex are ``anti-periodic'', 
\begin{equation}
f_{2}(\ell)=-f_{3}(\ell),\qquad f_{2}'(\ell)=f_{3}'(\ell).\label{eq:antiperiodic}
\end{equation}
This graph has $C_{2}$ symmetry (up-down reflection) and the quotients
with respect to this symmetry are exactly $\Gamma_{\rho}$ and $\Gamma_{\hat{\rho}}$
shown in Fig.~\ref{fig:tetraQGall3}(a) and (b).

The graph in Fig.~\ref{fig:tetraQGall3}(c) captures \emph{all} triply
degenerate eigenvalues of the tetrahedron graph $\Gamma$ for a typical
potential\footnote{To what extent it is true in general is an open question, see \cite{Zel_aif90,BerLiu_lmp18}.}.
With the hindsight, we can get to this graph more directly, by restricting
the representation $\rho$ to the subgroup $A_{4}$ of $S_{4}$ consisting
of all even permutations. This restricted representation $\sigma$
is specified by 
\begin{equation}
\sigma\colon(2\,3\,4)\mapsto\begin{pmatrix}0 & 0 & 1\\
1 & 0 & 0\\
0 & 1 & 0
\end{pmatrix},\qquad(1\,2)(3\,4)\mapsto\begin{pmatrix}-1 & 0 & 0\\
0 & -1 & 0\\
0 & 0 & 1
\end{pmatrix},\label{eq:irrepA4}
\end{equation}
and has the property that $\Ind{A_{4}}{S_{4}}{\sigma}=\rho\oplus\hat{\rho}$.
Therefore, by Proposition~\ref{prop: Algebraic     properties},
parts (\ref{enu: prop_algebraic_subgroup_quotient}) and (\ref{enu: prop_algebraic_decomp_2parts}),
the quotient $\Gamma_{\sigma}$ is unitarily equivalent to the direct
sum of $\Gamma_{\rho}$ and $\Gamma_{\hat{\rho}}$. In fact, computing
$\Gamma_{\sigma}$ is somewhat simpler because the stabilizer groups
are trivial.

The same computation can be performed on the operator $H$ from Example~\ref{Ex: Reducing points},
resulting in 
\begin{equation}
H_{\sigma}=\begin{pmatrix}q_{\bullet} & a & a & a\\
a & q_{\circ}-b & 0 & 0\\
a & 0 & q_{\circ}-b & 0\\
a & 0 & 0 & q_{\circ}+b
\end{pmatrix}.\label{eq:Hsigma}
\end{equation}
The corresponding graph appears in Fig.~\ref{fig:discrete_tetra}(b).
\end{example}


\section{Applications of Quotient Operators}

\label{Sec: Applications of quot ops}

In this section we point out some applications of quotient operators.
In particular, we show how some earlier results are obtained as particular
cases of the theory constructed in Section~\ref{sec:quotient_operators}.
This allows for extension of those earlier results and thus provides
further applications.

\subsection{Isospectrality\label{Subsec: Applications - Isospectrality}}

Proposition \ref{prop: Algebraic properties} may be perceived as
an application of quotients to get isospectral examples. Indeed, that
Proposition points out a few pairs of operators that are unitarily
equivalent and hence isospectral. The third part of that Proposition
allows to prove the following theorem which lies in the heart of many
isospectral examples.
\begin{thm}
\label{thm:isospectral_finite_dim_operators} Let $G$ be a finite
group. Let $\Op$ be a finite dimensional operator which is $G$-symmetric.
Let $H_{1},H_{2}$ be subgroups of $G$ with corresponding $H_{1}$-module
$\textrm{V}_{\sigma_{1}}$ and $H_{2}$-module $\textrm{V}_{\sigma_{2}}$.
If 
\begin{equation}
\Ind{H_{1}}G{\textrm{V}_{\sigma_{1}}}\cong\Ind{H_{2}}G{\textrm{V}_{\sigma_{2}}},\label{eq:Induction_condition}
\end{equation}
then $\Op_{\sigma_{1}}$ and $\Op_{\sigma_{2}}$ are unitarily equivalent
and, hence, isospectral.
\end{thm}

\begin{proof}
The proof is a straightforward corollary of Proposition \ref{prop: Algebraic properties},(\ref{enu: prop_algebraic_subgroup_quotient}).
Denote $\textrm{V}_{\rho_{n}}\cong\Ind{H_{n}}G{\textrm{V}_{\sigma_{n}}}$
for $n=1,2$. By (\ref{eq:quotient_of_induced_rep}) in Proposition
\ref{prop: Algebraic properties} we get $\Op_{\sigma_{n}}\cong\Op_{\rho_{n}}$
for both $i=1,2$. The condition (\ref{eq:Induction_condition}),
or equivalently $\textrm{V}_{\rho_{1}}=\textrm{V}_{\rho_{2}}$ implies
$\Op_{\rho_{1}}\cong\Op_{\rho_{2}}$ and finishes the proof.
\end{proof}
To connect this theorem to existing isospectral examples in the literature,
we first refer to the pioneering work of Sunada \cite{Sun85}, which
provided a general method for the construction of isospectral objects.
That construction is based on taking the ``trivial'' quotients of
a manifold by some groups in the following sense. Let $M$ be a manifold
with a group $G$ acting on it. Let $H_{1}$ and $H_{2}$ be subgroups
of $G$ which satisfy the following condition 
\begin{equation}
\forall g\in G,\quad\quad~~~\left|\left[g\right]\cap H_{1}\right|=\left|\left[g\right]\cap H_{2}\right|,\label{eq:Sunada_condition}
\end{equation}
where $\left[g\right]$ indicates the conjugacy class of $g$. In
\cite[Thm. 1]{Sun85} it is stated that if condition (\ref{eq:Sunada_condition})
holds then the quotient manifolds $\nicefrac{M}{H_{1}}$, $\nicefrac{M}{H_{2}}$
are isospectral, . It is shown in \cite{Brooks_cm99,Pesce_cm94} that
(\ref{eq:Sunada_condition}) is equivalent to the alternative condition
\begin{equation}
\Ind{H_{1}}G{\triv_{H_{1}}}\cong\Ind{H_{2}}G{\triv_{H_{2}}},\label{eq:Sunada_condition_Induction}
\end{equation}
where $\triv_{H_{n}}$ is the trivial representation of $H_{n}$ ($n=1,2$).

Observe that the condition (\ref{eq:Sunada_condition_Induction})
in Sunada's theorem is a particular case of condition (\ref{eq:Induction_condition})
in Theorem \ref{thm:isospectral_finite_dim_operators} when $\sigma_{n}=\mathrm{triv}H_{n}$
for both $n=1,2$. In this case, the isospectrality implied by Theorem
\ref{thm:isospectral_finite_dim_operators} is the same as Sunada's,
once it is shown that the topological quotient, $\nicefrac{M}{H_{n}}$,
is the same manifold as the quotient $M_{\mathrm{triv}H_{n}}$. This
indeed follows by the quotient definition in \cite{Ban09,Par10} and
in that sense, \cite{Ban09,Par10} extended Sunada's isospectrality
theorem (Indeed see \cite[Cor 4.4]{Ban09}, \cite[Cor 4]{Par10} for
the manifold and metric graph version\footnote{The theorem here is slightly stronger since it establishes that the
operators are not only isospectral but unitarily equivalent.} of Theorem \ref{thm:isospectral_finite_dim_operators}). Another
aspect of \cite{Ban09,Par10} is the adjustment of the theory from
manifolds to quantum (metric) graphs. In the current paper we provide
the discrete graph analogue of this isospectral theory (which actually
holds for any finite-dimensional operator). Moreover, in Section \ref{sec:QuantumGraphs}
we show that our methods extend to metric graphs (where the operator
of interest is unbounded) and provide a more explicit and compact
construction of a quotient quantum graph than the one given in \cite{Ban09,Par10}.

In light of the above we are now able to re-examine some earlier works
on isospectrality of discrete graphs and show how previous results
may be obtained as particular applications of the theory in the current
paper.

\subsubsection*{Regular graphs and free group action}

In \cite{Brooks_aif99} Brooks considers $k$-regular graphs and groups
which act freely on them and constructs isospectral graphs. Explicitly,
Theorem 1.1 in \cite{Brooks_aif99} can be viewed as a particular
case of Theorem \ref{thm:isospectral_finite_dim_operators} above,
in the following sense: We take $\Op$ in Theorem \ref{thm:isospectral_finite_dim_operators}
to be the discrete Laplacian on some $k$-regular graph, 
\begin{equation}
\Op_{i,j}=\begin{cases}
1-\frac{1}{k}A_{i,i} & i=j\\
-\frac{1}{k}A_{i,j} & i\neq j
\end{cases},\label{eq:normalized_Lap_k_reg}
\end{equation}
where $A_{i,j}$ indicates the number of edges connecting $i$ with
$j$ (and in $A_{i,i}$ every loop is counted twice). Assume that
$\Op$ is $G$-symmetric and the group $G$ acts freely on the graph
vertices. Take two subgroups $H_{1},H_{2}<G$, which satisfy condition
(\ref{eq:Induction_condition}) with their trivial representations
$\sigma_{n}=\mathrm{triv}H_{n}$ ($n=1,2$). In order to express the
operators $\Op_{\mathrm{triv}H_{1}}$ and $\Op_{\mathrm{triv}H_{2}}$
we may employ formulas \eqref{eq:quotient_formula_trivial_rep} (thanks
to using trivial representation) and \eqref{eq:quotient_formula_free}
(thanks to the free actions) and have 
\begin{equation}
\left[\Op_{\mathrm{triv}H_{n}}\right]_{i,j}=\sum_{g\in H_{n}}\Op_{i,gj}.\label{eq:quotient_of_Brooks}
\end{equation}

Both $\Op_{\mathrm{triv}H_{1}}$ and $\Op_{\mathrm{triv}H_{2}}$ are
also normalized Laplacians corresponding to $k$-regular graphs. By
Theorem \ref{thm:isospectral_finite_dim_operators} we get that they
are isospectral and this is exactly the statement of Theorem 1.1 in
\cite{Brooks_aif99}.

\subsubsection*{Graph Laplacians and weak fixed point condition}

Halbeisen and Hungerbühler extend in \cite{Hal99} the isospectral
construction of Brooks. They consider graphs which are not necessarily
$k$-regular and a group action which is not necessarily free. They
consider the (non-normalized) Laplacian, $\Op=D-A$, where $A$ is
the adjacency matrix (as in (\ref{eq:normalized_Lap_k_reg})) and
$D$ is a diagonal matrix of vertex degrees. Assume that $\Op$ is
$G$-symmetric, and replace the free action requirement by a weaker
condition\footnote{This is called the weak fixed point condition in \cite{Hal99}, where
free action is referred to as the strong fixed point condition.}, which in our terminology may be stated as $\left\{ \Op_{i,j}\neq0~\textrm{and }i,j\in\D\right\} ~\Rightarrow~G_{i}=G_{j}$.
Namely, if two vertices in the fundamental domain are adjacent then
their stabilizers are equal. Employing this condition in our formula
(\ref{eq:quotient_formula_trivial_rep}), we get $\left[\Op_{\mathrm{triv}G}\right]_{i,j}=\sum_{j'\in O_{j}}\Op_{i,j'}$.
It can be checked that this quotient operator, $\Op_{\mathrm{triv}G}$,
is exactly the (non-normalized) Laplacian associated with the quotient
graph. In this case our quotient graph, $\Gamma_{\mathrm{triv}G}$,
is identical to the quotient graph in the topological sense, $\nicefrac{\Gamma}{G}$.
Using this construction, \cite{Hal99} in effect employ condition
(\ref{eq:Induction_condition}) with $\sigma_{1}=\mathrm{triv}H_{1}$,
$\sigma_{2}=\mathrm{triv}H_{2}$ to generate examples of isospectral
discrete graphs. We stated above their construction using our terminology,
whereas in \cite{Hal99} it is done and proven differently.

\subsubsection*{Isospectral quantum (metric) graphs}

We mentioned above that an isospectral construction for metric graphs
in the spirit of the current paper already appeared in \cite{Ban09,Par10}.
In Section \ref{sec:QuantumGraphs} we provide an alternative construction
of a quotient metric graphs which together with Proposition \ref{prop: Algebraic properties}
and Theorem \ref{thm:isospectral_finite_dim_operators} may be used
to provide isospectral examples of metric graphs. With this construction
we may recover earlier isospectral examples such as the ones in \cite{GutSmi_jpa01,BanShaSmi_jpa06,BanSawSmi_jpa10}.
Furthermore, there are quite a few very interesting recent works on
isospectrality of metric graphs \cite{Pistol_arXiv23,KurMul_arXiv21,SteLlePos_arXiv22,Mutlu_cpaa21,JevLip_polonica21,ChePiv_ieot20,MugPiv_jpa23,LawSawBiaSir_Scirep21};
among which the papers \cite{Mutlu_cpaa21,JevLip_polonica21} are
actually based on the quotient construction presented here (these
papers cite an earlier version of the present work). A natural question
is which of the isospectral examples in those recent works may be
reproduced using the theory presented here. We leave the thorough
examination of this question and comparison of those isospectral methods
to future works. Here we briefly explain how some of the isospectral
constructions in \cite{SteLlePos_arXiv22,KurMul_arXiv21} can be deduced
from the theory presented in the current paper.

The paper \cite{SteLlePos_arXiv22} presents a construction method
of isospectral graphs (both discrete and metric and including magentic
fields). The heart of the isospectral method lies in theorem 5.7 of
that paper (see also (3.1) there). This theorem presents a spectral
decomposition of a special type of graphs. The graphs there are constructed
from a basic building block which is duplicated according to a certain
partition of some number $r\in\N$. Specifically, the partition is
denoted by $A=(a_{1},a_{2},\ldots,a_{s})$, where $r=\sum_{i=1}^{s}a_{i}$,
and the corresponding graph is denoted there by $F_{A,V_{1}}$. In
\cite[thm 5.7]{SteLlePos_arXiv22} the spectrum of this graph is shown
to be equal to a union of spectra of some other smaller graphs. This
spectral decomposition depends only on $r$, $r-s$ and the building
block graph. Hence, any other partition $A'$ of $r$ with the same
value of $s$ would yield an isospectral graph $F_{A',V_{1}}$. It
can be shown that the spectral decomposition of \cite[Theorem 5.7]{SteLlePos_arXiv22}
may be obtained directly from Proposition \ref{Prop:Spectral_property},(\ref{enu:Prop_Spectral_property_2}).
To give more details, the graph $F_{A,V_{1}}$ is symmetric under
the group $C_{a_{1}}\times C_{a_{2}}\times\ldots\times C_{a_{s}}$,
where $A=(a_{1},a_{2},\ldots,a_{s})$ is the given partition and $C_{n}$
is the cyclic group of order $n$. Apart from the trivial representation,
the group $C_{n}$ has an $n-1$ irreducible representations given
by the non-trivial $n$-roots of unity. Taking all quotients with
respect to those representations as in Proposition \ref{Prop:Spectral_property},(\ref{enu:Prop_Spectral_property_2})
gives a spectral decomposition of $F_{A,V_{1}}$, which is very close
to the one in \cite[thm 5.7]{SteLlePos_arXiv22}. To get exactly the
spectral decomposition in \cite[thm 5.7]{SteLlePos_arXiv22} we need
just to notice that the quotient of $F_{A,V_{1}}$ with respect to
the trivial representation of $C_{a_{1}}\times C_{a_{2}}\times\ldots\times C_{a_{s}}$
is a graph which is still symmetric, but under the symmetry group
$C_{s}$. Further applying Proposition \ref{Prop:Spectral_property},(\ref{enu:Prop_Spectral_property_2})
for that graph yields exactly the spectral decomposition in \cite[thm 5.7]{SteLlePos_arXiv22}.

Incidentally, the isospectral graphs presented in \cite[Example 1]{KurMul_arXiv21}
may be obtained from the above construction with partitions $A=(2,2)$
and $A'=(1,3)$. However, we ought to emphasize that \cite{KurMul_arXiv21}
contains further isospectral construction methods; it is not clear
to us at this point whether those can be reproduced from the theory
presented here.

\subsection{Spectral computations and graph factorization\label{Subsec: Applications - Computations}}

In this section we focus on using the decompositions given in Proposition
\ref{prop: Algebraic properties},(\ref{enu: prop_algebraic_Operator decomposition})
and in its spectral counterpart, Proposition \ref{Prop:Spectral_property},(\ref{enu:Prop_Spectral_property_2}),
to facilitate the computation of spectra of large symmetric graphs
(or any other domain).

\subsubsection*{Spectral decomposition for computations: Chung--Sternberg formula}

Chung and Sternberg use the idea of a quotient operator (although
they do not use this term) as a tool for spectral computations \cite{Chu92}
(see also Section 7.5 in the book of Chung \cite{Chung_spectralgraph}).
Explicitly, they calculate the eigenvalues of the discrete Laplacian
of large symmetric graphs.

The graphs considered in \cite{Chu92} are taken to be symmetric with
respect to a transitive action of a group $G$ on their vertices.
Namely, every vertex can be transformed to every other by some group
element; equivalently, the fundamental domain has only one vertex.
In this case, all vertices are of the same degree. Assign weights
$\{q_{e}\}$ to the edges connected to some vertex (the transitive
action makes the choice of the particular vertex irrelevant) and denote
their sum by $Q=\sum q_{e}$. The operator studied in \cite{Chu92}
is the weighted (discrete) Laplacian 
\begin{equation}
\Op_{i,j}=\Id-\begin{cases}
\nicefrac{q_{e}}{Q} & i\textrm{ is connected to }j~\textrm{by edge}~e\\
0 & i\textrm{ is not connected to }j
\end{cases}.\label{eq:Operator-of-Chung}
\end{equation}
Note that (\ref{eq:normalized_Lap_k_reg}) is a particular case of
(\ref{eq:Operator-of-Chung}), with $q_{e}=1$ for all edges. Assume
that $\Op$ is $G$-symmetric. Given an irreducible representation
$\rho$, we may apply our formula (\ref{eq:Quotient_formula_explicit})
to get the explicit form of the quotient operator. We choose $i$
to be a representative vertex, so that $\D=\{i\}$ and in this case
(\ref{eq:Quotient_formula_explicit}) reads 
\begin{equation}
\Op_{\rho}=\frac{1}{|G_{i}|}\sum_{g\in G}\left(\Phi_{i}^{*}\,\rho(g)\,\Phi_{i}\right)\Op_{i,gi},\label{eq:Quotient_formula_Chung_ours}
\end{equation}
with all notations similar to those of Definition \ref{def: matrix_form_of_quotient}.
Namely, $\Phi_{i}$ is an $r\times d_{i}$ matrix whose columns form
an orthonormal basis for $(\Vr)^{G_{i}}$, the $G_{i}$-invariant
subspace of $\Vr$.

We note that (\ref{eq:Quotient_formula_Chung_ours}) is precisely
the expression in equation (11) of \cite{Chu92}. To see this, we
provide the following dictionary translating from our formula (\ref{eq:Quotient_formula_explicit})
to \cite[eq. (11)]{Chu92}:
\begin{itemize}
\item $\left|G_{i}\right|=\left|H\right|=\left|C_{e}\right|$
\item $\Vr\cong W$,~ $(\Vr)^{G_{i}}\cong W^{H}$
\item $g=a$, and $\Phi_{i}^{*}\,\rho(g)\,\Phi_{i}$ is ``$a$ evaluated
on $W^{H}$''.
\end{itemize}
To conclude, the quotient expression in \cite[eq. (11)]{Chu92} may
be viewed as an application of Definition \ref{def: matrix_form_of_quotient}
for the particular case of a discrete Laplacian with a transitive
action on its vertex set. Chung and Sternberg combine this with a
matrix decomposition such as the one in Proposition \ref{prop: Algebraic properties},(\ref{enu: prop_algebraic_Operator decomposition})
to present the eigenvalues of a large graph as the union of the eigenvalues
of its quotients. We hope that the generalized theory provided in
the current paper would aid in further enhancement of eigenvalue computations
in the spirit of \cite{Chu92}.

\subsubsection*{Graph factorizations}

A similar method of graph factorization appears in the classical book
of Cvetkovi\'{c}, Doob and Sachs \cite[Ch. 5.2]{CveDooSac_spectra_of_graphs_book}.
They discuss the restriction of an operator on a graph to a particular
irreducible representation. However, the reduction is to a system
of linear equations rather than providing an explicit expression for
the quotient operator (such as (\ref{eq:Quotient_formula_explicit})).
The discussion in \cite{CveDooSac_spectra_of_graphs_book} mentions
the spectral point of view of such factorization (similarly to \cite{Chung_spectralgraph,Chu92}),
but lacks the algebraic aspect such as Proposition \ref{prop: Algebraic properties}.

Another useful concept appearing in that book is the one of graph
divisors \cite[Ch. 4]{CveDooSac_spectra_of_graphs_book}. In the literature
a graph divisor is also known by other names: an equitable partition,
a coloration, a quotient graph or an orbigraph \cite{Schwenk_proc74,PowSul_laa82,Brooks_aif99,DalGavMonOchStaSte_involve19}.
Using our terminology, a graph divisor is the quotient of an operator
with respect to the trivial representation of a symmetry group. That
the spectrum of a divisor is contained in the spectrum of the original
graph is widely known, appearing for example in the books \cite[Ch. 4]{CveDooSac_spectra_of_graphs_book},
\cite[Ch. 3.9]{CveRowSim_An_introduction_to_the_theory_of_graph_spectra_2010},
\cite[Ch. 5]{Godsil_algebraic_combinaotrics_1993}, \cite[Ch. 9.3]{GodRoy_algebraic_graph_theory_book_2001}.
It is mentioned in \cite[Ch. 4]{CveDooSac_spectra_of_graphs_book}
that the spectrum of a graph may be decomposed as the union of spectra
of its divisor and codivisor, but without providing a general and
explicit formula for the codivisor (such as, for instance, our Proposition
\ref{prop: Algebraic properties},(\ref{enu:   prop_algebraic_Operator decomposition})).

This divisor--codivisor decomposition is further developed in a series
of recent works \cite{BarFraWeb_laa17,FraSmiSorWeb_laa17,FraSmiWeb_laa19}
in the special case of $G$ being a cyclic group. In this sense, those
works contain results which are particular cases of Proposition \ref{prop:   Algebraic properties},(\ref{enu: prop_algebraic_Operator   decomposition}).

\subsubsection*{Spectral decompositions on finite and infinite quantum (metric) graphs}

Spectral decompositions for metric graphs in the form of Proposition~\ref{Prop:Spectral_property}
are used for various purposes in spectral analysis. For instance,
\cite{Exner_ppn21,ExnLip_jmp19} studied the spectra of graphs formed
by the edges of platonic solids. Naturally, these graphs are symmetric
with respect to well-known groups. Hence, spectral computations in
\cite{Exner_ppn21,ExnLip_jmp19} (which cite an earlier version of
the present work) are made mush easier by a decomposition of the type
established in Proposition~\ref{Prop:Spectral_property}.

A decomposition of the Laplacian on radially symmetric metric trees
by Naimark--Solomyak \cite{NaiSol_rjmp01,Solomyak_wrm04} was very
influential in the modern spectral theory of non-compact metric graphs.
More recently, in the series of works \cite{BreKel_opem13,BreKel_opem13-err,BreLev_ahp20,KosNic_jst21},
the decomposition of Naimark--Solomyak was extended to metric and
discrete non-compact graphs with partial spherical symmetry.

We expect that this decomposition can also be understood in the spirit
of Proposition \ref{Prop:Spectral_property},(\ref{enu:Prop_Spectral_property_2}),
but for non-compact metric graphs and infinite symmetry groups (which
are infinite products of finite groups). The works \cite{BreKel_opem13,BreKel_opem13-err,BreLev_ahp20}
do not explicitly describe the symmetry groups of the underlying graphs
and there is still work to be done to link the two approaches.

Finally, we mention that the results presented here were used to significantly
simplify computations of the spectrum of a Euler--Bernoulli beam
structure in \cite{BerEtt_sam22}. The model therein is a fourth order
differential operator acting on vector-valued functions supported
on a symmetric graph. While the graph and its symmetry groups are
simple, the high order of the operator makes analysis without spectral
decomposition prohibitively cumbersome.


\section*{Acknowledgments}

We thank Rostislav Grigorchuk, Maxim Gurevich, Peter Kuchment, Ji\v{r}\'{\i}
Lipovsk\a'{y}, Delio Mugnolo and Ori Parzanchevski for their critical
comments and friendly encouragement. RB and GB were supported by the
Binational Science Foundation Grant (Grant No. 2024244). RB was supported
by ISF (Grant No. 844/19). GB was partially supported by National
Science Foundation grants DMS-1410657 and DMS-2247473. CHJ would like
to thank the Leverhulme Trust for financial support (ECF-2014-448).


\appendix

\section{A primer on representation theory \label{sec: appendix-representation_theory}}

We bring here the relevant definitions and statements from representation
theory of finite groups. In particular, we make the connection between
group module terminology to matrix representation terminology. We
are aware that there are many who are not familiar with the former,
but are very fond of the latter.

\subsection{Representations and $G$-modules}
\begin{defn}
\label{def: representation_and_module} Let $G$ be a finite group.
Let $V$ be a vector space over $\C$.
\end{defn}

\begin{enumerate}
\item \label{enu:def-representation_and_module_1}A group homomorphism,
$\rho:G\rightarrow\GL(V)$ is called a representation of $G$.
\item \label{enu:def-representation_and_module_2}The group algebra $\C G$
consists of the elements $\sum_{g\in G}c_{g}g$ with $c_{g}\in\C$,
and multiplication in $\C G$ extends the multiplication in $G$.
\\
We say that $V$ is a (left) $\C G$-module (or $G$-module) if there
is an action of $\C G$ on $V$ (i.e., a map $\C G\times V\rightarrow V$)
denoted by $r\cdot v$ for all $r\in\C G$, $v\in V$, such that the
following holds:
\begin{enumerate}
\item $(r+s)\cdot v=r\cdot v+s\cdot v$ for all $r,s\in\C G$, $v\in V$.
\item $(rs)\cdot v=r\cdot(s\cdot v)$ for all $r,s\in\C G$, $v\in V$.
\item $r\cdot(v+u)=r\cdot v+r\cdot u$ for all $r\in\C G$, $v,u\in V$.
\item $1\cdot v=v$ for all $v\in V$, where $1\in G$ is the identity element.
\end{enumerate}
\end{enumerate}
\begin{rem}
We will call $V$ as in Definition \ref{def: representation_and_module},(\ref{enu:def-representation_and_module_2})
a $G$-module rather than a $\C G$-module. Both terms are interchangeably
used in the literature (see \cite{Fulton-2013,DummitFoote_abstract_algebra}).
\end{rem}

The two parts of Definition \ref{def: representation_and_module}
are the same. Namely, if $\rho:G\rightarrow\GL(V)$ is a representation
of $G$, then $V$ becomes a $G$-module by

\[
\forall v\in V,\quad\left(\sum_{g\in G}c_{g}g\right)\cdot v=\sum_{g\in G}c_{g}\thinspace\rho(g)\thinspace v.
\]
Conversely, if $V$ is a $G$-module, then $\rho:G\rightarrow\GL(V)$
defined by $\rho(g)\thinspace v=g\cdot v$ is a representation. See
more in \cite[sec.~34]{Bump_LieGroups_book}.

In light of the above, we tend to denote $\Vr$ to emphasize the connection
between the module and the corresponding representation.

To develop and prove the theory in this paper it is more convenient
to use the module terminology rather than the matrix representation
one. On the other hand, working out the examples requires the use
of specific matrix representations. Hence, both approaches are mentioned
throughout the paper.


\subsection{Intertwiners\label{subsec:Appendix-Intertwiners}}

Let $G$ be a finite group and $\Vr$ and $\Vp$ be two $G$-modules.
Denote by $\mathrm{Hom}(\Vr,\Vp)\cong\Vr^{*}\otimes\Vp$ the space
of linear homomorphisms between the vector spaces $\Vp$ and $\Vr$.
We consider $\mathrm{Hom}(\Vr,\Vp)$ as a $G$-module, with the following
action, $g\times\phi\mapsto g\phi g^{-1}$ for all $g\in G$ and $\phi\in\mathrm{Hom}(\Vr,\Vp)$.
The set of fixed points of this action is denoted by 
\begin{equation}
\Hom(\Vr,\Vp):=\set{\phi\in\mathrm{Hom}(\Vr,\Vp)}{g\phi=\phi g}.\label{eq: hom_space_definition-1}
\end{equation}
This is a linear space, known as the trivial isotypic component of
$\mathrm{Hom}(\Vr,\Vp)$ and it is also common to refer to the elements
$\phi\in\Hom(\Vr,\Vp)$ as \emph{intertwiner}s.\\

In this paper we consider only finite dimensional $G$-modules and
assume that they are equipped with an inner product (hence, they are
Hilbert spaces). We also define the following inner product on the
space of intertwiners,

\begin{equation}
\left\langle \phi_{1},\phi_{2}\right\rangle _{\Hom(\Vr,\Vp)}:=\Tr\left(\phi_{2}^{*}\phi_{1}\right)=\sum_{v\in B(\Vr)}\left\langle \phi_{1}(v)\thinspace,\thinspace\phi_{2}(v)\right\rangle _{\Vp},\label{eq:inner_product_on_intertwiner_space}
\end{equation}
where $\left\langle \phantom{}~,~\phantom{}\right\rangle _{\Vp}$
denotes the inner product in $\Vp$ and $B(\Vr)$ is some choice of
basis for $\Vr$ (obviously, it can be shown that the inner product
above does not depend on the choice of basis). This is just the usual
Frobenius inner product on the space of finite dimensional linear
operators, restricted to $\Hom(\Vr,\Vp)$.


\subsection{Induced representations and Frobenius reciprocity\label{subsec:Induced-representations-appendix}}

There is more than one standard way to present induced representations
and Frobenius reciprocity. Here we follow the one in \cite[sec.~34]{Bump_LieGroups_book},
which suits the needed arguments for the proofs in Section \ref{sec:quotient_operators}.
\begin{defn}
\label{def: induced_representation} Let $H$ be a subgroup of a finite
group $G$. Let $\Vs$ be an $H$-module. Define the induced representation
(from $H$ to $G$), $\Vr:=\Ind HG{\Vs}$, as the following $G$-module,
\begin{equation}
\Vr=\set{f:G\rightarrow\Vs}{f(hg)=\sigma(h)f(g),~~\forall h\in H},\label{eq:definition_of_induced_representation}
\end{equation}
where the $G$ action is by $\left(\rho(g)f\right)(\tilde{g}):=f(\tilde{g}g)$
for all $f\in\Vr$ and $g,\tilde{g}\in G$.
\end{defn}

On the induced representation, $\Ind HG{\Vs}$, we fix the following
inner product:

\begin{equation}
\forall f_{1},f_{2}\in\Ind HG{\Vs},\quad\quad\left\langle f_{1},f_{2}\right\rangle :=\frac{1}{\left|G\right|}\sum_{g\in G}\left\langle f_{1}(g)\thinspace,\thinspace f_{2}(g)\right\rangle _{\Vs}.\label{eq:inner_product_induced_rep}
\end{equation}

As a straightforward exercise, the reader might compute that when
taking $H=\{\mathrm{id}\}$ (the trivial group) and $\Vs$ to be the
trivial representation of $H$, we obtain in the definition above
$\Vr=\C G$ and so $\Ind HG{\Vs}=\reg_{G}$, where $\reg_{g}$ is
the regular representation of $G$. This observation is used in the
proof of Proposition \ref{prop: Algebraic properties},(\ref{enu: prop_algebraic_Operator decomposition}).
\begin{thm}
\label{thm:Frobenius_reciprocity} {[}Frobenius reciprocity Theorem{]}
\cite[sec.~34]{Bump_LieGroups_book}

Let $H$ be a subgroup of a finite group $G$. Let $\Vs$ be an $H$-module
and $\Vp$ be a $G$-module. Then
\begin{enumerate}
\item \label{enu: thm-Frobenius-reciprocity-1} $\mathrm{Hom}_{H}\left(\Vp,\Vs\right)$
is isomorphic to $\Hom\left(\Vp,\Ind HG{\Vs}\right)$.\\
 This isomorphism maps $\phi\in\mathrm{Hom}_{H}\left(\Vp,\Vs\right)$
to $\Phi\in\Hom\left(\Vp,\Ind HG{\Vs}\right)$ defined by 
\begin{equation}
\Phi(v)(g)=\phi(\pi(g)v).\label{eq: Frobenius_map_relation-1}
\end{equation}
for $v\in\Vp$ and $g\in G$.
\item \label{enu: thm-Frobenius-reciprocity-2} $\mathrm{Hom}_{H}\left(\Vs,\Vp\right)$
is isomorphic to $\Hom\left(\Ind HG{\Vs},\Vp\right)$.\\
 This isomorphism maps $\phi\in\mathrm{Hom}_{H}\left(\Vs,\Vp\right)$
to $\Phi\in\Hom\left(\Ind HG{\Vs},\Vp\right)$ defined by 
\begin{equation}
\Phi:\Ind HG{\Vs}\to\Vp,\qquad\Phi:f\mapsto\frac{\left|H\right|}{\left|G\right|}\sum_{g\in G/H}\pi(g)\thinspace\phi(f(g^{-1})),\label{eq: Frobenius_map_relation-2}
\end{equation}
where the sum is over some choice of representatives of $H$ left
cosets in $G$ (and does not depend on the particular choice of the
representatives).
\end{enumerate}
\end{thm}

\begin{rem*}
For any two $G$-modules, $\Vs,\Vp$ it is not hard to see that $\mathrm{Hom}_{H}\left(\Vp,\Vs\right)\cong\mathrm{Hom}_{H}\left(\Vs,\Vp\right)$,
where the isomorphism is explicitly given by the adjoint. Hence, the
two parts of the theorem above are equivalent in the sense of existence
of isomorphisms. Yet, we bring both versions (as in \cite[sec.~34]{Bump_LieGroups_book})
since the explicit form of the isomorphisms is useful for the proofs
in the paper. The first part of the theorem is used in proving Theorem~\ref{thm: Quotient     fundamental property}
and the second is used in the proof of Proposition \ref{prop: Algebraic properties},(\ref{enu:     prop_algebraic_subgroup_quotient}).
\end{rem*}
\begin{rem*}
In (\ref{eq: Frobenius_map_relation-2}) above we have introduced
a prefactor as part of Frobenius maps which does not appear in \cite[sec. 34]{Bump_LieGroups_book}.
Obviously, such a factor does not affect the map being an isomorphism,
but it is used to show that the map is unitary.
\end{rem*}
In the current paper we need to use that the Frobenius reciprocity
is not only an isomorphism, but that it also preserves the inner product.
As we did not find explicit references to this observation, we complement
Theorem \ref{thm:Frobenius_reciprocity} by stating and proving this.
\begin{prop}
\label{prop: Frobenius_preserves_inner_product} Both isomorphisms
in Theorem \ref{thm:Frobenius_reciprocity} preserve the inner product.
\end{prop}

\begin{rem}
The following proof may be simplified by observing that the two Frobenius
isomorphisms are adjoints of each other and checking unitarity for
only one of them.
\end{rem}

\begin{proof}
We first show that the isomorphism in Theorem \ref{thm:Frobenius_reciprocity},(\ref{enu: thm-Frobenius-reciprocity-1})
preserves the inner product.

Let $\phi_{1},\phi_{2}\in\mathrm{Hom}_{H}(\Vp,\Vs)$ and $\Phi_{1},\Phi_{2}\in\Hom\left(\Vp,\Ind HG{\Vs}\right)$
the corresponding images by the Frobenius reciprocity map (\ref{eq: Frobenius_map_relation-1}).

\begin{align*}
\left\langle \Phi_{1}~,~\Phi_{2}\right\rangle _{\Hom\left(\Vp,\Ind HG{\Vs}\right)} & =\sum_{v\in B(\Vp)}\left\langle \Phi_{1}(v)~,~\Phi_{2}(v)\right\rangle _{\Ind HG{\Vs}}\\
 & =\frac{1}{\left|G\right|}\sum_{v\in B(\Vp)}\sum_{g\in G}\left\langle \Phi_{1}(v)(g)~,~\Phi_{2}(v)(g)\right\rangle _{\Vs}\\
 & =\frac{1}{\left|G\right|}\sum_{v\in B(\Vp)}\sum_{g\in G}\left\langle \phi_{1}(\pi(g)v)~,~\phi_{2}(\pi(g)v)\right\rangle _{\Vs}\\
 & =\frac{1}{\left|G\right|}\sum_{g\in G}\sum_{v\in B(\Vp)}\left\langle \phi_{1}(\pi(g)v)~,~\phi_{2}(\pi(g)v)\right\rangle _{\Vs}\\
 & =\frac{1}{\left|G\right|}\sum_{g\in G}\left\langle \phi_{1}~,~\phi_{2}\right\rangle _{\mathrm{Hom}_{H}(\Vp,\Vs)}=\left\langle \phi_{1}~,~\phi_{2}\right\rangle _{\mathrm{Hom}_{H}(\Vp,\Vs)}
\end{align*}
where moving to the last line we used the fact that for any $g\in G$,
the set $\left\{ \pi(g)\thinspace v\right\} _{v\in B(\Vp)}$ is a
basis for $\Vp$.

Next, we show that the isomorphism in Theorem \ref{thm:Frobenius_reciprocity},(\ref{enu: thm-Frobenius-reciprocity-2})
preserves the inner product.

We define for each $g\in G$ and $v\in\Vs$ the following map 
\begin{align*}
\epsilon_{g}(v) & :G\rightarrow\Vs\\
\epsilon_{g}(v) & (\tilde{g})=\sqrt{\nicefrac{\left|G\right|}{\left|H\right|}}\begin{cases}
\sigma(\tilde{g}g)v & \textrm{if}~\tilde{g}g\in H,\\
0 & \textrm{otherwise}.
\end{cases}
\end{align*}
One can check that for all $g\in G$ and $v\in\Vs$, $\epsilon_{g}(v)\in\Ind HG{\Vs}$.
We use the map $\epsilon_{g}(v)$ to construct a basis for $\Ind HG{\Vs}$
from a given basis for $\Vs$. Let $\left\{ v_{i}\right\} $ some
choice of an orthonormal basis for $\Vs$ and $\left\{ g_{j}\right\} $
some choice of representatives for the right cosets of $H$. Then,
it is straightforward to check that $\left\{ \epsilon_{g_{j}}(v_{i})\right\} $
is an orthonormal basis for $\Ind HG{\Vs}$ (with the inner product
(\ref{eq:inner_product_induced_rep})). Take $\phi\in\mathrm{Hom}_{H}(\Vs,\Vp)$
which is sent by the Frobenius isomorphism to $\Phi\in\mathrm{Hom}_{G}(\Ind HG{\Vs},\Vp)$.
This map is given explicitly in (\ref{eq: Frobenius_map_relation-2}),
so that 
\begin{align}
\forall i,j,\quad\quad\Phi(\epsilon_{g_{j}}(v_{i})) & =\nicefrac{\left|H\right|}{\left|G\right|}\thinspace\sum_{g\in G/H}\pi(g)\thinspace\phi(\epsilon_{g_{j}}(v_{i})(g^{-1}))\nonumber \\
 & =\sqrt{\nicefrac{\left|H\right|}{\left|G\right|}}\thinspace\pi(\tilde{g})\thinspace\phi(\sigma(\tilde{g}^{-1}g_{j})v_{i})\nonumber \\
 & =\sqrt{\nicefrac{\left|H\right|}{\left|G\right|}}\thinspace\pi(g_{j})\thinspace\phi(v_{i}),\label{eq:Phi_phi_relation}
\end{align}
where in the second line we take $\tilde{g}\in G/H$ to be the $H$
left coset representative such that $g_{j}\in\tilde{g}H$, and on
the third line we used that $\phi$ is an intertwiner. Applying (\ref{eq:Phi_phi_relation}),
we show that the Frobenius isomorphism preserves the inner product.
Indeed, take $\phi_{1},\phi_{2}\in\mathrm{Hom}_{H}(\Vs,\Vp)$ which
are sent by the Frobenius isomorphism to $\Phi_{1},\Phi_{2}\in\mathrm{Hom}_{G}(\Ind HG{\Vs},\Vp)$,
correspondingly, and get 
\begin{align*}
\left\langle \Phi_{1},\Phi_{2}\right\rangle  & =\sum_{i,j}\left\langle \Phi_{1}(\epsilon_{g_{j}}(v_{i}))\thinspace,\thinspace\Phi_{2}(\epsilon_{g_{j}}(v_{i}))\right\rangle _{\Vp}\\
 & =\frac{\left|H\right|}{\left|G\right|}\sum_{i,j}\left\langle \pi(g_{j})\thinspace\phi_{1}(v_{i})\thinspace,\thinspace\pi(g_{j})\thinspace\phi_{2}(v_{i})\right\rangle _{\Vp}\\
 & =\sum_{i}\left\langle \phi_{1}(v_{i})\thinspace,\thinspace\phi_{2}(v_{i})\right\rangle _{\Vp}\\
 & =\left\langle \phi_{1},\phi_{2}\right\rangle ,
\end{align*}
where on the second line we used the unitarity of the representation
$\pi$.
\end{proof}
\begin{prop}
\label{prop:Frobenius_commutes} If $T:\Vp\to\Vp$ is a $G$-symmetric
operator (as in Definition \ref{def:sym_operator}), then in commutes
with the Frobenius isomorphism from $\mathrm{Hom}_{H}\left(\Vs,\Vp\right)$
to $\Hom\left(\Ind HG{\Vs},\Vp\right)$ given by \eqref{eq: Frobenius_map_relation-2}.
\end{prop}

\begin{proof}
For all $f\in\Ind HG{\Vs}$ we check 
\begin{multline*}
\Op\left[\sum_{g\in G/H}\pi(g)\thinspace\phi(f(g^{-1}))\right]=\sum_{g\in G/H}\Op\left[\pi(g)\thinspace\phi(f(g^{-1}))\right]=\sum_{g\in G/H}\pi(g)\thinspace\Op\left[\phi(f(g^{-1}))\right]\\
=\sum_{g\in G/H}\pi(g)\thinspace\Op\left[\phi\right](f(g^{-1})),
\end{multline*}
where the first equality follows by linearity, the second holds since
$\Op$ is $G$-symmetric and the last is the definition of $\Op$
action on $\mathrm{Hom}_{H}(\Vs,\Vp)$, see (\ref{eq:T_acts_on_Hom_part1}).
To summarize, we get that $T\Phi$ is the element of $\Hom\left(\Ind HG{\Vs},\Vp\right)$
corresponding to $T\phi\in\mathrm{Hom}_{H}\left(\Vs,\Vp\right)$ by
the isomorphism \eqref{eq: Frobenius_map_relation-2}.
\end{proof}

\section{Proofs of the properties of the quantum graph quotient}

\label{app:qg_quotient}


\subsection{Preliminaries}

\label{app:qg_quotient_prelim}

The vectorization map is a linear transformation which converts a
matrix to a column vector by essentially `stacking' the columns on
the matrix on top of each other. More formally
\begin{defn}
\label{def: Vectorization}Let $A$ be an $m\times n$ matrix with
$m$-dimensional column vectors $a_{1},\ldots,a_{n}$ then $\vec{~}{:}~M_{m\times n}(\C)\to\C^{nm}$
is the following map 
\[
\vect(A)=\vect((a_{1}\ldots a_{n}))=\begin{pmatrix}a_{1}\\
\vdots\\
a_{n}
\end{pmatrix}=\sum_{i=1}^{n}(\ei\otimes a_{i}),
\]
which we call \emph{vectorization}.
\end{defn}

So, for example, if $\phi=(\phi_{1},\phi_{2})$ is the $2\times2$
matrix, then 
\[
\phi_{1}=\begin{pmatrix}x_{1}\\
y_{1}
\end{pmatrix},\quad\phi_{2}=\begin{pmatrix}x_{2}\\
y_{2}
\end{pmatrix},\quad\Longrightarrow\vect(\phi)=\begin{pmatrix}x_{1}\\
y_{1}\\
x_{2}\\
y_{2}
\end{pmatrix}.
\]

Using the definition of $\vect$ we have the following trivial property
for products of matrices.
\begin{lem}
\label{Lem: Vectorization} Let $A$ and $B$ be two $k\times n$
and $n\times m$ matrices respectively. Then 
\begin{equation}
\vect(AB)=(\Id_{m}\otimes A)\thinspace\vect(B)=(B^{T}\otimes\Id_{k})\thinspace\vect(A).\label{eq:vectorization_identity}
\end{equation}
\end{lem}

Next, we define the matrix $\Theta$ whose columns are the orthonormal
basis $\bigcup_{i\in\D}\left\{ \theta_{i}^{(n)}\right\} _{n=1}^{d_{i}}$
for $\Hom(\Vr,\Vp)$, as given in Lemma \ref{lem: orthonormal basis for Hom_G}.
\begin{defn}
\label{def: Theta matrix} Following the notation in Section~\ref{subsec:Quotient-QG},
we denote for each $i\in\D$ and $1\le i\leq d_{i}$, the map $\theta_{i}^{(n)}:\Vr\rightarrow\Vp$
by 
\begin{equation}
\forall v\in\Vr,\quad\quad\theta_{i}^{(n)}(v):=\frac{1}{\sqrt{\left|G\right|\left|G_{i}\right|}}\sum_{g\in G}\left\langle \rho(g)\varphi_{i}^{(n)}\thinspace,\thinspace v\right\rangle _{\Vr}\thinspace\mathbf{e}_{g\thinspace i},\label{eq:definition_of_abstract_thetas - to use for Theta matrix}
\end{equation}
exactly as in Lemma \ref{lem: orthonormal basis for Hom_G}. With
$r:=\dim\Vr$ and $\dim\Vp=\left|\E\right|$, consider each $\theta_{i}^{(n)}$
as an $\left|\E\right|\times r$ matrix, so that $\vec{(}\theta_{i}^{(n)})$
is a column vector of length $\left|\E\right|r$.\\
 Specifically, we define 
\begin{equation}
\Theta_{i}=\frac{1}{\sqrt{\left|G\right|\left|G_{i}\right|}}\sum_{g\in G}\left(\overline{\rho(g)}\otimes\pi(g)\right)\left(\Phi_{i}\otimes\ei\right),\label{eq: Theta_i expressed with Phi_i}
\end{equation}
to be a $\left|\E\right|r\times d_{i}$ matrix whose columns are $\left\{ \vec{(}\theta_{i}^{(n)})\right\} _{n=1}^{d_{i}}$.
Subsequently, define $\Theta\in M_{\left|\E\right|r\times d}(\C)$
to be 
\begin{equation}
\Theta=\begin{pmatrix}\Theta_{1} & \Theta_{2} & \ldots & \Theta_{\left|\D\right|}\end{pmatrix}.\label{eq: Theta expressed at many Theta_i}
\end{equation}

We also denote $\hat{\Theta}:=\Theta\otimes\Id_{2}$ (similarly to
the $\hat{\pi}$ notation introduced in~\eqref{eq:pi-invariance_def}),
\end{defn}

\begin{lem}
\label{lem: properties of Theta} The matrices $\Theta\in M_{\left|\E\right|r\times d}(\C)$
and $\hat{\Theta}\in M_{2\left|\E\right|r\times2d}(\C)$ defined above
satisfy
\begin{enumerate}
\item \label{enu: Theta property - orthogonality}$\Theta^{*}\Theta=\Id_{d}$
.
\item \label{enu: Theta property - projection}$\Theta\Theta^{*}=\frac{1}{|G|}\sum_{g\in G}\cc{\rho(g)}\otimes\pi(g)$.
\item \label{enu: Theta property - commuting with A B}$\hat{\Theta}\hat{\Theta}^{*}$
commutes with $\Id_{r}\otimes A$ and $\Id_{r}\otimes B$.
\item \label{enu: Theta property - quotient A B}$A_{\rho}=\hat{\Theta}^{*}[\Id_{r}\otimes A]\hat{\Theta}$
and $B_{\rho}=\hat{\Theta}^{*}[\Id_{r}\otimes B]\hat{\Theta}$.
\item \label{enu: Theta property - intertwining}$[\Id_{r}\otimes\pi(g)-\rho(g)^{T}\otimes\Id_{p}]\Theta=0$
for all $g\in G$.
\end{enumerate}
\end{lem}

\begin{proof}
The proof of the first two parts of the Lemma are based on Lemma \ref{lem: orthonormal basis for Hom_G}.
In Lemma \ref{lem: orthonormal basis for Hom_G} it is stated that
the maps $\theta_{i}^{(n)}:\Vr\rightarrow\Vp$ as defined in (\ref{eq:definition_of_abstract_thetas - to use for Theta matrix})
form an orthonormal basis for $\Hom(\Vr,\Vp)$. Turning those maps
to matrices and applying $\vec{\ensuremath{}}$ do not affect this
quality. Hence we get that the columns of $\Theta$ form an orthonormal
basis for $\Hom(\Vr,\Vp)$.
\begin{enumerate}
\item The orthonormality of $\Theta$ columns yields $\Theta^{*}\Theta=\Id_{2d}$,
as in the first part of the Lemma.
\item That the columns form an orthogonal \emph{basis} for $\Hom(\Vr,\Vp)$
implies that $\Theta\Theta^{*}$ is a projector from $\Vr\otimes\Vp=\mathrm{Hom}(\Vr,\Vp)$
onto $\Hom(\Vr,\Vp)$. To get the explicit expression of this projector,
we note (see, for example, \cite[Ch. 2.2]{Fulton-2013}) that if $R$
is a unitary representation of $G$ on $V$ then the operator 
\[
\Pb:=\frac{1}{|G|}\sum_{g\in G}R(g)
\]
is an orthogonal projector onto the trivial component of $R$, i.e.
\[
\Image[\Pb]=\{v\in V_{R}~:~R(g)v=v~\quad\forall g\in G\}.
\]
Taking $R=\cc{\rho}\otimes\pi$, where $\cc{\rho}$ is the dual representation
of $\rho$ (see also \cite[Ch. 1.1]{Fulton-2013}), we get that $\frac{1}{|G|}\sum_{g\in G}\cc{\rho(g)}\otimes\pi(g)$
is an orthogonal projector onto the trivial component of $\Vr\otimes\Vp=\mathrm{Hom}(\Vr,\Vp)$
which is $\Hom(\Vr,\Vp)$ (see \cite[Ex. 1.2]{Fulton-2013}). Therefore,
$\Theta\Theta^{*}=\frac{1}{|G|}\sum_{g\in G}\cc{\rho(g)}\otimes\pi(g)$,
as in the second part of the Lemma.
\end{enumerate}
We continue to prove the last three parts of the Lemma.

{[}resume{]}
\begin{enumerate}
\item Using the second part of the Lemma we get that 
\begin{equation}
\hat{\Theta}\hat{\Theta}^{*}=\frac{1}{|G|}\sum_{g\in G}\cc{\rho(g)}\otimes\hat{\pi}(g).\label{eq:JJstar_hat}
\end{equation}
This clearly commutes with the matrices $\Id_{r}\otimes A$ and $\Id_{r}\otimes B$
because $A$ and $B$ commute with $\hat{\pi}(g)$ (as is assumed
throughout the section and see also Lemma \ref{lem:invariant_AB}).
\item The assertions $A_{\rho}=\hat{\Theta}^{*}[\Id_{r}\otimes A]\hat{\Theta}$
and $B_{\rho}=\hat{\Theta}^{*}[\Id_{r}\otimes B]\hat{\Theta}$ are
obtained as a corollary of Theorem \ref{thm: Quotient fundamental property},
with the slight modification of considering $\hat{\pi}$ instead of
$\pi$. To see this, we first observe that $\left.A\right|_{\Hom(\Vr,\mathrm{V}_{\hat{\pi}})}=\hat{\Theta}^{*}[\Id_{r}\otimes A]\hat{\Theta}$.
This last observation follows from the linear action of $A$ on $\Hom(\Vr,\mathrm{V}_{\hat{\pi}})$
(see (\ref{eq:T_acts_on_Hom_part1}),(\ref{eq:T_acts_on_Hom_part2})),
when keeping in mind that the columns of $\hat{\Theta}$ forming an
orthonormal basis for $\Hom(\Vr,\mathrm{V}_{\hat{\pi}})$ (as in the
first two parts of this lemma). Now, Theorem \ref{thm: Quotient fundamental property}
implies $A_{\rho}=\hat{\Theta}^{*}[\Id_{r}\otimes A]\hat{\Theta}$
when taking $\Op$ to be $A$ and $\Vp$ to be $\mathrm{V}_{\hat{\pi}}$.
Exactly the same holds for $B$ and $B_{\rho}$.
\item Let $g\in G$. We show that $[\Id_{r}\otimes\pi(g)-\rho(g)^{T}\otimes\Id_{p}]\Theta\Theta^{*}=0$,
which after right multiplication by $\Theta$ and applying the first
part of the Lemma gives the needed statement. 
\begin{align*}
[\Id_{r}\otimes\pi(g)-\rho(g)^{T}\otimes\Id_{p}]\Theta\Theta^{*} & =\frac{1}{|G|}[\Id_{r}\otimes\pi(g)-\rho(g)^{T}\otimes\Id_{p}]\sum_{g'\in G}\cc{\rho(g')}\otimes\pi(g')\\
 & =\frac{1}{|G|}\sum_{g'\in G}\cc{\rho(g')}\otimes\pi(gg')-\frac{1}{|G|}\sum_{g'\in G}\cc{\rho(g^{-1})}\cc{\rho(g')}\otimes\pi(g')\\
 & =\frac{1}{|G|}\sum_{\tilde{g}\in G}\cc{\rho(g^{-1}\tilde{g})}\otimes\pi(\tilde{g})-\frac{1}{|G|}\sum_{g'\in G}\cc{\rho(g^{-1}g')}\otimes\pi(g')=0,
\end{align*}
where in the first line we used the second part of the Lemma and in
the second line we use the unitarity of the representation, $\rho(g)^{T}=\overline{\rho(g^{-1})}$.
\end{enumerate}
\end{proof}
\begin{rem*}
One could also prove some parts of the Lemma above by brute force
computations using the expressions (\ref{eq: Theta_i expressed with Phi_i}),(\ref{eq: Theta expressed at many Theta_i}).
\end{rem*}

\subsection{Proof of Theorem \ref{thm:commutative_diagram-QG}.}

\label{app:qg_quotient_proof}
\begin{lem}
\label{lem:valid_sa} If $A$ and $B$ define a self-adjoint quantum
graph then $A_{\rho}$ and $B_{\rho}$ satisfy
\begin{enumerate}
\item \label{Item: Full rank} $\Rank(A_{\rho}|B_{\rho})=2d$
\item \label{Item: Self-adjoint} $A_{\rho}B_{\rho}^{*}$ is self-adjoint.
\end{enumerate}
\end{lem}

\begin{proof}
We begin by computing 
\begin{align}
\Rank(A_{\rho}|B_{\rho}) & =\Rank\left(\hat{\Theta}^{*}[\Id_{r}\otimes A]\hat{\Theta}\Big|\hat{\Theta}^{*}[\Id_{r}\otimes B]\hat{\Theta}\right)\label{eq:rank_eval1}\\
 & =\Rank\left[\left(\hat{\Theta}^{*}[\Id_{r}\otimes A]\hat{\Theta}\Big|\hat{\Theta}^{*}[\Id_{r}\otimes B]\hat{\Theta}\right)(\Id_{2}\otimes\hat{\Theta}^{*})\right],\nonumber 
\end{align}
where we used Lemma \ref{lem: properties of Theta},(\ref{enu: Theta property - quotient A B})
and that the matrix 
\begin{equation}
\Id_{2}\otimes\hat{\Theta}^{*}=\begin{pmatrix}\hat{\Theta}^{*} & 0\\
0 & \hat{\Theta}^{*}
\end{pmatrix}\label{eq:twice_Theta_star}
\end{equation}
is of full rank (it has $4d$ independent rows). We continue, applying
parts (\ref{enu: Theta property - orthogonality}) and (\ref{enu: Theta property - commuting with A B})
of Lemma \ref{lem: properties of Theta}, 
\begin{align}
\Rank(A_{\rho}|B_{\rho}) & =\Rank\left(\hat{\Theta}^{*}[\Id_{r}\otimes A]\hat{\Theta}\hat{\Theta}^{*}\Big|\hat{\Theta}^{*}[\Id_{r}\otimes B]\hat{\Theta}\hat{\Theta}^{*}\right)\label{eq:rank_eval2}\\
 & =\Rank\left(\hat{\Theta}^{*}\hat{\Theta}\hat{\Theta}^{*}[\Id_{r}\otimes A]\Big|\hat{\Theta}^{*}\hat{\Theta}\hat{\Theta}^{*}[\Id_{r}\otimes B]\right)\nonumber \\
 & =\Rank\left[\hat{\Theta}^{*}\left(\Id_{r}\otimes A\big|\Id_{r}\otimes B\right)\right]\nonumber \\
 & =\Rank(\hat{\Theta}^{*})=2d,\nonumber 
\end{align}
with the last transition possible because the matrix $\left(\Id_{r}\otimes A\big|\Id_{r}\otimes B\right)$
is of full rank. To show the second part of the lemma, we use \ref{lem: properties of Theta},(\ref{enu: Theta property - quotient A B})
and write 
\begin{align*}
A_{\rho}B_{\rho}^{*} & =\hat{\Theta}^{*}[\Id_{r}\otimes A]\hat{\Theta}\hat{\Theta}^{*}[\Id_{r}\otimes B^{*}]\hat{\Theta}\\
 & =\hat{\Theta}^{*}\hat{\Theta}\hat{\Theta}^{*}[\Id_{r}\otimes A][\Id_{r}\otimes B^{*}]\hat{\Theta}\\
 & =\hat{\Theta}^{*}[\Id_{r}\otimes AB^{*}]\hat{\Theta}\\
 & =\hat{\Theta}^{*}[\Id_{r}\otimes BA^{*}]\hat{\Theta}=B_{\rho}A_{\rho}^{*},
\end{align*}
where we used parts (\ref{enu: Theta property - orthogonality}) and
(\ref{enu: Theta property - commuting with A B}) of Lemma \ref{lem: properties of Theta}
and that $BA^{*}$ is self-adjoint by our assumption on the original
operator $\Op$ (see Section \ref{sec:qg_review}).
\end{proof}
Lemma \ref{lem:valid_sa} implies that the operator $\Op_{\rho}$
is self-adjoint (see Section \ref{sec:qg_review} for self-adjointness
conditions). What remains to show to establish Theorem~\ref{thm:commutative_diagram-QG}
is that the following diagram commutes: 
\begin{equation}
\begin{CD}\Dom(\Op_{\rho})@>\cong>>\Hom(\Vr,\Dom(\Op))\\
@V{\Op_{\rho}}VV@VV{\Op}V\\
L_{2}(\Gamma_{\rho})@>\cong>>\Hom(\Vr,L_{2}(\Gamma))
\end{CD}\label{eq:commutative_diagram-QG}
\end{equation}
where the isomorphisms in the diagram are Hilbert space isomorphisms
(i.e. they preserve the inner product).
\begin{proof}[Proof of Theorem \ref{thm:commutative_diagram-QG}]

We will show that the Hilbert-space isomorphism $\Dom(\Op_{\rho})\cong\Hom(\Vr,\Dom(\Op))$
is given by the map\footnote{Here $\vec{^{-1}}$ is understood as the inverse of $\vec{~}{:}~M_{|\E|\times r}(\C)\to\C^{|\E|r}$
as in Definition \ref{def: Vectorization}.} $\vec{^{-1}}\circ\Theta$.

We need to explain what is meant by applying the (numerical) matrix
$\Theta\in M_{\left|\E\right|r\times d}(\C)$ to $f\in\Dom(\Op_{\rho})$,
since the latter is just a list of functions living, a priori, in
different spaces. The informal explanation is that edges belonging
to the same orbit have the same lengths and the functions defined
on them can therefore be formed into linear combinations. More formally,
assume we have two direct sum spaces $\oplus_{j=1}^{n}W_{j}$ and
$\oplus_{k=1}^{m}U_{k}$ and assume we have a set of isomorphisms
among the individual spaces $\{W_{j},U_{k}\}$. The set of isomorphisms
is assumed to be closed by transitivity. For a matrix $Q\in M_{n\times m}$
the action between $\oplus_{j=1}^{n}W_{j}$ and $\oplus_{k=1}^{m}U_{k}$
is well-defined only if $W_{j}\cong U_{k}$ for all $Q_{jk}\neq0$.
Then we can say 
\begin{equation}
Q\begin{pmatrix}w_{1}\\
\vdots\\
w_{n}
\end{pmatrix}=\begin{pmatrix}u_{1}\\
\vdots\\
u_{m}
\end{pmatrix},\label{eq:action_T_def}
\end{equation}
if $u_{k}$ are defined by applying the rules of matrix multiplication,
\begin{equation}
u_{k}=\sum_{j=1}^{n}\ii\left(Q_{k,j}w_{j}\right)=\sum_{j=1}^{n}Q_{k,j}\ii\left(v_{j}\right)\label{eq:matrix_mult_def}
\end{equation}
where $\ii:W_{j}\to U_{K}$ is some appropriate linear isomorphism.
For example, if $W_{j}\cong H^{2}([0,l])\cong U_{k}$ correspond to
an edge in the full graph and quotient graph respectively then we
have simply $w=\ii(v)=v$. The sum is taken over all $j$ such that
$Q_{k,j}\neq0$ (alternatively, we can allow $\ii$ to map 0 to 0
for any two spaces). It is easy to see that this action is associative,
i.e. if we have the action of $Q$ from $\oplus_{j=1}^{n}W_{j}$ to
$\oplus_{k=1}^{m}U_{k}$ and the action of $P$ from $\oplus_{k=1}^{m}U_{k}$
to $\oplus_{i=1}^{l}V_{l}$, then $P(Qw)=(PQ)w$.\\
 \sloppy We may now complete the argument and show that $\Theta f$
is well defined, for $f=(f_{e_{1,1}},\ldots,f_{e_{1,d_{1}}},f_{e_{2,1}},\ldots f_{e_{2,d_{2}}},\ldots)^{T}\in\Dom(\Op_{\rho})$.
From the construction of $\Theta$ using (\ref{eq:definition_of_abstract_thetas - to use for Theta matrix})
we can check that the action of $\Theta$ on $f$ is well defined,
i.e. each row of $\Theta f$ is a linear combination of functions
defined on quotient edges of the same length (recall that $l_{e_{i,k}}=l_{e_{i}}$).

We proceed to show that $\phi:=\vec{^{-1}}(\Theta f)\in\Hom(\Vr,\Dom(\Op))$,
i.e. each $\phi^{(j)}\in\Dom(\Op)$ and $\phi$ satisfies (\ref{eq:intertwining_cond_qg}).
To show that each $\phi^{(j)}$ satisfies the vertex conditions of
$\Gamma$ let us define 
\begin{equation}
A\gamma_{D}(\phi)+B\gamma_{N}(\phi):=\left(A\gamma_{D}(\phi^{(1)})+B\gamma_{N}(\phi^{(1)}),\ldots,A\gamma_{D}(\phi^{(r)})+B\gamma_{N}(\phi^{(r)})\right).\label{eq:tensor_VC_unvec}
\end{equation}
Next, we apply $\vec{\ensuremath{}}$ to the left-hand side (see Lemma
\ref{Lem: Vectorization}). Using that $\Theta f=\vec{(}\phi)$ and
denoting $\hat{\Theta}:=\Theta\otimes\Id_{2}$ (similarly to the $\hat{\pi}$
notation), we write 
\begin{align}
\vec{\Big(}A\gamma_{D}(\phi)+B\gamma_{N}(\phi)\Big) & =[\Id_{r}\otimes A]\gamma_{D}(\vec{(}\phi))+[\Id_{r}\otimes B]\gamma_{N}(\vec{(}\phi))\label{eq:tensor_VC}\\
 & =[\Id_{r}\otimes A]\gamma_{D}(\Theta f)+[\Id_{r}\otimes B]\gamma_{N}(\Theta f)\nonumber \\
 & =[\Id_{r}\otimes A]\hat{\Theta}\gamma_{D}(f)+[\Id_{r}\otimes B]\hat{\Theta}\gamma_{N}(f)\nonumber \\
 & =[\Id_{r}\otimes A]\hat{\Theta}\hat{\Theta}^{*}\hat{\Theta}\gamma_{D}(f)+[\Id_{r}\otimes B]\hat{\Theta}\hat{\Theta}^{*}\hat{\Theta}\gamma_{N}(f)\nonumber \\
 & =\hat{\Theta}\left(\hat{\Theta}^{*}[\Id_{r}\otimes A]\hat{\Theta}\gamma_{D}(f)+\hat{\Theta}^{*}[\Id_{r}\otimes B]\hat{\Theta}\gamma_{N}(f)\right)\nonumber \\
 & =\hat{\Theta}\left(A_{\rho}\gamma_{D}(f)+B_{\rho}\gamma_{N}(f)\right),\nonumber 
\end{align}
where in the third line we have used that $\gamma_{D}(\Theta f)=(\Theta\otimes\Id_{2})\gamma_{D}(f)$
and similarly for $\gamma_{N}$ (see also (\ref{eq:trace_transformation})).
In the last three lines we used properties of $\hat{\Theta}$ which
are stated in Lemma \ref{lem: properties of Theta}. In particular,
we use that $\hat{\Theta}^{*}\hat{\Theta}=\Id_{2d}$, that $\hat{\Theta}\hat{\Theta}^{*}$
commutes with $\Id_{r}\otimes A$ and $\Id_{r}\otimes B$ and that
$A_{\rho}=\hat{\Theta}^{*}[\Id_{r}\otimes A]\hat{\Theta}$ and $B_{\rho}=\hat{\Theta}^{*}[\Id_{r}\otimes B]\hat{\Theta}$.

Since $\vec{\ensuremath{}}$ is invertible and $\hat{\Theta}$ is
left-invertible, $A\gamma_{D}(\phi)+B\gamma_{N}(\phi)=0$ if and only
if $A_{\rho}\gamma_{D}(f)+B_{\rho}\gamma_{N}(f)=0$. In other words,
\begin{equation}
\forall j=1,\ldots,r,\quad\phi^{(j)}\in\Dom(\Op)\qquad\Leftrightarrow\qquad f\in\Dom(\Op_{\rho}),\label{eq:two_domains}
\end{equation}
which shows that $\vec{^{-1}}(\Theta f)\in\mathrm{Hom}(\Vr,\Dom(\Op))$.
Next we show that the image of $\vec{^{-1}}\circ\Theta$ is actually
in $\Hom(\Vr,\Dom(\Op))$.

To verify $\phi:=\vec{^{-1}}(\Theta f)$ is an intertwiner (see (\ref{eq:intertwining_cond_qg}))
we calculate 
\begin{align}
\vec{\left(\pi(g)\phi-\phi\rho(g)\right)} & =[\Id_{r}\otimes\pi(g)-\rho(g)^{T}\otimes\Id_{p}]\vec{(}\phi)\label{eq:intertwiner_qg}\\
 & =[\Id_{r}\otimes\pi(g)-\rho(g)^{T}\otimes\Id_{p}]\Theta f\nonumber \\
 & =\left([\Id_{r}\otimes\pi(g)-\rho(g)^{T}\otimes\Id_{p}]\Theta\right)f=0,\nonumber 
\end{align}
where we used associativity of the matrix action on $f$ and Lemma
\ref{lem: properties of Theta},(\ref{enu: Theta property - intertwining}).
We have established that $\vec{^{-1}}\circ\Theta$ maps $\Dom(\Op_{\rho})$
to a subset of $\Hom(\Vr,\Dom(\Op))$. This mapping is obviously linear.
We still need to show it is one-to-one and onto. We do it by furnishing
its inverse.

Let us introduce the mapping $\Theta^{*}\circ\vec{\ensuremath{}}$
acting on $\Hom(\Vr,\Dom(\Op))$. The multiplication by $\Theta^{*}$
can be shown to be well defined in the sense of (\ref{eq:matrix_mult_def}),
i.e. it only takes the linear combinations of the functions defined
on the edges from the same orbit. By (\ref{eq:two_domains}), the
codomain of the mapping is $\Dom(\Op_{\rho})$. Using Lemma \ref{lem: properties of Theta},(\ref{enu: Theta property - orthogonality})
we get that $\Theta^{*}\circ\vec{\ensuremath{}}$ is a left inverse
of $\vec{^{-1}}\circ\Theta$. We now need to show that the product
$\vec{^{-1}}\circ\Theta\Theta^{*}\circ\vec{\ensuremath{}}$ is identity
on $\Hom(\Vr,\Dom(\Op))$. Taking $\phi\in\Hom(\Vr,\Dom(\Op))$ ,
we write 
\begin{equation}
\vec{^{-1}}\Big(\Theta\Theta^{*}\vec{(}\phi)\Big)=\frac{1}{|G|}\sum_{g\in G}\vec{^{-1}}\Big(\left(\cc{\rho}(g)\otimes\pi(g)\right)\vec{(}\phi)\Big),\label{eq:expand_P}
\end{equation}
where we used Lemma \ref{lem: properties of Theta},(\ref{enu: Theta property - projection})
and the linearity of $\vec{^{-1}}$. On the other hand, using the
properties of $\vec{\ensuremath{}}$ from Lemma \ref{Lem: Vectorization}
, we get 
\begin{align}
\vec{^{-1}}\Big(\left(\cc{\rho}(g)\otimes\pi(g)\right)\vec{(}\phi)\Big) & =\vec{^{-1}}\Big((\Id_{r}\otimes\pi(g))(\cc{\rho}(g)\otimes\Id_{p})\vec{(}\phi)\Big)\label{eq:alt_proof_Pinvariant}\\
 & =\pi(g)\vec{^{-1}}\Big((\cc{\rho}(g)\otimes\Id_{p})\vec{(}\phi)\Big)\nonumber \\
 & =\pi(g)\phi\rho^{-1}(g)=\phi,\nonumber 
\end{align}
since $\phi\in\Hom(\Vr,\Dom(\Op))$ satisfies the intertwining condition
\eqref{eq:intertwining_cond_qg}. Substituting into \eqref{eq:expand_P}
we obtain $\vec{^{-1}}\big(\P\vec{(}\phi)\big)=\phi$.

To summarize, we have established that $\vec{^{-1}}\circ\Theta$ gives
the isomorphism $\Dom(\Op_{\rho})\cong\Hom(\Vr,\Dom(\Op))$. Additionally,
it is a Hilbert space isomorphism since for $\phi=\vec{^{-1}}(\Theta f)$
and $\psi=\vec{^{-1}}(\Theta g)$ from $\Hom(\Vr,\Dom(\Op))$ we have
\begin{multline*}
\langle\phi,\psi\rangle_{\Hom(\Vr,\Dom(\Op))}=\Tr\left(\phi^{*}\psi\right)=\vec{(}\phi)^{*}\vec{(}\psi)\\
=(\Theta f)^{*}(\Theta g)=f^{*}\Theta^{*}\Theta g=f^{*}g=\langle f,g\rangle:=\sum_{e\in\E}\langle f_{e},g_{e}\rangle_{L^{2}([0,l_{e}])},
\end{multline*}
where during linear algebra operations in the middle we have suppressed
the scalar product meaning of multiplying the entries of $f,g$ and
$\phi,\psi$. An identical proof (without checking the vertex conditions)
works to establish the isomorphism between $L_{2}(\Gamma_{\rho})$
and $\Hom(\Vr,L_{2}(\Gamma))$ (the bottom arrow in the diagram (\ref{eq:commutative_diagram-QG})).\\
 Finally, we deal with the vertical arrows of the diagram (\ref{eq:commutative_diagram-QG}).
We remember that $\Op_{\rho}$ and $\Op$ act as a linear differential
expression separately applied to each entry. They commute with taking
(the allowed) linear combinations imposed by $\Theta$, giving us
$\Theta\Op_{\rho}f=[I_{r}\otimes\Op]\Theta f$ for all $f\in\Dom(\Op_{\rho})$.
Applying $\vec{^{-1}}$ to both sides and using Lemma \ref{Lem: Vectorization}
gives 
\begin{equation}
\vec{^{-1}}(\Theta\Op_{\rho}f)=\Op\,\vec{^{-1}}(\Theta f),
\end{equation}
which shows the commutativity of the diagram.
\end{proof}

\subsection{Proof of scattering quotient}

\label{app:qg_quotient_scat}
\begin{proof}[Proof of Proposition~\ref{prop:scat_mat_fact}]
First, the matrix $S(k)$ is $G$-symmetric, as a consequence of
(\ref{eq:Vertex_Scattering_Matrix}). To see this, note that by assumption,
$A$ and $B$ are $G$-symmetric and also that an invertible matrix
is $G$-symmetric if and only if its inverse is also $G$-symmetric.
By (\ref{eq:Vertex_Scattering_Matrix}) the scattering matrix of $\Gamma_{\rho}$
is given by 
\[
-(A_{\rho}+ikB_{\rho})^{-1}(A_{\rho}-ikB_{\rho})\left(\Id_{d}\otimes\left(\begin{smallmatrix}0 & 1\\
1 & 0
\end{smallmatrix}\right)\right).
\]
Plugging the expressions for $A_{\rho},B_{\rho}$ from Lemma \ref{lem: properties of Theta},(\ref{enu: Theta property - quotient A B}),
and rearranging, we get

\begin{align*}
- & \left[\hat{\Theta}^{*}\left(\Id_{r}\otimes({A}+ik{B})\right)\hat{\Theta}\right]^{-1}\left[\hat{\Theta}^{*}\left(\Id_{r}\otimes({A}-ik{B})\right)\hat{\Theta}\right]\left(\Id_{d}\otimes\left(\begin{smallmatrix}0 & 1\\
1 & 0
\end{smallmatrix}\right)\right)\\
=- & \left[\hat{\Theta}^{*}\left(\Id_{r}\otimes({A}+ik{B})\right)^{-1}\hat{\Theta}\right]\left[\hat{\Theta}^{*}\left(\Id_{r}\otimes({A}-ik{B})\right)\hat{\Theta}\right]\left(\Id_{d}\otimes\left(\begin{smallmatrix}0 & 1\\
1 & 0
\end{smallmatrix}\right)\right)\\
=- & \hat{\Theta}^{*}\left(\Id_{r}\otimes({A}+ik{B})^{-1}\right)\left(\Id_{r}\otimes({A}-ik{B})\right)\hat{\Theta}\left(\Id_{d}\otimes\left(\begin{smallmatrix}0 & 1\\
1 & 0
\end{smallmatrix}\right)\right)\\
=- & \hat{\Theta}^{*}\left(\Id_{r}\otimes({A}+ik{B})^{-1}({A}-ik{B})\right)\hat{\Theta}\left(\Id_{d}\otimes\left(\begin{smallmatrix}0 & 1\\
1 & 0
\end{smallmatrix}\right)\right)\\
=- & \hat{\Theta}^{*}\left[\Id_{r}\otimes({A}+ik{B})^{-1}({A}-ik{B})\left(\Id_{|\E|}\otimes\left(\begin{smallmatrix}0 & 1\\
1 & 0
\end{smallmatrix}\right)\right)\right]\hat{\Theta}\\
=\phantom{-} & \hat{\Theta}^{*}(\Id_{r}\otimes S(k))\hat{\Theta}=S_{\rho}(k),
\end{align*}

where we used that $\hat{\Theta}\hat{\Theta}^{*}$ commutes with $\Id_{r}\otimes{A}$
and with $\Id_{r}\otimes{B}$ and that $\hat{\Theta}^{*}\hat{\Theta}=\Id_{2d}$
(see Lemma \ref{lem: properties of Theta}). The fifth line of the
calculation is obtained from 
\[
\hat{\Theta}=\Theta\otimes\Id_{2},
\]
\[
\hat{\Theta}\left(\Id_{d}\otimes\left(\begin{smallmatrix}0 & 1\\
1 & 0
\end{smallmatrix}\right)\right)=\Theta\otimes\left(\begin{smallmatrix}0 & 1\\
1 & 0
\end{smallmatrix}\right)=\left(\Id_{r|\E|}\otimes\left(\begin{smallmatrix}0 & 1\\
1 & 0
\end{smallmatrix}\right)\right)\hat{\Theta}.
\]

The last equality $\hat{\Theta}^{*}(\Id_{r}\otimes S(k))\hat{\Theta}=S_{\rho}(k)$
follows from the same argument as in the proof of Lemma \ref{lem: properties of Theta},(\ref{enu: Theta property - quotient A B}).

Next, we turn to prove a similar statement for the unitary evolution
operator. Conditions in Definition \ref{def:QG_pi_symmetric} guarantee
that the edge length matrix of $\Gamma$, $\hat{L}=L\otimes\Id_{2}$
is $G$-symmetric. This, together with the $G$-symmetry of $S(k)$
yields the $G$-symmetry of $U(k)$.

We now show that the unitary evolution matrix of the quotient graph
is the quotient of $U(k)$. We form a diagonal $d\times d$ matrix,
$L_{\rho}:=\diag(\{l_{e_{i,j}}\}_{i\in\D,j=1,\ldots,d_{i}})$, which
stores the edge lengths of the quotient graph $\Gamma_{\rho}$. This
matrix is related to the edge length matrix of the original graph,
$\Gamma$ by 
\begin{equation}
L_{\rho}=\Theta^{*}[\Id_{r}\otimes L]\Theta,\label{eq:length_matrix_relation}
\end{equation}
which is a consequence of the locality and orthonormality of $\Theta$'s
columns, as provided by Lemma \ref{lem: properties of Theta}. Explicitly,
a $\Theta$ column which corresponds to an edge $e_{i,j}$ may have
non-vanishing entries only corresponding to edges which are in the
orbit of $e_{i}$ and all such edges are of the same length $l_{e_{i}}$.
We may now write the unitary evolution operator of $\Gamma_{\rho}$
as 
\begin{align*}
 & \phantom{=}\ue^{\ui k\widehat{L_{\rho}}}S_{\rho}(k)\\
 & =\ue^{\ui k\widehat{L_{\rho}}}\hat{\Theta}^{*}(\Id_{r}\otimes S(k))\hat{\Theta}\\
 & =\hat{\Theta}^{*}[\Id_{r}\otimes\ue^{\ui k\widehat{L}}]\hat{\Theta}\hat{\Theta}^{*}(\Id_{r}\otimes S(k))\hat{\Theta}\\
 & =\hat{\Theta}^{*}(\Id_{r}\otimes\ue^{\ui k\widehat{L}}S(k))\hat{\Theta}\\
 & =\hat{\Theta}^{*}(\Id_{r}\otimes U(k))\hat{\Theta}=U_{\rho}(k),
\end{align*}
where the third line is obtained from (\ref{eq:length_matrix_relation})
and using that $\hat{\Theta}\hat{\Theta}^{*}$ commutes with $\Id_{r}\otimes\hat{L}$.
This commutation together with $\hat{\Theta}^{*}\hat{\Theta}=\Id_{2d}$
gives the fourth line. The final equality, $\hat{\Theta}^{*}(\Id_{r}\otimes U(k))\hat{\Theta}=U_{\rho}(k)$,
follows from the same argument as in the proof of Lemma \ref{lem: properties of Theta},(\ref{enu: Theta property - quotient A B}).
\end{proof}
%


\bibliographystyle{myalpha}
\bibliography{ref}

\newcommand{\etalchar}[1]{$^{#1}$}
\providecommand{\bysame}{\leavevmode\hbox to3em{\hrulefill}\thinspace}
\providecommand{\MR}{\relax\ifhmode\unskip\space\fi MR }
\providecommand{\MRhref}[2]{%
  \href{http://www.ams.org/mathscinet-getitem?mr=#1}{#2}
}
\providecommand{\href}[2]{#2}
\begin{thebibliography}{DGMdO{\etalchar{+}}19}

\bibitem[BBS12]{BanBerSmi_ahp12}
R.~Band, G.~Berkolaiko, and U.~Smilansky, \emph{Dynamics of nodal points and
  the nodal count on a family of quantum graphs}, Annales Henri Poincare
  \textbf{13} (2012), no.~1, 145--184.

\bibitem[BC18]{BerCom_jst18}
G.~Berkolaiko and A.~Comech, \emph{Symmetry and {D}irac points in graphene
  spectrum}, J. Spectr. Theory \textbf{8} (2018), no.~3, 1099--1147, preprint
  {\tt arXiv:1412.8096}.

\bibitem[BE22]{BerEtt_sam22}
G.~Berkolaiko and M.~Ettehad, \emph{Three-dimensional elastic beam frames:
  rigid joint conditions in variational and differential formulation}, Stud.
  Appl. Math. \textbf{148} (2022), no.~4, 1586--1623.

\bibitem[Ber17]{Ber_crm17}
G.~Berkolaiko, \emph{An elementary introduction to quantum graphs}, Geometric
  and computational spectral theory, Contemp. Math., vol. 700, Amer. Math.
  Soc., Providence, RI, 2017, pp.~41--72. \MR{3748521}

\bibitem[BFW17]{BarFraWeb_laa17}
W.~Barrett, A.~Francis, and B.~Webb, \emph{Equitable decompositions of graphs
  with symmetries}, Linear Algebra and its Applications \textbf{513} (2017),
  409--434.

\bibitem[BG18]{MR3880377}
R.~Band and S.~Gnutzmann, \emph{Quantum graphs via exercises}, Spectral theory
  and applications, Contemp. Math., vol. 720, Amer. Math. Soc., [Providence],
  RI, [2018] \copyright 2018, preprint {\tt arXiv:1711.07435}, pp.~187--203.

\bibitem[BK13a]{BerKuc_graphs}
G.~Berkolaiko and P.~Kuchment, \emph{Introduction to quantum graphs},
  Mathematical Surveys and Monographs, vol. 186, AMS, 2013.

\bibitem[BK13b]{BreKel_opem13}
J.~Breuer and M.~Keller, \emph{Spectral analysis of certain spherically
  homogeneous graphs}, Oper. Matrices \textbf{7} (2013), no.~4, 825--847.
  \MR{3154573}

\bibitem[BK22]{BreKel_opem13-err}
J.~Breuer and M.~Keller, \emph{Erratum to {\it {s}pectral analysis of certain
  spherically homogeneous graphs}, {O}perators and {M}atrices, 7 (2013),
  825--847}, Oper. Matrices \textbf{16} (2022), no.~4, 1041--1044. \MR{4543377}

\bibitem[BL18]{BerLiu_lmp18}
G.~Berkolaiko and W.~Liu, \emph{Eigenspaces of symmetric graphs are not
  typically irreducible}, Lett. Math. Phys. \textbf{Online First} (2018),
  10.1007/s11005--018--1050--7, preprint {\tt arXiv:1705.01653}.

\bibitem[BL20]{BreLev_ahp20}
J.~Breuer and N.~Levi, \emph{On the decomposition of the {L}aplacian on metric
  graphs}, Ann. Henri Poincar\'{e} \textbf{21} (2020), no.~2, 499--537.
  \MR{4056276}

\bibitem[BPBS09]{Ban09}
R.~Band, O.~Parzanchevski, and G.~Ben-Shach, \emph{The isospectral fruits of
  representation theory: quantum graphs and drums}, J. Phys. A: Math. Theor.
  \textbf{42} (2009), 175202.

\bibitem[Bro99a]{Brooks_aif99}
R.~Brooks, \emph{Non-{S}unada graphs}, Annales de l'institut Fourier
  \textbf{49} (1999), no.~2, 707--725 (eng).

\bibitem[Bro99b]{Brooks_cm99}
R.~Brooks, \emph{The {S}unada method}, Tel {A}viv {T}opology {C}onference:
  {R}othenberg {F}estschrift (1998), Contemp. Math., vol. 231, Amer. Math.
  Soc., Providence, RI, 1999, pp.~25--35. \MR{1705572}

\bibitem[BSS06]{BanShaSmi_jpa06}
R.~Band, T.~Shapira, and U.~Smilansky, \emph{Nodal domains on isospectral
  quantum graphs: the resolution of isospectrality?}, J. Phys. A \textbf{39}
  (2006), no.~45, 13999--14014. \MR{2277370 (2008i:58032)}

\bibitem[BSS10]{BanSawSmi_jpa10}
R.~Band, A.~Sawicki, and U.~Smilansky, \emph{Scattering from isospectral
  quantum graphs}, J. Phys. A \textbf{43} (2010), no.~41, 415201, 17.
  \MR{2726689 (2012a:81106)}

\bibitem[Bum13]{Bump_LieGroups_book}
D.~Bump, \emph{Lie groups}, second ed., Graduate Texts in Mathematics, vol.
  225, Springer, New York, 2013. \MR{3136522}

\bibitem[CDS95]{CveDooSac_spectra_of_graphs_book}
D.~s.~M. Cvetkovi\'c, M.~Doob, and H.~Sachs, \emph{Spectra of graphs}, third
  ed., Johann Ambrosius Barth, Heidelberg, 1995, Theory and applications.
  \MR{1324340}

\bibitem[Chu97]{Chung_spectralgraph}
F.~R.~K. Chung, \emph{Spectral graph theory}, CBMS Regional Conference Series
  in Mathematics, vol.~92, Published for the Conference Board of the
  Mathematical Sciences, Washington, DC, 1997. \MR{MR1421568 (97k:58183)}

\bibitem[CP20]{ChePiv_ieot20}
A.~Chernyshenko and V.~Pivovarchik, \emph{Recovering the shape of a quantum
  graph}, Integral Equations Operator Theory \textbf{92} (2020), no.~3, Paper
  No. 23, 17. \MR{4109187}

\bibitem[CRS10]{CveRowSim_An_introduction_to_the_theory_of_graph_spectra_2010}
D.~M. Cvetkovi{\'c}, P.~Rowlinson, and S.~Simi{\'c}, \emph{An introduction to
  the theory of graph spectra}, vol.~75, Cambridge University Press Cambridge,
  2010.

\bibitem[CS92]{Chu92}
F.~R.~K. Chung and S.~Sternberg, \emph{Laplacian and vibrational spectra for
  homogeneous graphs}, Journal of Graph Theory \textbf{16} (1992), no.~6,
  605--627.

\bibitem[DF04]{DummitFoote_abstract_algebra}
D.~S. Dummit and R.~M. Foote, \emph{Abstract algebra}, third ed., John Wiley \&
  Sons, Inc., Hoboken, NJ, 2004. \MR{2286236}

\bibitem[DGMdO{\etalchar{+}}19]{DalGavMonOchStaSte_involve19}
K.~Daly, C.~Gavin, G.~Montes~de Oca, D.~Ochoa, E.~Stanhope, and S.~Stewart,
  \emph{Orbigraphs: a graph-theoretic analog to riemannian orbifolds}, Involve,
  a Journal of Mathematics \textbf{12} (2019), no.~5, 721--736.

\bibitem[EL19]{ExnLip_jmp19}
P.~Exner and J.~Lipovsk\'{y}, \emph{Spectral asymptotics of the {L}aplacian on
  {P}latonic solids graphs}, J. Math. Phys. \textbf{60} (2019), no.~12, 122101,
  21. \MR{4043812}

\bibitem[Exn21]{Exner_ppn21}
P.~Exner, \emph{Quantum graphs with vertices violating the time reversal
  symmetry}, Physics of Particles and Nuclei \textbf{52} (2021), 330--336.

\bibitem[FCLP23]{SteLlePos_arXiv22}
J.~S. Fabila-Carrasco, F.~Lled\'{o}, and O.~Post, \emph{A geometric
  construction of isospectral magnetic graphs}, Analysis and Mathematical
  Physics \textbf{31} (2023), no.~64.

\bibitem[FH13]{Fulton-2013}
W.~Fulton and J.~Harris, \emph{Representation theory: a first course}, vol.
  129, Springer Science \& Business Media, 2013.

\bibitem[FSSW17]{FraSmiSorWeb_laa17}
A.~Francis, D.~Smith, D.~Sorensen, and B.~Webb, \emph{Extensions and
  applications of equitable decompositions for graphs with symmetries}, Linear
  Algebra and its Applications \textbf{532} (2017), 432--462.

\bibitem[FSW19]{FraSmiWeb_laa19}
A.~Francis, D.~Smith, and B.~Webb, \emph{General equitable decompositions for
  graphs with symmetries}, Linear Algebra and its Applications \textbf{577}
  (2019), 287--316.

\bibitem[God93]{Godsil_algebraic_combinaotrics_1993}
C.~D. Godsil, \emph{Algebraic combinatorics}, Chapman and Hall Mathematics
  Series, Chapman \& Hall, New York, 1993. \MR{1220704}

\bibitem[GR01]{GodRoy_algebraic_graph_theory_book_2001}
C.~Godsil and G.~F. Royle, \emph{Algebraic graph theory}, vol. 207, Springer
  Science \& Business Media, 2001.

\bibitem[GS01]{GutSmi_jpa01}
B.~Gutkin and U.~Smilansky, \emph{Can one hear the shape of a graph?}, J. Phys.
  A \textbf{34} (2001), no.~31, 6061--6068. \MR{MR1862642 (2002k:05205)}

\bibitem[GS06]{GnuSmi_ap06}
S.~Gnutzmann and U.~Smilansky, \emph{Quantum graphs: Applications to quantum
  chaos and universal spectral statistics}, Adv. Phys. \textbf{55} (2006),
  no.~5--6, 527--625.

\bibitem[GWW92a]{Gor92b}
C.~Gordon, D.~L. Webb, and S.~Wolpert, \emph{Isospectral plane domains and
  surfaces via riemannian orbifolds}, Inventiones mathematicae \textbf{110}
  (1992), no.~1, 1--22.

\bibitem[GWW92b]{Gor92a}
C.~Gordon, D.~L. Webb, and S.~Wolpert, \emph{One cannot hear the shape of a
  drum}, Bull. Amer. Math. Soc. \textbf{27} (1992), 134--138.

\bibitem[Haa10]{Haake-2010}
F.~Haake, \emph{Quantum signatures of chaos}, enlarged ed., Springer Series in
  Synergetics, Springer-Verlag, Berlin, 2010, With a foreword by H. Haken.
  \MR{2604131}

\bibitem[HH99]{Hal99}
L.~Halbeisen and N.~Hungerb\"{u}hler, \emph{Generation of isospectral graphs},
  Journal of Graph Theory \textbf{31} (1999), no.~3, 255--265.

\bibitem[JL21]{JevLip_polonica21}
V.~Je\v{z}ek and J.~Lipovsk\'{y}, \emph{Application of quotient graph theory to
  three-edge star graphs}, Acta Physica Polonica A \textbf{140} (2021),
  514--524.

\bibitem[JMS12]{Joyner-2012}
C.~H. Joyner, S.~M\"{u}ller, and M.~Sieber, \emph{Semiclassical approach to
  discrete symmetries in quantum chaos}, Journal of Physics A: Mathematical and
  Theoretical \textbf{45} (2012), no.~20, 205102.

\bibitem[JMS14]{Joy14}
C.~H. Joyner, S.~M\"{u}ller, and M.~Sieber, \emph{{GSE} statistics without
  spin}, EPL \textbf{107} (2014), no.~5, 50004.

\bibitem[Kac66]{Kac66}
M.~Kac, \emph{Can one hear the shape of a drum?}, Am. Math. Mon. \textbf{73}
  (1966), 1--23.

\bibitem[KM21]{KurMul_arXiv21}
P.~Kurasov and J.~Muller, \emph{On isospectral metric graphs}, 2021.

\bibitem[KN21]{KosNic_jst21}
A.~Kostenko and N.~Nicolussi, \emph{Quantum graphs on radially symmetric
  antitrees}, J. Spectr. Theory \textbf{11} (2021), no.~2, 411--460.
  \MR{4293483}

\bibitem[KR97]{Keating-1997}
J.~P. Keating and J.~M. Robbins, \emph{Discrete symmetries and spectral
  statistics}, Journal of Physics A: Mathematical and General \textbf{30}
  (1997), no.~7, L177.

\bibitem[KS97]{KotSmi_prl97}
T.~Kottos and U.~Smilansky, \emph{Quantum chaos on graphs}, Phys. Rev. Lett.
  \textbf{79} (1997), no.~24, 4794--4797.

\bibitem[KS99]{Kostrykin-1999}
V.~Kostrykin and R.~Schrader, \emph{Kirchhoff's rule for quantum wires},
  Journal of Physics A: Mathematical and General \textbf{32} (1999), no.~4,
  595.

\bibitem[KS03]{KotSmi_jpa03}
T.~Kottos and U.~Smilansky, \emph{Quantum graphs: a simple model for chaotic
  scattering}, J. Phys. A \textbf{36} (2003), no.~12, 3501--3524, Random matrix
  theory. \MR{MR1986432 (2004g:81064)}

\bibitem[Kuc16]{Kuc_bams16}
P.~Kuchment, \emph{An overview of periodic elliptic operators}, Bull. Amer.
  Math. Soc. (N.S.) \textbf{53} (2016), no.~3, 343--414.

\bibitem[Lau91]{Lauritzen-1991}
B.~Lauritzen, \emph{Discrete symmetries and the periodic-orbit expansions},
  Phys. Rev. A \textbf{43} (1991), 603--606.

\bibitem[LSBS21]{LawSawBiaSir_Scirep21}
M.~Lawniczak, A.~Sawicki, M.~Bialous, and L.~Sirko, \emph{Isoscattering strings
  of concatenating graphs and networks}, Scientific Reports \textbf{11} (2021),
  1575.

\bibitem[MP23]{MugPiv_jpa23}
D.~Mugnolo and V.~Pivovarchik, \emph{Distinguishing cospectral quantum graphs
  by scattering}, J. Phys. A \textbf{56} (2023), no.~9, Paper No. 095201, 18.
  \MR{4555083}

\bibitem[Mug14]{Mugnolo_book}
D.~Mugnolo, \emph{Semigroup methods for evolution equations on networks},
  Understanding Complex Systems, Springer, Cham, 2014. \MR{3243602}

\bibitem[Mut21]{Mutlu_cpaa21}
G.~Mutlu, \emph{On the quotient quantum graph with respect to the regular
  representation}, Commun. Pure Appl. Anal. \textbf{20} (2021), no.~2,
  885--902. \MR{4214048}

\bibitem[NS01]{NaiSol_rjmp01}
K.~Naimark and A.~Solomyak, \emph{Geometry of sobolev spaces on regular trees
  and the hardy inequalities}, Russian Journal of Mathematical Physics
  \textbf{8} (2001), no.~3, 14 (eng).

\bibitem[PB10]{Par10}
O.~Parzanchevski and R.~Band, \emph{Linear representations and isospectrality
  with boundary conditions}, J. Geom. Anal. \textbf{20} (2010), 439Ð471.

\bibitem[Pes94]{Pesce_cm94}
H.~Pesce, \emph{Vari\'et\'es isospectrales et repr\'esentations de groupes},
  Geometry of the spectrum ({S}eattle, {WA}, 1993), Contemp. Math., vol. 173,
  Amer. Math. Soc., Providence, RI, 1994, pp.~231--240. \MR{1298208}

\bibitem[Pis23]{Pistol_arXiv23}
M.-E. Pistol, \emph{Generating isospectral but not isomorphic quantum graphs},
  2023.

\bibitem[Pos12]{Post_book12}
O.~Post, \emph{Spectral analysis on graph-like spaces}, Lecture Notes in
  Mathematics, vol. 2039, Springer Verlag, Berlin, 2012.

\bibitem[PS82]{PowSul_laa82}
D.~L. Powers and M.~M. Sulaiman, \emph{The walk partition and colorations of a
  graph}, Linear Algebra and its Applications \textbf{48} (1982), 145--159.

\bibitem[RAJ{\etalchar{+}}16]{Rehemanjiang-2016}
A.~Rehemanjiang, M.~Allgaier, C.~H. Joyner, S.~M\"{u}ller, M.~Sieber, U.~Kuhl,
  and H.-J. St\"{o}ckmann, \emph{Microwave realization of the gaussian
  symplectic ensemble}, Phys. Rev. Lett. \textbf{117} (2016), 064101.

\bibitem[Rob89]{Robbins-1989}
J.~M. Robbins, \emph{Discrete symmetries in periodic-orbit theory}, Phys. Rev.
  A \textbf{40} (1989), 2128--2136.

\bibitem[RS78]{ReedSimon_v4}
M.~Reed and B.~Simon, \emph{Methods of modern mathematical physics. {IV}.
  {A}nalysis of operators}, Academic Press, New York, 1978.

\bibitem[Sch27]{Schur-1927}
I.~Schur, \emph{\"{U}ber die rationalen {D}arstellungen der allgemeinen
  linearen {G}ruppe}, Sitzungsber. Akad. Wiss., Phys.-Math. Kl. (1927), 58--75.

\bibitem[Sch74]{Schwenk_proc74}
A.~J. Schwenk, \emph{Computing the characteristic polynomial of a graph},
  Graphs and Combinatorics (Berlin, Heidelberg) (R.~A. Bari and F.~Harary,
  eds.), Springer Berlin Heidelberg, 1974, pp.~153--172.

\bibitem[Ser77]{Serre_linear_representations}
J.-P. Serre, \emph{Linear representations of finite groups}, Graduate Texts in
  Mathematics, Vol. 42, Springer-Verlag, New York-Heidelberg, 1977, Translated
  from the second French edition by Leonard L. Scott. \MR{0450380}

\bibitem[Sol04]{Solomyak_wrm04}
M.~Solomyak, \emph{On the spectrum of the {L}aplacian on regular metric trees},
  Waves Random Media \textbf{14} (2004), no.~1, S155--S171. \MR{2046943}

\bibitem[Sun85]{Sun85}
T.~Sunada, \emph{Riemannian coverings and isospectral manifolds}, Annals of
  Mathematics \textbf{121} (1985), no.~1, 169--186.

\bibitem[vB85]{Bel_laa85}
J.~von Below, \emph{A characteristic equation associated to an eigenvalue
  problem on {$c^2$}-networks}, Linear Algebra Appl. \textbf{71} (1985),
  309--325.

\bibitem[Wey46]{Weyl-1946}
H.~Weyl, \emph{The classical groups, their invariants and representations},
  Princeton University Press, 1946.

\bibitem[Zel90]{Zel_aif90}
S.~Zelditch, \emph{On the generic spectrum of a {R}iemannian cover}, Ann. Inst.
  Fourier (Grenoble) \textbf{40} (1990), no.~2, 407--442.

\end{thebibliography}

\end{document}